\documentclass[11pt]{article}
\usepackage{fullpage}
\usepackage{url}
\usepackage[colorlinks=true,linkcolor=Emerald,citecolor=RoyalBlue]{hyperref}
\usepackage{amsmath,nccmath, amsfonts,amsthm,amssymb,multirow,mdframed}

\usepackage{times}
\usepackage{paralist}
\usepackage{graphicx}
\usepackage{floatpag}
\usepackage{bbold}
\usepackage{enumitem}
\usepackage{subcaption} 
\usepackage{soul}
\usepackage{comment}
\usepackage{calc}
\usepackage{upgreek}
\usepackage{mathtools}
\usepackage{color}
\usepackage[dvipsnames]{xcolor}
\usepackage{xspace}
\usepackage{tcolorbox}
\allowdisplaybreaks

\DeclareMathOperator*{\E}{\mathbb{E}}

\DeclarePairedDelimiter\ip{\langle}{\rangle}
\newcommand\skipi{{\vskip 10pt}}
\def \F {{\mathbb F}}

\def \N {\mathbb{N}}

\def \cC {\mathcal{C}}
\def \cD {\mathcal{D}}

\def \cF {\mathcal{F}}
\def \cK {\mathcal{K}}
\def \cL {\mathcal{L}}

\def \cR {\mathcal{R}}

\def \cT {\mathcal{T}}

\def \poly {\mathrm{poly}}

\def \eps {{\varepsilon}}

\def \val {\mathrm{val}}

\def \dist {\Delta}
\def \maj {\text{Maj}}
\def \agr {\text{agr}}
\def \Agr {\text{Agr}}

\def\ggg{\gtrsim}
\def\lll{\lesssim}

\renewcommand{\leq}{\leqslant}

\renewcommand{\geq}{\geqslant}

\newtheorem{theorem}{Theorem}[section]

\newtheorem{corollary}[theorem]{Corollary}

\newtheorem{lemma}[theorem]{Lemma}
\newtheorem*{lemma*}{Lemma}
\newtheorem*{theorem*}{Theorem}
\newtheorem{claim}[theorem]{Claim}
\newtheorem*{claim*}{Claim}

\newtheorem{remark}[theorem]{Remark}
\newtheorem{definition}[theorem]{Definition}
\theoremstyle{definition}
\newtheorem{algorithm}{Algorithm}
\newtheorem{agr-test}{Agreement-Test}
\newtheorem{list-agr-test}{List-Agreement-Test}

\makeatletter
\let\c@fconjecture\c@conjecture
\makeatother

\makeatletter
\let\c@fconj\c@conj
\makeatother

\title{Characterizing Direct Product Testing via Coboundary Expansion}

\author{Mitali Bafna \thanks{Carnegie Mellon University. Supported in part by the Computer Science Department, CMU and a gift from CYLAB, CMU.}
 \and
 	Dor Minzer\thanks{Department of Mathematics, Massachusetts Institute of Technology. Supported by a Sloan Research Fellowship, NSF CCF award 2227876 and NSF CAREER award 2239160.}}
\date{\vspace{-5ex}}
\begin{document}
\maketitle
\begin{abstract}
A $d$-dimensional simplicial complex $X$ is said to support a direct product tester if any locally consistent function defined on its $k$-faces (where $k\ll d$) necessarily come from a function over its vertices. More precisely, a direct product tester has a distribution $\mu$ over pairs of $k$-faces $(A,A')$, and given query access to $F\colon X(k)\to\{0,1\}^k$ it samples $(A,A')\sim \mu$ and
checks that $F[A]|_{A\cap A'} = F[A']|_{A\cap A'}$. The tester should have (1) the ``completeness property'', meaning that any assignment $F$ which is a direct product assignment passes the test with
probability $1$, and (2) the ``soundness property'', meaning that if $F$ passes the test with probability $s$, then $F$ must be correlated with a direct product function. 

Dinur and Kaufman showed that a sufficiently good 
spectral expanding complex $X$ admits a direct product tester in the ``high soundness'' regime where $s$ is close to $1$. They asked whether there are high dimensional expanders that support direct product tests in the ``low soundness'', when $s$ is close to $0$.

We give a characterization of high-dimensional expanders that support a direct product tester in the low soundness regime. We show that spectral expansion is insufficient, and the complex must 
additionally satisfy a variant of coboundary expansion, which we refer to as \emph{Unique-Games coboundary expanders}. Conversely, we show that this property is also sufficient to get direct product testers. This property can be seen as a high-dimensional generalization of the standard notion of coboundary expansion over non-Abelian groups for 2-dimensional complexes. It asserts that any locally consistent Unique-Games instance obtained using the low-level faces of the complex, must admit a good global solution.
\end{abstract}

\section{Introduction}
The problem of testing direct product 
functions lies at the intersection of 
many areas within theoretical computer 
science, such as error correcting codes, 
probabilistically checkable proofs (PCPs), 
hardness amplification and property
testing. In its purest form, one wishes 
to encode a function $f\colon [n]\to \{0,1\}$ using local views in a 
way that admits local testability/ local correction. More precisely, given a parameter $1\leq k<n$, the encoding of
$f$ using subsets of size $k$ can be 
viewed as $F\colon \binom{[n]}{k}\to \{0,1\}^k$ that to each subset $A\subseteq [n]$ of size $k$ assigns a vector of length
$k$ describing the restriction of $f$ to $A$.\footnote{To be more specific, one
identifies $A = \{a_1,\ldots,a_k\}$ with 
the ordered tuple $(a_1,\ldots,a_k)$ where 
$a_1<a_2<\ldots<a_k$ and defines 
$F[A] = (f(a_1),\ldots,f(a_k))$.} 
We refer to this encoding as 
the direct product encoding of $f$ according to the Johnson graph (for reasons that will become apparent shortly).
The obvious downside of this encoding 
scheme is, of course, that its length is much larger than the
description of $f$ (roughly $n^k$ vs $\Theta(n)$). However, as this 
encoding contains many redundancies, one hopes that 
it more robustly stores
the information in the function $f$, 
thereby being more resilient against 
corruptions.

\subsection{Direct Product Testing with $2$ Queries}
Indeed, one of the primary benefits of the above direct product 
encoding is that it admits 
local testers using a few queries. 
These testing algorithms also go 
by the name ``agreements testers'' or ``direct
product testers'', and are often very natural to design. 
A direct product tester for the above 
encoding, which we parameterize by a 
natural number $1\leq t\leq k$ and 
denote by $\mathcal{T}_{t}$, proceeds as follows:
\begin{enumerate}
    \item Choose two subsets $A,A'\subseteq [n]$ uniformly 
    at random conditioned on $|A\cap A'|=t$.
    \item Query $F[A]$, $F[A']$
    and check that $F[A]$ and $F[A']$
    agree on $A\cap A'$.
\end{enumerate}
These type of testers have been first considered and used 
by Goldreich and Safra~\cite{GoldreichSafra} 
in the context of the PCP theorem. 
They later have been identified by
Dinur and Reingold~\cite{DinurReingold} 
as a central component in gap amplification. To get some intuition to
this test, note that a direct product
function clearly passes the test with 
probability $1$. 
Thus, we say that the tester has perfect completeness.
The soundness of the test -- namely the probability that a table $F$ which is far from a direct product encoding passes the test -- is more difficult to analyze. Intuitively, 
querying $F$ at a single location gives the value of 
a (supposed) $f$ on $k$ inputs.
Thus, if $F$ is far from any direct product
function, the chance this will be detected should grow with $k$. Formalizing
this intuition is more challenging however, and works in the literature 
are divided into two regimes: the so-called 
$99\%$ regime, and the 
$1\%$ regime. To be more precise, suppose
the table $F$ passes the direct product tester $\mathcal{T}_t$ with probability 
at least $s>0$; what can be said 
about its structure?

In the $99\%$ regime, namely the case where
$s = 1-\eps$ is thought of as close to $1$, results in the literature~\cite{DinurReingold,DinurSteurer} show that $F$ has to be close to a direct product function. More specifically, for $t = \Theta(k)$ the result of~\cite{DinurSteurer} asserts that there exists $f\colon [n]\to\{0,1\}$ such that 
$F[A] = f|_{A}$ for $1-O(\eps)$ fraction 
of the $k$-sets $A$. A structural result of this form is a useful building block in several applications. It can be used to
construct constant 
query PCPs
with constant soundness; it also serves 
as a building block in other results 
within complexity theory; see for instance~\cite{DDGKS,DFH}. 

The $1\%$ regime, namely the case where 
$s = \delta$ is thought of as a small 
constant, is more challenging. 
In this case, the works~\cite{DinurG08,ImpagliazzoKW09} 
show that $F$ has to be correlated with 
a direct product function. More specifically, these works show that 
for (say) $t = \sqrt{k}$ if $\delta \geq 1/k^{\Omega(1)}$, then there exists $f\colon [n]\to\{0,1\}$ such that for at least $\delta^{O(1)}$ fraction of the $k$-sets $A$, we have that
\[
\Delta(F[A], f|_{A})\leq k^{-\Omega(1)},
\]
where for two strings $x,y\in\{0,1\}^k$, $\Delta(x,y)=\frac{\#\{i\in [k]~|~x_i\neq y_i\}}{k}$ denotes the fractional 
Hamming distance between them.\footnote{In contrast to the
$99\%$ regime, in this case one has to 
settle with agreement with $f$ only on 
a small portion of the $k$-sets, and 
furthermore this agreement is not perfect; it is on $(1-o(1))$ fraction of the elements in the $k$-sets. 
As discussed in~\cite{DinurG08,ImpagliazzoKW09}, qualitatively speaking (namely, up to 
the precise parameters) this is the 
best type of results possible.} 
The motivation for studying this
more challenging regime of parameters
stems mainly from the perspective 
of hardness amplification (where 
one wishes to show that if a given 
task is somewhat hard, then repeating
this task $k$-times in parallel 
gets exponentially harder) as well as 
from the study of PCPs with small soundness. 
Indeed, in~\cite{ImpagliazzoKW09} 
the authors show that direct product
testers similar to the above facilitate 
soundness amplification schemes 
for PCPs with similar performance 
to parallel repetition theorems~\cite{Raz,Holenstein,Rao,BravermanGarg,DinurSteurerAnalytical}. Direct product
testers in the low soundness regime have additional applications in property 
testing, as well as in the study of the 
complexity of satisfiable 
constraint satisfaction 
problems
~\cite{BKM1,BKM2,BKM3,BKM4}. 

\subsection{Size Efficient $2$-Query Direct Product Testing}
In the context of PCPs and hardness 
amplification, one typically thinks of the parameter $n$ as
very large, and $k$ as a large constant 
number. With this in mind, representing 
an assignment $f\colon [n]\to\{0,1\}$ 
using its direct product encoding incurs
a polynomial blow-up in size. Indeed, 
this type of step is often the only 
step in the PCP reduction that introduces 
a polynomial (as opposed to just linear)
blow-up in the instance size. In 
this light, a natural question is whether
it is possible to perform  hardness 
amplification with a significantly 
smaller blow up in the encoding/instance size. Efficient 
schemes of this type are often referred to as 
``derandomized direct product tests'', 
``derandomized hardness amplification''
 or ``derandomized parallel repetition 
 theorems''.

In~\cite{ImpagliazzoKW09} a more efficient 
hardness amplification procedure is proposed. Therein, instead of considering 
all $k$-sets inside $[n]$, the domain 
$[n]$ is thought of as a vector space 
$\mathbb{F}_q^{d}$ and one considers all
subspaces of dimension $\log_q(k)$. It is easy to see
that the encoding size then becomes $n^{\Theta(\log k)}$, making it more
efficient. The paper~\cite{ImpagliazzoKW09}
shows that direct product testers analogous
to the tester above work in this setting
as well; they essentially 
match all of the results achieved by the 
Johnson scheme. Building upon~\cite{ImpagliazzoKW09}, Dinur and Meir~\cite{DinurMeir} show how to establish
parallel repetition theorems using the
more efficient direct product encoding
via subspaces. This parallel repetition
theorem works for structured instances, 
which the authors show to still capture
the entire class NP. 

High dimensional expanders (HDX), which have recently surged in popularity, can be 
seen as sparse models of the Johnson graph. This leads us to 
the main problem considered in this paper, 
due to Dinur and Kaufman~\cite{DinurK17}:
\begin{tcolorbox}
Do high dimensional expanders facilitate 
direct product testers in the low soundness regime?
\end{tcolorbox}
The main goal of this paper is to 
investigate the type of expansion properties that
are necessary and sufficient  for direct product 
testing with low soundness. 
It is known that there are HDXs of size $O_k(n)$ and 
$O_k(1)$ degree,
and if any of these objects facilitates a 
direct product tester with small soundness,
they would essentially be the ultimate form of derandomized direct product 
testers.\footnote{We remark that to be useful, 
it seems that a derandomized direct
product tester would need to roughly have 
equal degrees. This is because in applications, each one of the values of the 
encoded function $f\colon [n]\to\{0,1\}$ 
is ``equally important''. In that case, the 
requirement of being a $O(n)$ sized direct product
tester is equivalent to having $O(1)$ degree.} To state our 
results, we first define the usual notion
of spectral high dimensional expansion, followed 
by our variant of the well-known
 notion of coboundary expansion.

\subsubsection{High Dimensional Local Spectral Expanders}
A $d$-dimensional complex is 
composed of $X(0) = \{\emptyset\}$, a set of vertices $X(1)$, 
which is often identified with $[n]$ 
and a set of $i$-uniform hyperedges, $X(i)\subseteq \binom{X(1)}{i}$, for each $i=2,\ldots,d$. A $d$-dimensional complex $X = (X(0),X(1),\ldots,X(d))$ is called simplicial if it is downwards closed. Namely, if for every $1\leq i\leq j\leq d$, and every $J\in X(j)$, if $I\subseteq J$ has size $i$, then $I\in X(i)$.
The size of a complex is the total number 
of hyperedges in $X$. The degree of a vertex $v\in X(1)$ is the number of faces in $X(d)$ containing it, and the degree of a complex $X$ is the maximum of the degree over all the vertices in $X(1)$. 

We need a few basic notions regarding simplicial complex, and we start by presenting the notion of links and spectral expansion.
\begin{definition}
For a $d$-dimensional simplicial complex $X = (X(0),X(1),\ldots,X(d))$, $0\leq i\leq d-2$
and $I\in X(i)$, the link of $I$ is the $(d-i)$-dimensional complex $X_I$ whose faces are given as
\[
X_I(j-i) = \{J\setminus I~|~J\in X(j), J\supseteq I\}.
\]
\end{definition}
For a $d$-dimensional complex $X=(X(0),X(1),\ldots,X(d))$ 
and $I\in X$ of size at most $d-2$, the 
graph underlying the link of $I$ is 
the graph whose vertices are $X_I(1)$ 
and whose vertices are $X_I(2)$.

\paragraph{Distributions over the complex.} It is convenient to equip
a complex $X$ with a measure $\mu_k$
for each one of its levels $X(k)$. 
For $k = d$ we consider the measure $\mu_d$ which 
is uniform over $X(d)$; for each $k<d$,
the measure $\mu_k$ is the push down
measure of $\mu_d$: to generate a sample 
according to $\mu_k$, sample $D\sim \mu_d$ 
and then $K\subseteq D$ of size $k$ uniformly. Abusing notation, we will
refer to all of the measures $\mu_k$ 
simply as $\mu$, as the cardinality of 
the sets in discussion will always be 
clear from context. The set of measures in the link of $I$ is the natural set of measures we get by conditioning $\mu$ on containing $I$. 

Equipped with measures over complexes, we may now define 
the notion of spectral HDX.

\begin{definition}
\label{def:gamma-hdx}
    A $d$-dimensional simplicial complex 
    $X$ is called 
    a $\gamma$ one-sided (two-sided) local spectral expander if for every $I\in X$ of size at most $d-2$,
    the second eigenvalue (singular value)
    of the normalized adjacency matrix of the graph underlying the link of $I$
    is at most $\gamma$. 
\end{definition}
In this work, we will only be concerned 
with simplicial complexes that are very
strong spectral expanders. With this
regard, following the works of~\cite{LSV1,LSV2,EvraKaufman} 
one can show that for every $\gamma>0$ and every $d\in\mathbb{N}$ there exists 
an infinite family of $d$-dimensional complexes of linear size that are $\gamma$ one-sided or two-sided local expanders (see~\cite[Lemma 1.5]{DinurK17}).

\subsubsection{Results in the High Soundness Regime}
Dinur and Kaufman~\cite{DinurK17} were
the first to consider the question of direct product testing over HDX. They showed 
that a sufficiently good high dimensional
spectral expander admits a direct product
tester in the high soundness regime. 
The tester they consider is essentially
the same as the tester in the Johnson 
scheme; one thinks of $k$ which is much
 larger than $1$ but much smaller than the dimension of the complex $d$. 
 The tester has parameters 
$1\leq s\leq k/2$ and is given
oracle access to a table
$F\colon X(k)\to \{0,1\}^k$, and proceeds as 
follows: 
\begin{mdframed}
\begin{agr-test}[$F,k,s$]\label{agr-test-hdx}\mbox{}
\begin{enumerate}
    \item Sample $D\sim \mu_{d}$.
    \item Sample $I\subseteq D$ of size $s$ uniformly.
    \item Sample $I\subseteq A, A'\subseteq D$ of size $k$ uniformly.
    \item Accept if $F[A]|_{I} = F[A']|_{I}$.
\end{enumerate}
\end{agr-test}
\end{mdframed}
Henceforth, we refer to this 
test as the $(k,s)$ direct 
product tester over $X$.
Dinur and Kaufman consider the case where $s = k/2$, and proved that for every $\eps>0$, provided that $\gamma$ is sufficiently small, if $F\colon X(k)\to\{0,1\}^k$ passes the above test with probability at least $1-\eps$, then 
there exists $f\colon X(1)\to \{0,1\}$ 
such 
\[
\Pr_{A\sim\mu_k}[F[A]\equiv f|_{A}]\geq 1-O(\eps).
\] 
A follow-up 
work by Dikstein and Dinur~\cite{DiksteinDinur} further refined
this result, and investigated more general
structures that support direct product
testing in the high soundness regime.

A problem related to direct
product testing, called 
the list agreement testing 
problem, was considered in the 
high soundness regime 
by Gotlib and Kaufman~\cite{GotlibK23}. 
In the list agreement
testing problem, 
each face is assigned a 
list of $m = O(1)$ 
functions, and one performs
a local test on these
lists to check that they 
are consistent. With this
in mind, the result of Gotlib and Kaufman~\cite{GotlibK23} 
asserts that under certain
structural assumptions 
on the lists, if the underlying complex has sufficiently good coboundary expansion, 
then one
can design a $3$-query 
list agreement tester that is sound. 
The list agreement problem
will play an important role
in the current work, and while we do not know
how to use the result of 
Gotlib and Kaufman for 
our purposes, their work inspired
us to look at connections 
between agreement testing
and notions of coboundary expansion.

\subsection{Main Results}
Despite considerable interest, no 
positive nor negative results are 
known regarding the question of whether HDX support direct product testers in the low-soundness regime. 
In fact, the majority of applications of HDX are in the high soundness regime, with the first construction of $c^3$-locally testable codes~\cite{DinurELLM22} and quantum LDPC codes~\cite{EvraKZ20,PanteleevK22,LeverrierZ22}.
At a first 
glance, this seems surprising: 
very good expander graphs give rise 
to objects in the low-soundness regime, and 
high dimensional expanders are essentially
their higher order analogs.

The main contribution of this work is 
an explanation to this phenomenon. We 
show that, to facilitate direct product
testers in the low-soundness regime, a 
high dimensional spectral expander must
posses a property that may be seen as 
a generalization of \emph{coboundary expansion}~\cite{LinialMeshulam}. On the other hand, we also show that coboundary expansion is sufficient to get direct product testers. Thus, to construct
constants degree, sparse
complexes facilitating 
direct product testing, 
one should first come 
up with local spectral expanders
that are also coboundary expanders. In section~\ref{sec:cobdry-discussion} we discuss what is known regarding sparse constructions of coboundary expanders.

Below, we state our main results
regarding the soundness of the test, which give 
analysis of the $(k,s)$ 
tester defined above  
assuming expansion properties
of the complex $X$. In a concurrent and independent work, Dikstein and Dinur~\cite{DiksteinDinur2023} established related results.

\subsubsection{Coboundary Expansion}
For
convenience, we follow the presentation of coboundary expansion from~\cite{DiksteinD23}. Suppose we have 
a function $f\colon X(2)\to \mathbb{F}_2$. 
The function $f$ is said to be consistent on the triangle $\{u,v,w\}\in X(3)$ if
it holds that $f(\{u,v\}) + f(\{v,w\}) + f(\{u,w\}) = 0$. What can we say about the structure of functions $f$ which are consistent with respect to $1-\xi$ 
measure of the triangles? Clearly, if $f$ is a function of the form $f(\{u,v\}) = g(u) + g(v)$ for some $g\colon X(1)\to\mathbb{F}_2$, then it is consistent
with respect to all triangles. In the case that $X$ is a coboundary expander, the converse is also true: any $f$ which is $(1-\xi)$ triangle consistent
is $O(\xi)$-close to a function of this form.

More broadly, the notion of coboundary expansion often refers to a property of
higher dimensional faces, and to more general
groups beyond $\mathbb{F}_2$. 
We refrain from defining these notions precisely and instead turn to our variant of coboundary expansion, which we show governs the soundness of direct product testing.

\subsubsection{Unique-Games Coboundary Expansion}
Our notion of coboundary expansion replaces the group $\mathbb{F}_2$ with non-Abelian groups,
more specifically with the permutation 
groups $S_{m}$; we also need to consider higher 
dimensional faces. Some definitions in this spirit have been made, for example 
in~\cite{DinurMeshulam,GotlibK23}, and our notion 
is inspired by theirs.

Given a $d$-dimensional complex $X$ 
and an integer $t\leq d/3$, we consider
the graph $G_{t}[X] = (X(t), E_t(X))$ whose vertices are the $t$-faces of $X$, namely $X(t)$, 
and $(u,v)$ is an edge if $u\cup v\in X(2t)$. We say $T = (u,v,w)$ is a triangle in $G_t[X]$ if each of $u,v,w \in X(t)$ and $u \cup v \cup w \in X(3t)$.

\begin{definition}\label{def:triangle_consistent_weak}
Let $X$ be a $d$-dimensional complex 
and let $t$ be an integer such that $t\leq d/3$. Let $\pi\colon E_t(X) \to S_m$ be a function that satisfies $\pi(u,v) = \pi(v,u)^{-1}$ for all $(u,v) \in E_t[X]$. We say that $\pi$ is consistent on the triangle $(u,v,w)$ in $G_t[X]$ if $\pi(u,v)\pi(v,w) = \pi(u,w)$. 

We say that $\pi$ is $(1-\xi)$-consistent on triangles if sampling $T\sim \mu_{3t}$ and then splitting $T$ as a triangle $u\cup v\cup w$ uniformly where $|u| = |v| = |w| = t$,
\[
\Pr_{
\substack{T\sim \mu_{3t}\\ T = u\cup v\cup w}}
\big[\pi(u, v) \pi(v, w) = \pi(u, w)\big]
\geq 1-\xi.
\]
\end{definition}

One way to think of this definition is 
as a locally consistent instance
of Unique-Games (see Definition~\ref{def:unique-games}). Indeed, a $\pi$ as above
specifies a Unique-Games (UG) instance on the graph $G_t[X]$ whose constraints are 
locally consistent on triangles. The 
goal in this UG instance may
be thought of assigning elements from $[m]$
to the vertices of $G_t[X]$, namely 
finding an assignment $A\colon X(t)\to [m]$, so  as to maximize the fraction of
edges $(u,v)$ for which $A(u) = \pi(u,v)A(v)$. 

With this definition in mind, we can now present a simplified version of our notion of coboundary expansion. One way to arrive at a locally consistent UG instance 
as in Definition~\ref{def:triangle_consistent_weak} is to first pick some
function $g\colon X(t) \to S_m$ and then define $\pi(u,v) = g(u)g(v)^{-1}$. Thus, a natural question is whether there are other ways to construct locally consistent UG instances on $G_t[X]$. In simple terms, our simplified notion of UG coboundary expansion asserts that this is essentially the only way to arrive
at instances of this form. More precisely:

\begin{definition}\label{def:ug_coboundary_expander_simplified}
We say that a $d$-dimensional simplicial 
complex $X$ is an $(m,r,\xi,c)$ UG coboundary expander if for all $t\leq r$
and for all functions $f\colon E_t[X]\to S_m$ that are $(1-\xi)$-consistent on triangles, there is $g\colon X(t)\to S_m$
such that
\[
\Pr_{u\cup v\sim \mu_{2t}}
\big[\pi(u,v) = g(u)g(v)^{-1}\big]\geq 1-c.
\]
\end{definition}
We remark that if a complex $X$ is an $(m,r,\xi,c)$ UG coboundary expander,
then given a $(1-\xi)$-locally consistent
instance of Unique-Games on $G_t[X]$ 
for some $t\leq r$, one may find an
assignment satisfying at least $1-c$
fraction of the constraints. Indeed, by 
definition, given the constraint map
$\pi$ we may find $g\colon X(t)\to S_m$
such that $\pi(u,v) = g(u)g(v)^{-1}$
with probability at least $1-c$ 
over the choice of $u\cup v\sim \mu_{2t}$.
Thus, taking the labeling $A(v) = g(v)(1)$, we see that $A$ satisfies all edges
on which $\pi(u,v) = g(u)g(v)^{-1}$.

The first result of this paper asserts
that a spectral HDX which is a UG coboundary expander admits a direct product
tester in the low soundness regime.
\begin{theorem}\label{thm:HDX_dp_weak_simplified}
Suppose that a simplicial complex $X$ is 
a sufficiently good spectral and UG coboundary expander.
If $F\colon X(k)\to\{0,1\}^k$ passes the $(k,\sqrt{k})$ direct product
test on $X$ with probability $\delta$, then there is 
$f\colon X(1)\to\{0,1\}$ such that
\[
\Pr_{A\sim \mu_k}
\big[\Delta(f|_A, F[A]) =o(1)\big]\geq \Omega_{\delta}(1).
\]

\end{theorem}
In words, being a UG coboundary expander 
is a sufficient condition for a spectral
expander to support a low soundness direct 
product tester. As far as we know, however, this condition may not be 
necessary; below, we present a condition
which is both necessary and sufficient. 
Nevertheless, we chose to present its
simpler to state version, Definition~\ref{def:ug_coboundary_expander_simplified}, as we find it more
appealing, intuitive and resembling
non-Abelian variants of the usual
notion of coboundary expansion.

\begin{remark}\label{rmk:non_ab_cobound}
 The usual definition of coboundary expansion in the literature refers to Abelian
 groups such as $\mathbb{F}_2$, 
 see for example~\cite{KaufmanKazLub,KaufmanMass1,KaufmanMass2,GotlibK23,DiksteinD23}. In the $\mathbb{F}_2$
 setting, coboundary expansion for 
 the base graph can be seen 
 as a UG instance 
 over $\mathbb{F}_2$, but it is often phrased in topological
 notions using the boundary and 
 coboundary maps; these 
 definitions extend well to 
 higher dimensional faces. 
 Coboundary expansion has also been defined for non-Abelian groups~\cite{DinurMeshulam,KaufmanMass2,GotlibK23}, 
 however, as far as we know, these definitions
 coincide with ours only for the 
 case that $t=1$ in Definition~\ref{def:ug_coboundary_expander_simplified}. 
\end{remark}

\subsubsection{A Necessary and Sufficient Condition for Low Soundness Direct Product Testing}
We now move on to stating a more complex
version of 
Definition~\ref{def:ug_coboundary_expander_simplified} which
is both necessary and sufficient 
for low soundness direct product testing.
Let us again
consider the graph $G_t[X]$ and  
a $(1-\xi)$ triangle consistent assignment of permutations on the 
edges $\pi\colon E_t[X]\to S_m$. However,
unlike before, these permutations are guaranteed to satisfy
an additional premise. Precisely, 
suppose that each face $u\in X(t)$ 
is assigned a list of $m$ elements from
$\{0,1\}^t$, say $L(u) = (L_1(u),\ldots,L_m(u))$, and 
each face $T\in X(3t)$ 
is also assigned a list 
$L'(T) = (L'_1(T),\ldots,L'_m(T))$.
In words, we would like the permutations
$\pi$ to be consistent with the lists
with respect to concatenations.
Towards this end, we introduce
a convenient 
but informal
notation to compare strings. 
Given $u,v\in X(t)$ that are 
disjoint and strings $L_i(u), L_i(v)\in \{0,1\}^t$, 
we shall think of $L_i(u)$ 
as an assignment to the vertices 
in $u$ and of $L_i(v)$ as
an assignment to the vertices 
in $v$. Thus, the notation
$L_i(u)\circ L_i(v)$ will 
be a string in $\{0,1\}^{2t}$
which encodes the assignment 
to $u\cup v$ provided by
the concatenation of the
two assignments.
More generally, given 
$u,v$ disjoint and 
list assignments $L(u),L(v)$
we define
\[
L(u)\circ L(v)
=(L_1(u)\circ L_1(v),
\ldots,
L_m(u)\circ L_m(v)).
\]
Lastly, given a list $L(u)$
as above and $\pi\in S_m$, 
we define $\pi L(u) = (L_{\pi(1)}(u),\ldots,L_{\pi(m)}(u))$.

\begin{definition}\label{def:triangle_consistent_strong}
Let $L\colon X(t)\to (\{0,1\}^t)^m$, 
$L'\colon X(3t)\to(\{0,1\}^{3t})^m$, 
and $\xi>0$.
We say $\pi$ is $(1-\xi)$-consistent
with the lists $L$ and $L'$
if choosing $T\sim \mu_{3t}$ 
and a splitting $T = u\cup v\cup w$
into a triangle, we have that
\[
\Pr_{\substack{T\sim \mu_{3t}\\ 
T=u\cup v\cup w}}
\big[L'(T) = L(u)\circ \pi(u,v) L(v) \circ \pi(u,w) L(w)\big]\geq 1-\xi.
\]
We say that $\pi$ is $(1-\xi)$-strongly triangle consistent if there are lists 
$L$ and $L'$ such that $\pi$ is $(1-\xi)$-consistent with respect to the lists $L$
and $L'$.

\end{definition}
It is easy to see that if $\pi$ is $(1-\xi)$-strongly triangle
consistent, then $\pi$ is $(1-O(\xi))$-triangle consistent (see Claim~\ref{claim:strong-to-weak-consistency}). Thus, the class of triangle
consistent functions $\pi$ is more restrictive.
With the notion of strong triangle consistency we are now ready to state a weaker variant of Definition~\ref{def:ug_coboundary_expander_simplified}; the only difference between
the two definitions is that in the definition below, we only require
that any strongly triangle consistent 
assignment admits a global structure. 
More precisely:
\begin{definition}\label{def:ug_coboundary_expander}
We say that a $d$-dimensional simplicial 
complex $X$ is a weak $(m,r,\xi,c)$ UG coboundary expander if the 
following condition is satisfied for all $t\leq r$. Suppose $\pi\colon E_t[X]\to S_m$ is a $(1-\xi)$-strongly triangle consistent function. Then there 
exists $g\colon X(t)\to S_m$
such that
\[
\Pr_{u\cup v\sim \mu_{2t}}
[\pi(u,v) = g(u)g(v)^{-1}]\geq 1-c.
\]
\end{definition}
The parameter $r$ in Definition~\ref{def:ug_coboundary_expander} 
is often referred to as the level at which UG coboundary expansion holds. With the notion of weak UG coboundary expansion, we can now state a stronger 
version of Theorem~\ref{thm:HDX_dp_weak_simplified}. 
Roughly speaking, the following two
results asserts that for a sufficiently good spectral simplicial complex 
$X$, the direct product tester over 
$X$ works in the low soundness regime 
if and only if $X$ is a weak UG coboundary
expander with sufficiently good parameters.

\begin{theorem}\label{thm:HDX_dp_weak}
The following results hold for any simplicial complex $X$.
\begin{enumerate}
    
\item
\textbf{Weak UG-coboundary is Necessary:} 
    If a simplicial complex $X$ is a sufficiently good spectral expander
    which is not a UG coboundary expander, then there is $\delta>0$ such that for sufficiently large $k$, there is $F\colon X(k)\to\{0,1\}^k$ 
that passes the $(k,\sqrt{k})$
direct product tester with probability $\delta$ 
and yet for all $f\colon X(1)\to\{0,1\}$
we have that
\[
\Pr_{A\sim \mu_k}
[\Delta(F[A], f|_A) = o(1)]=o(1).
\]

\item 
    \textbf{Weak UG-coboundary is Sufficient:} 
    For all $\eps, \delta>0$,
    if a simplicial complex $X$ is a sufficiently good spectral expander
    and a weak UG coboundary expander on level $O(1)$, then the direct
product test over $X$ with respect
to sufficiently large $k$ has soundness $\delta$. Namely, if 
$F\colon X(k)\to\{0,1\}^k$ passes the
$(k,\sqrt{k})$ direct product tester with respect to 
$X$ with probability at least $\delta$, 
then there is $f\colon X(1)\to\{0,1\}$
such that 
\[
\Pr_{A\sim \mu_k}
[\Delta(F[A], f|_A) \leq \eps]\geq \Omega(1).
\]
\end{enumerate}
\end{theorem}

We refer the reader to
Theorems~\ref{thm:necessary-formal} and~\ref{thm:sufficient_ugexpand_formal}
for more formal statements. We use our necessary result above to conclude that some of the best known sparse spectral expanders -- namely some 
LSV complexes -- do not 
support direct product testers in the low soundness regime precisely 
because they fail to satisfy coboundary expansion (see Corollary~\ref{cor:lsv}). As the result of Dinur and Kaufman~\cite{DinurK17} 
asserts that LSV complexes admit direct product testers in the high soundness regime, we conclude that the low soundness regime is qualitatively different.

\begin{remark}
We would like to remark
that the ``sufficient'' part of  Theorem~\ref{thm:HDX_dp_weak} 
only uses UG coboundary expansion on a constant
level $r$, which is potentially much smaller than $k$. Additionally, it suffices that the complex is an $(m,r,\exp(-o(r)),c)$-UG coboundary expander, i.e. there exists a global UG solution for all UG instances that are at least $1-f(r)$ triangle consistent for some fixed function $f(r) = \exp(-o(r))$. The proof of this quantitative version follows along similar lines of the proof of Theorem~\ref{thm:sufficient_ugexpand_formal}, and we elaborate on it in Section~\ref{sec:appx-improved-cbdry}.
\end{remark}
In the above theorem, the structure for $F$ we get is relatively weak though, and only asserts that with significant probability over the choice of $A\sim \mu_k$, we have that $F[A]_i = f(i)$ 
for $(1-\eps)$ fraction of $i\in A$. In the next theorem, we show that 
if the level $r$ on which
coboundary expansion
holds is linear in $k$, 
then the conclusion of Theorem~\ref{thm:HDX_dp_weak}  can be 
strengthened to say 
that with 
significant probability over $A\sim \mu_k$, it holds that $F[A]_i = f(i)$
for all but constantly many 
of $i\in A$.\footnote{For that, we need to consider a direct product tester with intersection parameter $s$, which is significantly smaller than $k$ but is linear in it. Indeed, it is easy to see that the conclusion of Theorem~\ref{thm:HDX_dp} would fail if either $s \leq k^{0.99}$ or $s\geq k/100$.}

\begin{theorem}\label{thm:HDX_dp}
    If a simplicial complex $X$ is a sufficiently 
    good spectral expander, and for 
    $k\in\mathbb{N}$ it holds that $X$ is a sufficiently good weak UG coboundary expander on level $\Omega(k)$, then the direct
product test over $X$ with respect
to $k$ has soundness $\delta$. Namely, 
for all $\delta>0$ there is $\eta>0$ such that if 
$F\colon X(k)\to\{0,1\}^k$ passes the $(k,\eta k)$
direct product tester with respect to 
$X$ with probability at least $\delta$, 
then there is $f\colon X(1)\to\{0,1\}$
such that 
\[
\Pr_{A\sim \mu_k}
\big[\Delta(F[A], f|_A)\leq O(1/k)\big]\geq \Omega(1).
\]
\end{theorem}
In Section~\ref{sec:known_complex} we examine several well known complexes. We show that dense complexes such as the complete and the Grassmann complex are UG coboundary expanders. On the flip side we use well-known theorems that some LSV complexes are not coboundary expanders, to show that they fail to support direct product testers.

\subsection{Proof Overview}\label{sec:techniques}
In this section we give 
an overview for the proofs of
Theorems~\ref{thm:HDX_dp_weak}
and~\ref{thm:HDX_dp}. 
The proofs of these two theorems 
is basically the same; the only
point in which they defer is what 
direct product theorem is used 
for the Johnson scheme. Thus, 
we will focus on the setting in 
Theorem~\ref{thm:HDX_dp_weak}.
\subsubsection{The Proof that weak UG Coboundary Expansion is Necessary}
Suppose that weak UG coboundary
expansion for $X$ fails. 
That means that we have a complex $X$ violating
the weak UG coboundary expansion, hence
there is a $(1-o(1))$-strongly triangle consistent Unique-Games
instance on $G_r[X]$ that has no 
$(1-c)$-global solution. 
Namely, denoting the constraints of $G_r[X]$ 
by $\pi\colon E_r(X)\to S_m$, for all 
$g\colon X(r)\to S_m$ 
it holds that $\pi(u,v) = g(u)g(v)^{-1}$ with probability
at most $1-c$ over the choice
of $(u,v)\sim \mu_{2r}$. Furthermore, we may find
lists $L$ and $L'$ such 
that $f$ is $(1-o(1))$-consistent
with $L,L'$ as in Definition~\ref{def:triangle_consistent_strong}. 
The parameters $m$ and $r$ should be thought of 
as large constants, and $c>0$ should be thought of 
as a small constant bounded away from $0$. 
The proof now
proceeds in the following steps: 
\vspace{1ex}
\textbf{Preprocessing:} We claim that (after appropriate pre-processing) 
we may assume that $G_r[X]$
doesn't have a solution that
satisfies more than $(1-c/m)$
of the constraints of $G_r[X]$.
Indeed, if there exists such 
a solution $A\colon X(r)\to [m]$, we may 
(1) remove $A(u)$ from the 
list $L[u]$, 
(2) For all edges $(u,v)$ remove the pair $(A(u), A(v))$ as a satisfying assignment
of $\pi(u,v)$. We argue that this process must terminate and in the end must produce a Unique-Games instance 
with no good global solution. For 
simplicity of notation, we 
assume henceforth that $G_r[X]$
doesn't have an assignment
satisfying at least 
$(1-c')$ of the constraints
to begin with (where $c' = c/m$). 

\vspace{1ex}
\textbf{Lifting the lists:} Take $k\gg r$, and
consider $k$-faces. Taking 
$K\sim \mu_k$ randomly 
and taking a random triangle
over $r$-faces in $K$, 
we see that it is marginally
distributed as a random triangle in $X$. Thus, by linearity of expectation 
the number of inconsistent
triangles in $K$ is at most 
$o(k^{3r}) = o(1)$. 
It follows that with probability
at least $1-o(1)$ over
the choice of $K$, all 
$r$-triangles in $K$ are 
consistent. 
In this case, it is easily
seen that there exists a unique 
list $L[K]$ of $m$ strings 
from $\{0,1\}^k$ which is
consistent with all the $r$-faces inside $K$, and we fix
it. For sake of simplicity, 
we assume that we have constructed a list $L[K]$
for each $K\in X(k)$.

\vspace{1ex}
\textbf{Constructing the assignment and its soundness:} It is easy to see that 
sampling $D\sim \mu_d$, 
$I\subseteq D$ of size 
$\sqrt{k}$ and then 
$I\subseteq K, K'\subseteq D$
uniformly, with probability $1-o(1)$ the
lists of $K$ and $K'$ 
are in $1$-to-$1$ correspondence 
with respect to agreement on $I$. \footnote{Roughly speaking, the reason is that the argument above could be made also for
$2k$-faces, and one gets a 
list of assignments that is 
in $1$-to-$1$ correspondence 
to the lists of $K$ faces with 
respect to containment.} 
Thus, choosing an assignment 
$F\colon X(k)\to\{0,1\}^k$ 
by assigning each $k$-face $K$
a random element from $L[K]$
gives an assignment which in expectation satisfies at least 
$1/m$ fraction of the constraints, and we fix $F$ achieving this expectation henceforth.

\vspace{1ex}
\textbf{No global structure:}
 Suppose for contradiction that $h\colon X(1)\to \{0,1\}$ has that $h(K) \equiv F[K]$ for at
least $\eta$ fraction of the
$k$-faces (strictly speaking, 
we may only assume that $h(K)$ 
is close to $F[K]$, but 
we ignore this issue for the sake of clarity), and let the set of 
these $k$-faces be denoted by $\mathcal{K}$.
Consider $\mathcal{R} = \{R\in X(r)~|~\exists K\in\mathcal{K}, R\subseteq K\}$. By the sampling
property of $X$, provided 
that $k$ is sufficiently large, we get $\mu_r(\mathcal{R})\geq 1-c/10$. Note that 
for $R\in\mathcal{R}$, 
$h_{R} = (h|_{K})|_{R}\in L[K]|_{R} = L[R]$ where $K\in\mathcal{K}$ contains $R$. It follows that $h$ satisfies
at least $1-c$ fraction 
of the constraints of $G_r[X]$,
and contradiction.

\subsubsection{The Proof that weak UG Coboundary Expansion is Sufficient}
Suppose we have an assignment 
$F\colon X(k)\to\{0,1\}^k$ that 
passes the $(k,s)$ direct product test over 
$X$ with probability at least $\delta$; throughout, $s=\sqrt{k}$. 
We also take a parameter 
$t$ where $k\ll t\ll d$. 
The proof now proceeds by
the following steps. 

We start by localizing to Johnson graphs, and show that almost all of them can be equipped with a short list of assignments that explain almost all of the agreement inside them. This naturally leads us to \emph{list-agreement testing} as described 
above (and defined more formally below). This problem was addressed in the work of Gotlib and Kaufman~\cite{GotlibK23} who analyzed a certain 
list agreement tester and proved it was sound in the high soundness regime. 
Although our result for list agreement testing is similar in spirit, our tester is different from theirs, and we do not how to know use their
techniques/tester for our purposes.

\subsubsection{Reduction from 1\% Direct Product Testing to 99\% List Agreement Testing}
\textbf{Localizing to a Johnson:}
The first part of the proof is to localize
the test to Johnson schemes. 
For $T\in X(t)$
define ${\sf pass}_t(F; T)$ 
as the probability
the following test passes: sample $I\subseteq T$ of size $s$ and then $I\subseteq A,B\subseteq T$ independently, and test that $F[A]|_{I} = F[B]|_{I}$.
\footnote{The diligent
reader may notice that the 
distribution of $(A,B)$ in
the test inside $T$ and in 
our direct product tester 
over $X$ is not quite the
same. The probability that $A\cap B = I$ in the test 
inside $T$ is $1-\Theta(k/t)$, 
whereas it is $1-\Theta(k/d)$
in the direct product tester. 
Conditioned on this event 
(which has probability close to 
$1$) the two distributions are identical, and thus they are 
close in total variation distance. Therefore, we will 
think of the two distribution as essentially
the same.
}
Similarly, we may define 
${\sf pass}_d(F; D)$
for $D\in X(d)$.
Clearly, 
one has that
$\E_{T\sim \mu_t}[{\sf pass}_t(F;T)],
\E_{D\sim \mu_d}[{\sf pass}_d(F;D)]
\geq \delta$ 
and hence for a large 
fraction of the $D$'s 
we have that 
${\sf pass}_d(F;D)\geq \delta/2$.
By the sampling properties
of HDX (which follow as $X$
is a sufficiently good spectral
expander) we are able to derive
the stronger conclusion 
that ${\sf pass}_d(F; D)\geq \delta^2/100$ for $1-O_{\delta}(\gamma)$ for the $d$-faces $D$. 
We refer to such faces $D$ as 
good.

\vspace{1ex}
\textbf{Getting a list on
each good Johnson:} 
Fix a good $D\in X(d)$, consider 
the Johnson graph $\binom{D}{k}$
and the assignment $F_D$ which
the restriction of $F$ to $\binom{D}{k}$. Then the fact 
that ${\sf pass}_d(F; D)\geq \delta^2/100$ implies that $F_D$
passes the Johnson scheme 
direct product test 
inside $T$ with probability
at least $\delta^2/100$. 
Thus, by direct product theorems
over the Johnson scheme -- 
and more precisely by the 
result of~\cite{DinurG08} -- we conclude 
that there is a function 
$f_{D}\colon D\to \{0,1\}$
such that $\Delta(f_{D}|_{A}, F_D[A]) = o(1)$ for at least 
$\delta' = \delta^{O(1)}$ 
fraction of the $k$-faces $A\subseteq D$. 
To simplify terminology, 
we refer to an $A$ on which 
$\Delta(f_{D}|_{A}, F_D[A]) = o(1)$ as an $A$ on which 
$f_{D}$ and $F_D[A]$ 
agree.

We would like to form a list of
all of the functions that achieve
$\delta'$ agreement with $F_D$.
The list of all of these 
functions though may be large 
(its size may typically depend on the parameters $k$ and 
$d$, whereas we wish our list 
sizes to only be a function of
$\delta$ and $\delta'$). To 
remedy this situation, we create
these lists in a more careful 
manner. One way to go about it
is to construct an ``$\eps$-net''
for all these functions, 
and the reader should have this
in mind (our precise execution 
is a bit different but morally
the same). Thus, we are able
to find, for each good $D$, a maximal list $f_{1,D},\ldots,f_{m,D}$ 
of functions that have at least
$\delta'$ agreement with $F_D$, 
and we have a list size 
bound $m\leq m(\delta')$. 
Here, maximality 
asserts that no function 
$f$ that is somewhat far from 
all $f_{i,D}$ has agreement 
at least $\delta'$ 
with $F_D$. We also remark that
a-priori, the list size $m$ 
could also depend on the 
identify of the face $D$, 
but we omit it from the
notation for now.\footnote{The reason is that, as we prove 
in subsequent steps, the
list size typically does 
not depend on the identity 
of the face $D$.}

\vspace{1ex}
\textbf{Generating a gap:}
Consider the integer valued map 
$m: \delta' \rightarrow m(\delta')$, mapping 
a soundness parameter to
an upper bound on the list
size corresponding to it. Considering its
values in the interval $[\delta'/2,\delta']$, we
see that its maximum value
is at most some $M(\delta')$. 
Partition the interval
$[\delta'/2,\delta']$ into 
$R\gg M$ intervals of equal length, and towards this end consider $\delta_i = \delta' - i\frac{\delta'}{2R}$
and the intervals 
$[\delta_{i+1},\delta_i]$
for $i=0,\ldots,R-1$.
Among these intervals there are at most $M$ intervals on which
the value of $m$ changes.
Thus, we will choose $i$ randomly, 
and in fact apply the above list-decoding procedure 
for $\delta_i'$ (as opposed to 
$\delta'$). The benefit of this procedure will be that it generates a gap: with probability $1-O(M/R)$ 
we get a list of functions that
all have agreement at least 
$\delta'_{i}$ with $F_T$, 
and all functions with agreement 
$\delta'_{i+1}$ are quite close
to at least one of the functions
in the list. Indeed, the idea is 
that of it wasn't the case, 
then the list size 
would exceed beyond $M$
for $i=R-1$.

Thus, after this step 
for $1-O(M/R)$ of the good 
$D$'s we get a list
of functions that have 
at least $\delta'' := \delta'_i$
agreement with $F_D$, 
and any function
that has agreement at least
$\delta''' := \delta'_{i+1}<\delta''$ with $F_D$ is close to some function
in the list. 
We refer to such $D$'s as 
very good henceforth, 
and define the list $L[D]$
to be the list we created 
for $D$. With additional work, one may
guarantee that any two functions
in $L[D]$ are somewhat far from
each other. More 
precisely, we are able to 
guarantee that they differ
on at least $\Omega_{\delta}(1)$ fraction of the points 
in $D$; the key point here 
is that the distance between 
the functions in the list
exceeds the closeness parameter
any function with agreement 
at least $\delta'''$ has.

\vspace{1ex}
\textbf{Consistency of the Local Lists:} The steps described 
above for 
faces of size $d$ can 
be applied also for
faces of size $d/2$. 
Indeed, we define
the notion of good and 
very good faces there 
as well. Thus, we now
have lists $L[D]$ 
and $L[P]$ for each very
good $D\in X(d)$ and $P\in X(d/2)$.

Consider sampling $D\sim \mu_d$
and then $P\subseteq D$ of size 
$d/2$ uniformly. Looking at 
the list of $D$, $L[D] = (f_{1,D},\ldots,f_{m,D})$
naturally gives the restricted
functions 
$f_{i,D}|_{P}\colon P\to\{0,1\}^{d/2}$ as candidate functions
for the list of $P$. 
Indeed, we show that with 
probability $1-o(1)$ it is 
the case that each one
of these still has agreement
at least $\delta'' - o(1)$
with $F_{P}$. Furthermore, we 
know that if a function $f$ 
is far from all $f_{i,D}$, 
then it has agreement at most 
$\delta'''$ with $F_D$, and 
one can show that with 
probability $1-o(1)$ it 
will be the case that $f|_{P}$
has agreement at most $\delta'''+o(1)$ with $F_D$. 
However, the number of these $f$'s is too large for us to use
the union bound.

To circumvent this issue, we 
formulate the problem as a 
Max-CSP problem and think of 
the CSP problem on $P$ as a random sub-instance on the CSP
problem over $D$. Appealing to
results about random sub-instances
of dense CSPs~\cite{AlonVKK02},\footnote{These in return,
rely on combinatorial ideas 
revolving around (weak) regularity
lemmas.} we are 
able to avoid the union bound and show 
that with probability $1-o(1)$, 
any $f$ which is far 
from all $f_{i,D}$ has agreement
at most $\delta'''+o(1)$ in $P$.

We conclude that with probability $1-o(1)$, 
the projection of the list 
$L[D]$ to $P$ constitutes a list
of size $m$ functions 
that each has agreement
at least $\delta''-o(1)$
with $F_P$ and any function
far from it has agreement 
at most $\delta'''+o(1) < \delta''-o(1)$. This implies
that we may find a natural
$1$-to-$1$ correspondence between $L[D]$ and $L[P]$: 
we pair functions that are closest in Hamming distance. 
Among other things, this 
asserts that the list 
sizes of $D$ and $P$ are
the same with probability 
$1-o(1)$.

In other words, the lists
$L[D]$ pass the following
list agreement tester 
with probability $1-o(1)$:
\begin{mdframed}
\begin{enumerate}
    \item Sample $P\sim \mu_{d/2}$.
    \item Sample 
    $D, D'\supseteq P$
    independently according
    to $\mu_d$, and check that
    $L[D]|_{P} = L[D']|_{P}$.
\end{enumerate}
\end{mdframed}
In the next part of the
argument we prove that the list agreement test above is sound.

\subsubsection{List Agreement Testing using Coboundary Expansion}
\textbf{Designing the Unique Games instance and proving triangle consistency:}
With the downwards consistency
step done, one is naturally 
led to consider Unique-Games
instances over a graph on $X(d)$ similar to the one defined earlier 
in the introduction. 
Namely, take $D$ and $D'$
that intersect on a $d/2$ 
face, call it $P$, one has 
that marginally the distribution
of $(D,P)$ and $(D',P)$ is 
as in the
downwards consistency
step, and we get a $1$-to-$1$
correspondence between lists
of $D$, $P$ and then to $D'$
again. Composing these correspondences, we get a
$1$-to-$1$ correspondence 
between the lists of $D$ and $D'$.

The issue with the UG instance over $X(d)$ is that we only know the UG coboundary property to hold up to some level $r$, which is substantially smaller than $k$ and in particular much smaller than $d$ (e.g., it could 
just be some function of $\delta$ in some of 
our results). 

To remedy this situation, 
we show that one may ``project''
this UG instance to 
a UG instance on $G_r[X]$, while 
retaining the constraint structure. To be more explicit,
suppose that we sample 
$D\sim \mu_d$, then $R\subseteq D$ of size $r$ uniformly 
and inspect the projection of the list $L[D]$ to $R$, i.e.\ $L[D]|_{R} = (f_{1,D}|_{R},\ldots, f_{m,D}|_{R})$. First, as the pairwise distance between the functions $f_{i,D}$'s is sufficiently large (and in particular larger 
than $d/r$) the projected functions in $L[D]|_{R}$ remain distinct. 
Second, we show that for 
a typical $R$, there exists 
a list $L[R]$ such that $L[D]|_{R} = L[R]$ for 
almost all $D\supseteq R$. 
\footnote{Intuitively, the reason for that is that the lists $L[D]|_{R}$ and $L[D']|_{R}$ 
agree with probability close to 
$1$ when $D$ and $D'$ intersect 
in size $d/2$, and this 
transition operator has second eigenvalue bounded away from $1$.} This gives lists on the 
$r$-faces, and one now has to 
inherit the constraint structure
from $X(d)$. To do that, note
that one may also define $L[R']$
for faces $2r$ of size $R'$
and argue in the same way. Thus,
constraints over the graph $G_r[X]$ follow naturally as: sample a $2r$-face $R'$ and partition it randomly as $R_1\cup R_2$ where 
$|R_1| = r = |R_2|$. 
Note that with probability $1-o(1)$ there is a $1$-to-$1$ correspondence between 
$L[R_1]$ and $L[R']$ 
as well as $L[R_2]$ and $L[R']$,
and thus one gets a $1$-to-$1$ correspondence between $L[R_1]$ 
and $L[R_2]$.

\vspace{1ex}
\textbf{Applying UG coboundary expansion:}
With the UG instance
over $G_r[X]$ defined, strong triangle consistency is easily proved 
by using ideas 
that are similar to
the construction of the constraints (looking at $3r$-faces instead of $2r$-faces). 
Thus, we may appeal to the UG-coboundary expansion of the complex $X$ and conclude that 
the constraints $\pi(u,v)$ on 
$G_r[X]$ have a form as specified in Definition~\ref{def:ug_coboundary_expander}. 

As shown earlier in the introduction, 
a UG instance 
with this type of constraints admits a solution satisfying 
almost of its constraints. Namely, we may find a map $\ell\colon X(r)\to [m]$ that satisfies 
at least $1-o(1)$ fraction
of the constraints on $G_r[X]$. 

\vspace{1ex}
\textbf{Concluding the global
structure:}
The next step is to consider 
the assignment $F'\colon X(r)\to\{0,1\}^r$ defined as $F'[R] = L[R]_{\ell(T)}$, 
and note that $F'$  passes 
the direct product tester
with respect to $X$ with probability close to $1$!
Indeed, this can be shown as an easy corollary of the 
fact that the labeling $\ell$
satisfies almost all of the
constraints of $G_r[X]$. 
This means that we are precisely in the setting 
of Dinur and Kaufman~\cite{DinurK17}, and
applying the result from there one concludes that 
$F'$ has global structure, 
namely that there is 
a function $f\colon X(1)\to \{0,1\}$
such that
\[
\Pr_{R\sim \mu_r}[\Delta(F'[R], f|_{R})=o(1)]\geq 1-o(1).
\]
Looking at a random 
$D\sim \mu_d$ and 
a uniformly chosen 
$R\subseteq D$, we 
see that there is a 
$1$-to-$1$ correspondence
between the list of $R$
and the list of $D$. 
As the labeling $\ell$
chooses a function from
the list of $R$ for each 
$R$ in a way that satisfies
almost all of the constraints, we conclude
that for a typical $D$
there is $\ell(D)$ such 
that the function $L[D]_{\ell(D)}$ 
agrees with almost all of 
the $r$-faces $R\subseteq T$.
Using this fact, one 
quickly concludes 
that $\Delta(f|_D, L[D]_{\ell(D)}) = o(1)$ 
for a typical $D$. Thus, sampling a $k$-face $K\subseteq D$
one has that 
$L[D]_{\ell(D)}|_{K}\equiv F[D]$ with probability at
least $\delta'$, and by 
Chernoff's inequality
\[
|\{i\in K~|~L[D]_{\ell(D)}(i)\neq f(i)\}| = o(k)
\]
with probability $1-o(1)$, 
and we conclude 
that with probability 
$\Omega(\delta')$ 
we have that $\Delta(F[K], f|_{K}) = o(1)$.

\section{Preliminaries}
In this section we present 
a few standard notations 
as well as tools
that will be used throughout.
\paragraph{Notation:} 
Given a string $x\in \{0,1\}^n$ and a subset $A\subseteq [n]$, we 
denote by $x_A$ the substring
of $x$ corresponding to keeping only the symbols 
in the coordinates of $A$.
Given two strings $x,y \in \{0,1\}^n$, we denote by 
$\dist(x,y)$ the fractional Hamming distance between $x$ and $y$, and given a set $A \subseteq [n]$, we define $\dist_A(f,g) = \dist(x_A,y_A)$. 
We use the notation
$x\neq_{\leq \eta} y$ to
denote that $\dist(x,y) \leq \eta$.
Given a list $L$ of strings in $\{0,1\}^n$ we say that the distance $\eta$ if all distinct $x,y\in L$ 
we have have $\dist(x,y) \geq \eta$.

We use standard big-$O$ notations: we denote $A = O(B)$ 
or $A\lll B$ if $A\leq c\cdot B$ for some absolute constant
$c>0$. Similarly, we denote $A = \Omega(B)$ or $A\ggg B$ 
if $A\geq c B$ for some absolute constant $c>0$. We also 
denote $k\ll d$ to denote the fact that $d$ is taken to 
be sufficiently large compared to any function of $k$.

Whenever we have a $d$-dimensional simplicial complex $X$ and $1\leq k\leq d$, we denote by $A \sim X(k)$ a sample according to 
the distribution measure $\mu_k$ over $X(k)$ (as defined in the introduction). We use $B \subset_t A$ to denote that $B$ is a uniform $t$-sized subset of $A$. 
Similarly, for $B$ of size $t$, when we write $A \supset_k B$ we mean that $A$
is distributed according to $A\sim X(k)$ conditioned on containing $B$. 


\subsection{Concentration Bounds}
We will need the following version of Chernoff's inequality:
\begin{theorem}\label{thm:chernoff}
Suppose $X_i$ are independent random variables taking values in $\{0,1\}$ and $X$ denotes their sum. If $\E[\sum X_i] = \mu$ then,
\[\Pr[|X - \mu| > \delta\mu] \leq \exp(-\delta^2\mu), ~~~\text{ for }\delta \in (0,1),\]    
\[\Pr[X > (1+\delta)\mu] \leq \exp(-\delta\mu), ~~~\text{ for }\delta \geq 1.\]    
\end{theorem}





\subsection{Constraint Satisfaction Problems: Value and Random Sub-Instances}

Our argument will make use of instances of the 
\emph{max-$k$-CSP} problem and properties of random sub-instances 
of a given instance.
\begin{definition}
  Let $k\in\mathbb{N}$. An instance $\Psi$ of 
  (Boolean) max-$k$-CSP consists 
  of a set of variables 
  $\{x_i\}_{i\in I}$, 
  along with constraints,
  each one of the form 
  $P(x_{i_1},\ldots,x_{i_k}) = 1$
  for some $P\colon \{0,1\}^k \to \{0,1\}$.
\end{definition}
Given an instance $\Psi$
of max-$k$-CSP, the goal 
is to find an assignment to
the variables of $\Psi$
that satisfies as many of the 
constraints as possible. 
We refer to this maximum 
fraction as the value of 
$\Psi$, and denote it by 
$\val(\Psi)$.

Given an instance of 
max-$k$-CSP $\Psi$ with
variables $V$ and a 
subset of variables $Q\subseteq V$, we define the induced
instance on $Q$, $\Psi|_{Q}$,
to be the instance of max-$k$-CSP resulting from $\Psi$ by keeping only the variables of $Q$, and only the constraints 
of $\Psi$ that involve only variables from $Q$.

Of special interest to us will
be dense instances of max-$k$-CSP. In fact, we will be concerned with instances wherein 
there is a constraint for every subset of size $k$ of the variables, and where the number
of variables $d$ is much larger than $k$. Given such an instance 
$\Psi$, we will want to consider random induced sub-instances of $\Psi$ and their value. With this regard, the following result from~\cite{AlonVKK02} 
asserts that the value of the
random sub-instance remains 
roughly the same.
\begin{theorem}\label{thm:max-csp}
For all $\gamma, \tau \in (0,1)$, $k \in \N$ and $d \geq \poly(k/\tau)\exp(1/\gamma^2)$,\footnote{For general dense $k$-CSPs they incur a $\exp(2^{2^k})$ dependence in $d$, which comes from the fact that there can be $2^{2^k}$ constraints in $\Psi$ that can be satisfied by setting a particular set of variables $I \subset_k [d]$ to a fixed assignment $z\in \{0,1\}^k$. In our setting, there could only be one constraint that gets satisfied by such fixing, and therefore we do not incur this triple-exponential dependence on $k$ (though this wouldn't matter for us in any case).} consider a $k$-CSP with ${n \choose k}$ constraints that each depend on a unique $k$-set of variables. If $q \geq \poly(k/\tau\gamma)$ then: 
\[\Pr_{Q \subset_q [d]}\big[|\val(\Psi|_Q) - \val(\Psi)| \leq \gamma\big] \geq 1-\tau.\]
\end{theorem}

We will also consider 
a special type of max-$2$-CSP, called Unique-Games. Unique-Games
have already been considered 
in the introduction, and we
will make use of them later on in our argument.
\begin{definition}\label{def:unique-games}
An instance of Unique-Games $\Psi = (G, \Pi)$ consists of a graph $G = (V,E)$, a finite alphabet $\Sigma$ and a collection of permutations, $\Pi = \{\pi_{u,v}\}_{(u,v)\in E}$, one for each edge in $G$. 
The goal in the Unique-Games problem is to find an assignment $A\colon V\to\Sigma$ that satisfies the maximum possible number of edges, that is, $A(u) =\pi(u,v)A(v)$. We define the value of the instance $\Psi$ as:
\[
{\sf val}(\Psi) = \max_{A\colon V\to\Sigma}\frac{\#\{e~|~\text{A satisfies $e$}\}}{|E|}.
\]
\end{definition}

\subsection{Properties of Expanders}
We need the following well known version of the expander
mixing lemma for bipartite graphs.
\begin{lemma}\label{lem:bip-eml}
Let $G = (U,V,E)$ be a bipartite graph in 
which the second singular value of the normalized adjacency matrix is at most $\lambda$. Then for all $A \subset U$ and $B \subset V$ we have that
\[
\left|\Pr_{(u, v) \in E}[u \in A, v \in B] -\mu(A)\mu(B)\right| \leq \lambda\sqrt{\mu(A)(1-\mu(A))\mu(B)(1-\mu(B))}.\]    
\end{lemma}

We also use the following standard sampling 
property of bipartite expanders.
\begin{lemma}
\label{lem:sampling}
Let $G = (U,V,E)$ be a bipartite graph with second singular value at most $\lambda$. If $B \subset U$ has 
$\Pr[B] = \delta$, then the set 
$T = \left\{v \in V \mid \Pr_{u\sim E|_v}[u \in B] > \eps+\delta\right\}$ has $\Pr[T]\leq \lambda^2\delta/\eps^2$.
\end{lemma}

\subsection{Properties of Local Spectral Expanders}
Recall that we associated with each $d$-dimensional simplicial 
complex $X$ a sequence of measures $\{\mu_k\}_{1\leq k\leq d}$, where 
$\mu_k$ is a probability measure over $X(k)$. Note that for all $0 \leq t \leq r \leq d$, a sample according to $\mu_t$ can be drawn by first sampling $R \sim \mu_r$, and then sampling $T\subseteq_{t} R$ uniformly. The converse is also true: a sample from $\mu_r$ can be drawn by first sampling $T \sim \mu_t$, and then sampling $R$ from $\mu_r$ conditioned on containing $T$. These observations 
give rise to the standard ``up'' and ``down'' operators, which we present next. We only mention a few of their 
properties that are necessary for our arguments, and refer the reader to~\cite{dikstein2018boolean} for a more
comprehensive exposition.

\begin{definition}
The operator $U_i^{i+1}$ is a map from $L_2(X(i); \mu_i)$ to $L_2(X(i+1); \mu_{i+1})$ defined as
\[
U_i^{i+1}f(u) 
= 
\E_{v \subset_i u}\big[f(v)\big]
\]
for all $u\in X(i+1)$. For $j\geq k+1$, we define $U_k^j$ via composition of up operators: $U_k^j = U_{j-1}^j \circ \ldots \circ U_k^{k+1}$.
\end{definition}

\begin{definition}
The operator $D_i^{i+1}$ is a map from $L_2(X(i+1); \mu_{i+1})$ to $L_2(X(i); \mu_i)$ defined as
\[
D_i^{i+1}f(u) = \E_{v \supseteq_{i+1} u}\big[f(v)\big]
\]
for all $u \in X(i)$.
For $j\geq k+1$, we define $D_k^j$ via composition of down operators: $D_k^j = D_{k}^{k+1} \circ \ldots \circ D^j_{j-1}$.
\end{definition}

Abusing notations, we use the notations $U^j_k, D^j_k$ 
to denote the operators, as well as the real valued 
matrices associated with them. A key property of
the down and up operators is that they are adjoint:

\begin{claim}
For all $k \leq j \leq d$, $U_k^{j}$ and $D^{j}_k$ are adjoint operators: for all functions $f\colon X(k)\to\mathbb{R}$ and $g\colon X(j)\to\mathbb{R}$ it holds that $\ip{U_k^{j}f,g} = \ip{f,D^{j}_kg}$. \end{claim}

We need the following lemma regarding the second eigenvalue of the down-up walks $U^j_{k}D^j_{k}$ on $X(j)$ ($j \geq k$), that can be found in~\cite{AlevL20}.

\begin{lemma}\label{lem:spectral_gap_of_graphs_from_HDX}
Let $(X, \mu)$ be a $d$-dimensional $\gamma$ one-sided local spectral expander. For all $i \leq d$ and $\alpha \in (1/i, 1)$, the largest singular value of $U^i_{\alpha i}$ and $D^i_{\alpha i}$ is at most $\sqrt{\alpha}+\poly(i)\gamma$. Thus the down-up random walk $U^i_{\alpha i}D^i_{\alpha i}$ on $X(i)$ has second largest singular value at most $\alpha + \poly(i)\gamma$. 
\end{lemma}

We will use the following theorem from \cite{DinurK17} that shows that spectral HDXs support an agreement test in the 99\% regime.\footnote{Strictly speaking, the direct product tester that Dinur and Kaufman analyze a bit different. The formulation we give is a bit more convenient for us to apply, and the proof in~\cite{DinurK17} applies to 
that setting in exactly the same way to give 
the soundness guarantee as stated in Theorem~\ref{thm:dinur-kaufman}.}

\begin{theorem}[\cite{DinurK17}]\label{thm:dinur-kaufman}
Let $X$ be a $d$-dimensional $\lambda$ one-sided local spectral expander and let $t^2 < d$, $\lambda < 1/d$ and $\eps > \eps_0(t,\lambda)$. Let $F: X(t) \rightarrow \{0,1\}^t$ such that:
\[\Pr_{\substack{D\sim \mu_{d}\\ Q\subseteq_{t/2} D\\ 
Q\subseteq B, B'\subseteq_{t} D}}[F[B]|_Q = F[B']|_Q] \geq 1-\eps.\]
Then, there exists a function $G: X(1) \rightarrow \{0,1\}$ such that,
\[\Pr_{B \sim X(t)}[F[B] = G[B]] \geq 1-O(\eps).\]
\end{theorem}

\section{UG Coboundary is Necessary}\label{sec:necessary}
In this section we prove 
the ``necessary'' part 
of Theorem~\ref{thm:HDX_dp_weak}, 
stated formally below.
\begin{theorem}\label{thm:necessary-formal}
For all $c,\eta>0$, 
$m, r\in\mathbb{N}$, there exist $d, k \in \N$, $\xi, \gamma>0$ such that the following holds. If a simplicial complex $X$ is a $d$-dimensional $\gamma$-spectral expander and is not $(m,r,\xi,c)$ coboundary expander, then there exists 
$F\colon X(k)\to\{0,1\}^k$ 
that passes the $(k,\sqrt{k})$
direct product tester with probability at least $\frac{1}{m}-O(\sqrt{\xi})$, 
and yet for all $f\colon X(1)\to\{0,1\}$
we have that
\[
\Pr_{A\sim \mu_k}
\big[\Delta(F[A], f|_A) \geq \Omega_{m,r,c}(1)\big]\leq \eta.
\]
\end{theorem}

Henceforth in this section, we assume that $X$ is \textbf{not} a weak $(m,r,\xi,c)$ UG coboundary expander. That is, for some $t \leq r$ there exists a collection of lists $\{L'(R)\}_{R \in X(t)}$, $\{L'(T)\}_{T \in X(3t)}$ and a set of permutations $\{\pi'(S)\}_{S \in X(2t)}$ such that,
\[
\Pr_{\substack{T \sim \mu_{3t} \\ A \cup B \cup C = T}}[L'(T) = L'(A) \circ \pi'(A,B)L'(B) \circ \pi'(B,C)L'(C)] \geq 1-\xi,\]
and yet for all $P\colon X(t)\to S_m$ it holds that 
$\pi'(A,B) = P(A)P(B)^{-1}$ with probability at most 
$1-c$ over the choice of $A\cup B\sim \mu_{2t}$. 
We refer to such $P$ as an $S_m$-solution to the 
Unique-Games instance $\Psi' = (G_t[X], \Pi')$.
This step is summarized in the following lemma.

\subsection{Preprocessing}
The first step of the proof is to show that we can convert $\Psi'$ as above to a (possibly different) UG instance $\Psi$ over $G_t[X]$ that doesn't 
have any $[m]$-valued solution satisfying more 
than $(1-c/m)$ weight of the constraints in $G_t[X]$.

\begin{lemma}\label{lem:preprocess_necessary}
There exists a collection of lists $\cL = \{L(R)\}_{R \in X(t)} \cup \{L(T)\}_{T \in X(3t)}$ such that:
\[\Pr_{\substack{T \sim X(3t) \\ A \cup B \cup C \sim T}}[L(T) = L(A) \circ \pi(A,B)L(B) \circ \pi(B,C)L(C)] \geq 1-\xi,\]
and a set of permutations $\Pi = \{\pi(S)\}_{S \in X(2t)}$ such that the corresponding UG instance $\Psi = (G_t[X], \Pi)$ has no solution $I: G_t[X] \rightarrow [m]$ with $\val(I) \geq 1-c/m$.
\end{lemma}
\begin{proof}
Start with the collection of lists $\cL' = \{L'(R)\}_{R \in X(t)} \cup \{L'(T)\}_{T \in X(3t)}$ and permutations $\Pi' = \{\pi'(S)\}_{S \in X(2t)}$. If for all $A: X(t) \rightarrow [m]$, $\val(A) \leq 1-c/m$ we are done, therefore let us assume that
there exists such a solution $A\colon X(t)\to [m]$ satisfying $1-c/m$ fraction of the edges. In this case we will do the following:
\begin{enumerate}
\item Remove $A(u)$ from the 
list $L[u]$ to get a new list of size $m-1$.
\item For all edges $(u,v)$, change $\pi'(u,v)$ to a permutation $\pi''(u,v) \in S_{m-1}$, by ``removing''  $(A(u),A(v))$, i.e. $\pi''(u,v)(i) = \pi'(u,v)(i)$ for all $i \neq A(u)$. 
\end{enumerate}
Continue this procedure until there is no such assignment, and let $(\cL,\Pi)$ be the set of lists and permutations when this ends. We argue that these lists are non-empty; in fact they will be of size at least $2$. Given that, the fact
that $(\cL,\Pi)$ is strongly triangle consistent is obvious.

To see that the lists are non-empty, first notice that the process must terminate
within $m-2$ steps. Indeed, otherwise for each step 
$1\leq i\leq m-1$ define $A_i(R)$ to be the 
assignment given to vertex $R$ in the $i$th iteration, 
and define $A_m(R)$ to be the last assignment left 
in the list of $R$ after all iterations are done. 
Thus, defining the permutation valued assignment 
$P\colon X(t)\to S_m$ as $P(R)$ being the unique permutation
for which $P(u)L'(R) = (g_1(R),\ldots,g_m(R))$, we 
see that $\pi'(U,V) = P(U)P(V)^{-1}$ unless the edge
$(U,V)$ was violated by at least one of the assignments
$A_1,\ldots,A_m$. By the union bound, the total weight
of edges violated by at least one of $A_1,\ldots,A_m$
is at most $m\cdot \frac{c}{m} = c$, hence $\pi'(U,V) = P(U)P(V)^{-1}$ for at least weight $1-c$ of the edges, 
in contradiction.

\end{proof}

Henceforth, we fix a Unique-Games instance $(\cL,\Pi)$ 
as in Lemma~\ref{lem:preprocess_necessary}, 
and denote $c' = c/m$.

\subsection{Lifting the Lists}
The next step in the proof is to lift the lists $\cL$ to lists on $X(k)$, and to do so we use Kneser graphs. 
For every $A \in X(k)$, denote by $K(A,t)$ the Kneser graph on $A$, whose vertex set is $\binom{A}{t}$, and two 
$t$-sets $T, T'$ are adjacent if $T\cap T' = \emptyset$.  Note that this is exactly the subgraph of $G_t[X]$ induced by the $t$-faces contained inside $A$. We say a $k$-face $A$ triangle-consistent if all the triangles in $K(A,t)$ are consistent with respect to $(\cL,\pi)$, and denote by $\cK(k) \subseteq X(k)$ the collection of triangle-consistent $k$-faces.

\begin{claim}\label{claim:triangle-consistent-k-faces}
If $\xi \leq \exp(-t\log k)$, then $\mu_k(\cK(k)) \geq 1-\sqrt{\xi}$.
\end{claim}
\begin{proof}
We know that a triangle picked as $T \sim X(3t)$ and $a \cup b \cup c = T$ is inconsistent with probability at most $\xi$. This is the same distribution as picking a $k$-face from $K \sim X(k)$, $T \subset_{3t} K$ and $a \cup b \cup c \sim T$. Thus by linearity of expectation the number of inconsistent triangles in $K$ is at most $k^{3t}\xi$, which implies that the probability that $K \sim \mu_k$ contains at least one inconsistent triangle is at most $k^{3t}\xi\leq \sqrt{\xi}$. 
\end{proof}

We now show how to lift the lists to $k$-faces in $\cK(k)$. To do so we will need the following simple claim about Kneser graphs. We will show that triangle consistent UG instances on Kneser graphs are satisfiable, i.e. any set of permutations on the edges of $K([k],t)$ that is consistent on triangles has value $1$.

\begin{claim}
\label{claim:kneser-coboundary}
For all $t \geq 1, k \geq 5t$ the following holds. Let $\Phi = (K([k], t), \Pi)$ be a UG instance on $K([k],t)$ with alphabet $[m]$, in which all triangles are consistent. Then $\val(\Phi) = 1$ and furthermore there exist $m$ distinct satisfiable assignments.
\end{claim}

\begin{proof}
Let $K$ denote $K([k],t)$. We will show that all cycles in $K$ are consistent, i.e. for any $\ell$-cycle ($\ell \geq 3$) $C = (v_1,\ldots,v_\ell) \in K$ it holds that $\pi(v_1,v_2)\cdot \ldots \cdot \pi(v_\ell,v_1) = \text{id}$. We can show this by induction on the length of the cycle. By assumption $3$-cycles are consistent. Given that all $\ell-1$-cycles are consistent we can prove that all $\ell$-cycles in $K$ are consistent. Given cycle $C = (v_1,\ldots,v_\ell)$, consider any vertex $u$ such that $u$ has an edge to $v_1,v_2,v_3$ and $v_4$. Such a vertex exists because $|\cup_{i=1}^4 v_i| \leq 4t$ and $k \geq 5t$. Since $u$ is connected to $v_1,\ldots,v_4$ we get that: 
\begin{align*}
\pi(v_1,v_2)\pi(v_2,v_3)\pi(v_3,v_4) &= \pi(v_1,u)\pi(u,v_2)\pi(v_2,u)\pi(u,v_3)\pi(v_3,u)\pi(u,v_4) \\
&=\pi(v_1,u)\pi(u,v_4)
\end{align*}
where we used the fact that the triangles $(v_1,u,v_2),\ldots, (u,v_3,v_4)$ are consistent. Therefore we have reduced the task of showing that the cycle $C$ is consistent, to the task of showing that the cycle $C' = (v_1,u,v_4,\ldots,v_\ell)$ of length $\ell-1$ is consistent; this last assertion is true by the induction hypothesis.

Given that all cycles in this graph are consistent, we can assign the first vertex in the graph to be some element from $[m]$, and find an assignment to the rest of the vertices via propagation. Note that we do not run into contradictions because all cycles are consistent. It is easy to see that such an assignment will satisfy all the edges of the graph, and we will get a set of $m$ distinct satisfiable assignments in this way.
\end{proof}

For convenience of notation, we will say that an assignment $(F[U], F[V])$ (for $U,V \in X(t), U \cap V = \emptyset$) satisfies the edge $(U,V)$ if $F[U] = L_i(U)$ and $F[V] = L_j[V]$ for some $i,j \in [m]$ and $j = \pi(i,j) i$.

\begin{lemma}
\label{lem:lifting-lists}
If $k \geq 5t$, then for every $A \in \cK(k)$, there exists a unique list $L(A)$ such that $L(A)|_R = L(R)$ for all $R \subset_t A$. Furthermore, for every assignment $F \in L(A)$ and edge $(U,V) \in K(A,t)$, $(F|_U, F|_V)$ satisfies the edge $(U,V)$.
\end{lemma}

\begin{proof}
Fix an $A \in \cK(k)$. Since $A$ is triangle-consistent we know that the permutations on its edges are triangle consistent, i.e. for all triangles $(a,b,c)$ in $K(A,r)$, $\pi(a,b)\pi(b,c)\pi(c,a) = \text{id}$. In Claim~\ref{claim:kneser-coboundary} we showed that all triangle consistent UG instances on $K(A,t)$ (for $|A| \geq 5t$) have $m$ global solutions that satisfy all the edges. Let these assignments by $S_1,\ldots,S_m$ that map the vertices of $K(A,t)$ to $[m]$.

For each $S_i$ construct the following assignment: $B_i: K(A,t)\rightarrow \{0,1\}^t$, defined as $B_i(T) = L_{S_i(T)}(T)$. Since $S_i$ satisfies the permutations on every edge, we know that, for all triangles $(T_1,T_2,T_3)$ in $K(A,t)$, $B_i(T_1)\circ B_i(T_2)\circ B_i(T_3) \in L(T_1 \cup T_2 \cup T_3)$. In particular by the connectivity of the graph we can check that in fact $B_i(T_1)\circ B_i(T_2)\circ B_i(T_3) = B_i(T'_1)\circ B_i(T'_2)\circ B_i(T'_3)$ for two different splittings of the $3t$-sized set $T_1 \cup T_2 \cup T_3$. This immediately implies that there is a unique assignment $C_i \in \{0,1\}^k$ to $A$ such that $C_i|_T = B_i(T)$ for all $T \subset_t A$. By definition $(B_i(U), B_i(V))$ satisfies the edge $(U,V) \in K(A,t)$, therefore so does the assignment $(C_i|_U, C_i|_{V})$.

Putting all the assignments $C_1,\ldots, C_m$ in a list we get $L(A)$ that satisfies $L(A)|_T = L(T)$ for all $T \in K(A,t)$.
\end{proof}

\subsection{Constructing the Assignment and Its Soundness}
Fix the lists $\{L(A)\}_{A\in X(k)}$ as in Lemma~\ref{lem:lifting-lists}. 
Below, we show how to construct an assignment $F$ 
that passes the direct product test with probability
$\Omega(1/m)$. In fact, it passes the related list-agreement
test with probability close to $1$.

\begin{lemma}
\label{lem:agr-test-passes}
For all $(5t)^2 \leq k \leq d$, there exists a function $F: X(k) \rightarrow \{0,1\}^k$ that satisfies $F[A] \in L(A)$, for all $A \in X(k)$, such that:
\[
\Pr_{\substack{D \sim \mu_d \\ B\subseteq_{\sqrt{k}} D \\ B \subset A, A' \subset_k D}}\big[F(A)|_{B} = F(A')|_B\big] \geq 
\frac{1}{m}\left(1-O(\sqrt{\xi})\right).
\]
\end{lemma}
\begin{proof}
We will first show that the list-agreement test passes with probability $1-O(\sqrt{\xi})$. Let $D, B, A, A'$ be generated as in the statement of the lemma, 
and denote by $\cD$ the distribution of $(B, A, A')$.
Note that the marginal distribution on $B \sim \cD$ is $\mu_{\sqrt{k}}$ and on $A$ and $A' \sim \cD$ is $\mu_k$. Therefore, by the union bound we get that with probability $1-3\sqrt{\xi}$ we have $B \in \cK(\sqrt{k})$ and $A, A' \in \cK(k)$ (Claim~\ref{claim:triangle-consistent-k-faces}). Consider such a triple $(B,A,A')$ as good.

Fix a good triple $(B,A,A')$. By virtue of being in $\cK(\sqrt{k})$, $L(B)|_T = L(T)$ for all $T \subset_t B$ and the same holds for $L(A), L(A')$, which in particular implies that $L(A)|_T = L(B)|_T = L(A')|_T$ for all $T \subset_t B$. There can only be one list on $B$ that satisfies $L(B)|_T = L(T)$ for all $T \subset_t B$. Therefore we get that $L(A)|_B = L(A')|_B$ for all good triples $(B,A,A')$, hence,
\[\Pr_{(B,A,A') \sim \cD}[L(A)|_{B} = L(A')|_B] \geq 1-O(\sqrt{\xi}).\]    
Thus, choosing $F[A]$ to be a random element from the list $L(A)$ we get that
\[\E_F\Big[\Pr_{(B,A,A')\sim \cD}[F(A)|_{B} = F(A')|_B]\Big] \geq \frac{1}{m}\left(1-O(\sqrt{\xi})\right).\]   
Therefore we can pick an assignment $F$ such that the above holds.
\end{proof}

\subsection{No Global Structure}
We finish by showing that for $F$ as constructed
in Lemma~\ref{lem:agr-test-passes}, there is no global function $f\colon X(1)\to \{0,1\}$ that has
significant agreement with it. In fact, we show for any $F$ on $X(k)$ such that $F[A] \in L(A)$, there is no global function on $X(1)$ that agrees with $F$ on a large fraction of $A$'s. 

\begin{lemma}\label{lem:no-global-structure}
Suppose that the UG instance $\Psi = (G_t[X], \Pi)$ has no solution of value $\geq 1-c'$. Then, for all functions $F: X(k) \rightarrow \{0,1\}^k$ where $F[A] \in L(A)$ for all $A$, and for any $G: X(1) \rightarrow \{0,1\}$ 
\[
\Pr_{A \sim \mu_k}
\Big[\Delta(G[A],F[A]) \leq O\left(\frac{c'}{t}\right)\Big] \leq O\left(\frac{t}{kc'^2}\right)+\sqrt{\xi}.
\]
\end{lemma}
\begin{proof}
Fix such a function $F$ and suppose for contradiction that there is $G: X(1) \rightarrow \{0,1\}$ such that
\begin{equation}\label{eq:global-assn-exists}
\Pr_{A \sim \mu}[\Delta(G[A],F[A]) \leq \eps] \geq \alpha + \sqrt{\xi}    
\end{equation}
for $\eps, \alpha$ to be determined later. In this case we will construct an assignment to the UG instance $\Psi$ with large value, which will be a contradiction.

Consider the following assignment $I: G_t[X] \rightarrow [m]$. For each $T \in X(t)$ let $I[T] = i$ if $G[T] = L_i(T)$, else assign $I[T]$ arbitrarily. Thus
\[
\val(I) \geq \Pr_{\substack{T \sim \mu_{2t}\\U \cup V \sim T}}\big[(G[U],G[V]) \text{ satisfies the edge }(U,V)\big],
\]
where $A \cup B$ is a random split of $T$ into two sets of size $t$ each. We show that $\val(I)$ is close to $1$.

Let $\cK'$ be the set of $k$-faces $A$ where $\Delta(G[A],F[A]) \leq \eps$ and $A\in \cK(k)$. 
By Claim~\ref{claim:triangle-consistent-k-faces} 
it holds that $\mu_k(\cK)\geq 1-\sqrt{\xi}$, 
and combining with \eqref{eq:global-assn-exists} yields that $\mu_k(\cK')\geq \alpha$.

Fix $A \in \cK'$. For $T\subseteq_{2t} A$, we get that $\Delta_T(G[A],F[A])=0$ with probability at least $1-2\eps t$, so
\begin{equation}\label{eq:projection}
\Pr_{T \subset_{2t} A}[F[A]|_T = G[T]] \geq 1-O(\eps t).   
\end{equation}
Thus, for $A \in \cK'$ define $\text{Good}(A)=\{T \in X(2t)~|~F[A]|_T = G[T]\}$, and define $\cT = \bigcup_{A\in\cK'}\text{Good}(A)$.
First note that for all $T \in \cT$ and every splitting $(U, V)$ of $T$, $(G[U],G[V])$ satisfies the edge $(U,V)$. 
To see this, consider some $T \in \cT$, where $T \in \text{Good}(A)$ for $A \in \cK'$. Since $A \in \cK$ and $F[A] \in L(A)$, by Lemma~\ref{lem:lifting-lists} $(F[A]|_U, F[A]|_V)$ satisfies the edge $(U,V)$. 
Since $F$ and $G$ are equal on $T$ this immediately implies that $G$ also satisfies the edges $(U,V)$ for $U \cup V = T$. Therefore to get a bound on the value of $I$ it suffices to lower bound the measure of $\cT$. 

\paragraph{Lower bounding the measure of $\cT$:}
Consider the bipartite graph $G_{k,2t} = (X(k), X(2t), D^k_{2t})$ where the edges are weighted according to the down walk from $X(k)$ to $X(t)$. Namely, an edge in $G_{k,2t}$ is sampled by picking $K \sim \mu_k$, and then taking $T \subset_{2t} K$ uniformly. We will be interested
in counting the number of edges between $\cK'$ and $\cT$.

For $\mathcal{A}\subseteq X(k)$ and 
$\mathcal{B}\subseteq X(2t)$, we denote by $E(\mathcal{A},\mathcal{B})$ the set of edges between 
$\mathcal{A}$ and $\mathcal{B}$, and 
we denote by $\mu(E(\mathcal{A},\mathcal{B}))$ the total
weight of edges in $E(\mathcal{A}, \mathcal{B})$. Using
these notations, we have that
\begin{equation}\label{eq:necessary_1}
\mu(E(\cK',\cT)) \geq \mu_k(\cK')\Pr_{T \subset_{2t} K}[T \in \cT \mid K \in \cK'] \geq \alpha(1-O(\eps t)),
\end{equation}
where in the last inequality we used~\eqref{eq:projection}. By Lemma~\ref{lem:bip-eml} we have
\begin{equation}\label{eq:necessary_2}
    |\mu(E(\cK',\cT)) - \mu(\cK')\mu(\cT)| \leq \lambda(D^{k}_{2t})\sqrt{\mu(\cK')}
\leq O(\sqrt{t/k}),
\end{equation}
where the last inequality is by Lemma~\ref{lem:spectral_gap_of_graphs_from_HDX} and
the fact that $\gamma < 1/\poly(k)$. Combining~\eqref{eq:necessary_1} and~\eqref{eq:necessary_2} 
and simplifying gives that 
$\mu(\cT) \geq 1-O(t\eps)-\sqrt{2t/k\alpha}$ which is at least $1-O(t\eps)$ if $\alpha \geq \frac{1}{tk\eps^2}$. In 
that case, we conclude that:
\[\val(I) \geq \Pr_{\substack{T \sim X(2t)\\U \cup V \sim T}}[(G[U],G[V]) \text{ satisfies the edge }(U,V)] \geq \mu_{2t}(\cT) \geq 1-O(t\eps),\]
which is a contradiction to $\Psi$ having value at most $1-c'$ if $\eps < O(c'/t)$. It follows that 
$\alpha\leq \frac{1}{tk\eps^2}$, and the proof
is concluded by choosing $\eps = \frac{c''c'}{t}$ 
for sufficiently small $c''>0$.
\end{proof}


\begin{proof}[Proof of Theorem~\ref{thm:necessary-formal}]
The result follows immediately by combining Lemmas~\ref{lem:agr-test-passes} and~\ref{lem:no-global-structure}.
\end{proof}

\section{Proof of Theorem~\ref{thm:HDX_dp_weak}: UG Coboundary is Sufficient}\label{sec:basic_pf}
In this section, we prove
the ``sufficient'' part 
of Theorem~\ref{thm:HDX_dp_weak},
formally stated below.
\begin{theorem}\label{thm:sufficient_ugexpand_formal}
There is $c>0$ such that for all $\eps,\delta>0$ there is $\xi,\eta>0$ and $m,r\in\mathbb{N}$ such that 
for sufficiently large $k$, sufficiently large $d$
and $\gamma$ small enough function of $d$, the 
following holds. 
If a $d$-dimensional simplicial complex $X$ is a $\gamma$-spectral expander and $(m,r,\xi,c)$ weak UG coboundary expander, then the direct
product test over $X$ with respect
to sufficiently large $k$ has soundness $\delta$. Namely, if 
$F\colon X(k)\to\{0,1\}^k$ passes the
$(k,\sqrt{k})$ direct product tester with respect to 
$X$ with probability at least $\delta$, 
then there is $f\colon X(1)\to\{0,1\}$
such that 
\[
\Pr_{A\sim \mu_k}
[\Delta(F[A], f|_A)\leq \eps ]\geq \eta.
\]
\end{theorem}

\begin{remark}
We remark here that in the above theorem, we require $\delta \geq 1/\log k$, $d \geq \poly(k)\exp(1/\delta)$, $r = \exp(\poly(1/\delta))$ and $\xi = \poly(\delta)$, which is equal to $1/(\log r)^c$ for some $c \in (0,1)$. In Section~\ref{sec:appx-improved-cbdry} we improve the latter dependence to show that UG coboundary expansion of $(m,r,\exp(-o(r)),c)$ is sufficient.    
\end{remark}

We begin by setting up some
notations that will be helpful
throughout the proof. Given a global function $f:[d] \rightarrow \{0,1\}$ and a set $B \subseteq [d]$ we let $f(B)$ denote the assignment to $B$ using $f$. For a function $f: [d] \rightarrow \{0,1\}$ and an assignment $F: X(k) \rightarrow \{0,1\}^k$ we let $\Agr(f,F)$ denote the subset of $X(k)$ where $f(s) = F(s)$ and $\agr(f, F)$ denote the probability of this event under the measure $\mu_k$. Furthermore for $\nu \in (0,1)$ let $\Agr_\nu(f,F)$ denote the subset of $X(k)$ where $f(B)$ and $F(B)$ agree on $(1-\nu)$-fraction of the elements in $B$ and $\agr_{\nu}(f,F)$ denotes the probability of this event under $\mu_k$. 

\subsection{High Level Structure of the Proof}
The proof of Theorem~\ref{thm:sufficient_ugexpand_formal} follows the
outline given in the introduction. 
For convenience we break 
it into two parts, encapsulated 
in the following two lemmas.
In the first lemma we 
implement the first four steps 
in the plan and reduce 
the problem of direct 
product testing to
the problem of ``list agreement'' testing. 
In this problem, for
each $d$-face $D$ in a 
complex $X$ we have a 
list $L[D]$ of $O(1)$ functions,
and we test whether 
these lists are in 
$1$-to-$1$ correspondence
according to the up-down-up
walk on the complex. More precisely, the problem is defined as follows:
\begin{mdframed}
\begin{list-agr-test}
\label{list-agr-test-hdx}
\phantom{\\}

\noindent Input: a list $L(D)$ for each $D\in X(d)$ and a parameter $\eta\in (0,1)$.\mbox{}
\begin{enumerate}
    \item Choose random $B \sim X(d/2)$.
    \item Choose independently random $A,A' \supseteq_d B$ from $X(d)$.
    \item Accept iff both lists are non-empty and $L[A]|_{B} \neq_{<\eta} L[A']|_{B}$.
\end{enumerate}
\end{list-agr-test}
\end{mdframed}

With the list agreement problem formally defined, 
we can now state the 
lemma encapsulating 
the first few steps
in the argument, saying 
that an assignment 
that passes the direct 
product test with probability
bounded away from $1$
implies a natural list assignment passing the list agreement test with probability close to $1$.
\begin{lemma}
\label{lem:agr-to-list-agr}
For all $\delta>0$, 
for sufficiently large $k \in \N$,
$d \geq \poly(k)2^{\poly(1/\delta)}$, 
sufficiently small $\gamma$ compared to 
$d$ and $\tau = O(\delta^{68})$, the following holds. Suppose that $X$ is a $d$-dimensional simplicial complex which is a $\gamma$-spectral expander, and $F: X(k) \rightarrow \{0,1\}^k$ passes the $(k,\sqrt{k})$-agreement-test~\ref{agr-test-hdx} with probability $\delta$. Then, there exists $2^{-1/\delta^{1200}}\leq \eta'
\leq \delta^{101}$ and lists $(L[D])_{D \in X(d)}$ satisfying: 
\begin{enumerate}
\item 
\textbf{Short, non-empty lists:} With probability $1-O(\tau)$ over the choice of $D\sim X(d)$, the list $L[D]$ is non-empty and has size at most $O(1/\delta^{12})$.
\item  
\textbf{Good agreement:} For all $D\in X(d)$ and every $f \in L[D]$, we have that $\agr_\nu(f, F|_D) \geq \Omega(\delta^{12})$ for $\nu = 1/k^{\Omega(1)}$.
\item 
\textbf{Distance in the lists:} With probability at least $1-O(\tau)$ over the choice of $D\sim X(d)$, the list $L[D]$ has distance at least $\delta^{-100}\eta'$.
\end{enumerate}
Furthermore the lists above pass the List-Agreement-Test~\ref{list-agr-test-hdx} with parameter $\eta'$, with probability $1-\tau$. 
\end{lemma}

Armed with the conversion
of our assignment $F$ to 
lists that pass the list 
agreement test with probability close to 
$1$, we implement the next
three steps in the introduction. 
Namely, we show that 
if $X$ is a sufficiently 
good UG coboundary expander, 
then we can use the lists 
above to define a locally consistent instance of Unique-Games on low levels 
of the complex and apply 
UG coboundary expansion 
to deduce the existence 
of a global solution.

\begin{lemma}\label{lem:list-agr-test}
Assume there exists a collection of lists $\{L[D]\}_{D \in X(d)}$ that satisfy the premise of Lemma~\ref{lem:agr-to-list-agr}, and assume that $X$ is a $\gamma$-spectral expander for $\gamma < 1/\poly(d)$ and a weak $(O(1/\delta^{12}),t, O(\sqrt{\tau}), c)$ UG coboundary expander for 
$t=\Theta\left(\frac{\tau\delta^{12}}{\eta'}\right)$. Then 
there exists $G: X(1) \rightarrow \{0,1\}$ such that
\[
\Pr_{D \sim X(d)}\left[\Delta(G(D),L[D]) \leq \delta\right] \geq 1-O(c^{1/2} + \tau^{1/4} + \gamma).\]
Here, 
the distance between a function $G(D)$ and a list of functions $L[D]$ is the minimal distance between $G(D)$ and any function in the list.
\end{lemma}

The proof of Theorem~\ref{thm:sufficient_ugexpand_formal} now readily 
follows from the above two
lemmas.
\skipi
\textbf{Proof of Theorem~\ref{thm:sufficient_ugexpand_formal}.}
In the setting of 
Theorem~\ref{thm:sufficient_ugexpand_formal}, first
assume that $\eps = \delta$ (otherwise we lower both
of them to be the minimum of $\eps$ and $\delta)$). 
Apply
Lemma~\ref{lem:agr-to-list-agr}
and then
Lemma~\ref{lem:list-agr-test}
to conclude that there is 
a function $G\colon X(1)\to\{0,1\}$ such that
\[
\Pr_{D \sim \mu_d}\left[\Delta(G(D),L[D]) \leq \eps\right] \geq \frac{1}{2}.
\]
Fix $D\in X(d)$ such that
$\Delta(G(D),L[D])\leq \eps$, 
and let $f\in L[D]$ 
be such that $\Delta(G(D), f)\leq \eps$. 
Sampling $A\subseteq_{k} D$, 
we have by the ``good agreement'' property of the 
list that $F[A] \neq_{<\nu} f|_{A}$ with probability 
at least $\Omega(\delta^{12})$. By Chernoff's 
bound we have that $G(D)|_{A}\neq_{<2\eps}f|_{A}$ with probability 
$1-o(1)$. It follows that 
with probability at least 
$\Omega(\delta^{12})$ over $A \subset_k D$, 
$F[A]$ and $G(A)$ differ 
on at most $2\eps+\nu\leq 3\eps$ 
fraction of the coordinates of $A$. Since the fraction of good $D$s is $\geq 1/2$, $\Delta(F[A],G[A])\leq 3\eps$ on at least $\Omega(\delta^{12})$ fraction of $X(k)$ as required.
\qed

\subsection{Auxiliary Claims}
Our proof requires a few
basic auxiliary probabilistic claims, which 
we record here. The 
first claim asserts 
that if the distance 
between two functions $f,g\colon [d]\to\{0,1\}$,
then choosing a random 
subset $A\subseteq_{k}[d]$, we have that 
the distance between 
$f|_{A}$ is also very
close to $R$. More precisely:
\begin{claim}\label{claim:dist-chernoff}
Suppose $R\in (0,1)$, 
and let $f,g\colon [d]\to\{0,1\}$ be functions such that $\Delta(f,g) = R$. 
Then, for $\frac{1}{R^2}\leq k\leq d$ we have that:
\begin{enumerate}
    \item $\Pr_{A\subseteq_{k} [d]}\big[\Delta_A(f,g) > 2R\big] \leq 2^{-\Omega(Rk)}$.
    \item $\Pr_{A\subseteq_{k} [d]}\big[\Delta_A(f,g) < R/2\big] \leq 2^{-\Omega(Rk)}$.
\end{enumerate}
\end{claim}
\begin{proof}
Both of the items 
are immediate 
consequences of Chernoff's 
inequality. The arguments
are essentially identical, and we give a proof of the first item only.

To see this, sample $A\subseteq [d]$ 
by including each element in $A$ with probability 
$k/d$. Let $I\subseteq [d]$ be the set of $i\in [d]$ such that $f(i)\neq g(i)$, and for each $i\in I$ define the random variable $Z_i$ to be 
the indicator of $i\in A$.
Define $Z = \sum\limits_{i\in I} Z_i$,
and note that $\Delta_A(f,g) = \frac{1}{|A|}Z$. Noting that 
$\E[Z] = \frac{k|I|}{d} = Rk$, by Theorem~\ref{thm:chernoff}
we get that 
$\Pr_{}\big[Z\geq 1.1Rk]\leq 2^{-\Omega(Rk)}$; 
also, by another application of
Theorem~\ref{thm:chernoff}
we get that $|A|\geq 0.9 k$
with probability $1-2^{-\Omega(k)}$. It follows
that except with probability $2^{-\Omega(Rk)}$ we 
have that $\Delta_A(f,g)\leq \frac{1.1 Rk}{0.9 k}\leq 2R$. The probability that $|A| = k$ is $\Omega(1/\sqrt{k})$, 
and conditioned on that 
$A$ is distributed 
as $A\subseteq_k[d]$, 
hence we get that the 
probability in the first
item is at most $O\left(\sqrt{k}2^{-\Omega(Rk)}\right)=2^{-\Omega(Rk)}$.
\end{proof}

The second claim 
asserts that if two functions $f$ and $g$ 
are relatively far,
then there are not many
$k$-sets $A$ on which 
they roughly agree. More 
precisely:
\begin{claim}\label{claim:dist-vs-agr}
Suppose that 
$F\colon X(k)\to \{0,1\}^k$ is an assignment, 
that $D\in X(d)$ is a face and that $f,g\colon D\to\{0,1\}$ are 
functions such that $\dist(f,g) > C\nu$, where $\nu \in (0,1), C \geq 6$. Then 
\[
\Pr_{A\subseteq_k D}
\big[A\in \Agr_\nu(f, F) \cap \Agr_\nu(g, F)\big] \leq 2^{-\Omega(Ck\nu)}.
\]
\end{claim}
\begin{proof}
By Claim~\ref{claim:dist-chernoff}, sampling $A\subseteq_{k}D$
we get that $\Delta_A(f,g)\geq \frac{C\nu}{2}$
with probability 
$1-2^{-\Omega(C\nu k)}$; we
claim that such $A$ 
cannot both be in $\Agr_\nu(f, F)$ and 
in $\Agr_\nu(g, F)$. 
Indeed, otherwise 
we would get that
\[
\frac{C\nu}{2}
\leq
\Delta(f|_{A}, g|_{A})
\leq 
\Delta(F[A], f|_{A})
+
\Delta(F[A], g|_{A})
\leq 2\nu,
\]
and contradiction since $C \geq 6$.
\end{proof}

\subsection{Proof of Lemma~\ref{lem:agr-to-list-agr}: Reduction from agreement to list agreement testing}\label{sec:agr-to-list-agr}

\subsubsection{Localizing to a Johnson}\label{sec:localizing}
The first step of the proof 
is to localize to a random 
$d$-face $D\sim \mu_d$, 
and show that with probability close to $1$, 
the assignment $F$ passes
the direct product test inside $T$ with noticeable probability.
More precisely:

\begin{lemma}\label{lem:agr-test-sampling}
If $(k,s)$-Agreement-Test~\ref{agr-test-hdx} on $F$ passes with probability $\delta$, then
\[\Pr_{D \sim X(d)}
\big[\text{the }(k,s)-\text{direct product test passes with probability } \geq \delta^2/16 \text{ inside }D\big] \geq 1-o(1).\]
\end{lemma}
\begin{proof}
Let $\mathcal{D}_1$ be the distribution on $(A,A',I)$ induced by Agreement-Test~\ref{agr-test-hdx}, and consider the following distribution $\mathcal{D}_2$ over $(A,A',I)$: 
\begin{enumerate}
\item Sample $B \sim \mu_{\sqrt{d}}$.
\item Sample $I \subseteq D$ of size $s$ uniformly.
\item Sample $I \subseteq A,A' \subseteq B$ of size $k$ uniformly.
\end{enumerate}
Note that conditioned on $|A\cap A'| = s$, 
the distributions $\mathcal{D}_1$ and $\mathcal{D}_2$ are identical. Thus, 
as the probability 
of this event is $1-O(k^2/\sqrt{d}) = 1-o(1)$ in both distributions, it follows
that the statistical distance between $\mathcal{D}_1$ and $\mathcal{D}_2$ is $o(1)$. Therefore,
\[\Pr_{(A,A',I) \sim D_2}\big[F[A]|_I = F[A']|_I\big] \geq \delta-o(1).\]
Denote by $\mathcal{D}_2(B)$ the distribution on $(A,A',I)$ conditioned on sampling $B$, and by $p_B$ 
the probability that 
$F[A]|_I = F[A']|_I$ if 
$B$ was chosen. By 
an averaging argument, with probability at least $\frac{\delta}{4}$ over 
the choice $B \sim \mu_{\sqrt{d}}$ 
we have that $p_B\geq \frac{\delta}{2}$; we call 
such $B$ good, and 
denote the set of good 
$B$'s by $\mathcal{B}$.
 
By Lemma~\ref{lem:sampling} we get that
\[\Pr_{D \sim \mu_d}\Big[\Pr_{B\subseteq_{\sqrt{d}}D}[B\in \mathcal{B}] \geq \frac{\delta}{8}\Big] \geq 1-O\left(\frac{1}{\sqrt{d}}+\gamma\right)
=1-o(1).\]
Fix a $d$-face $D$ satisfying the above event. Thus, picking $B \subset_{\sqrt{d}} D$ and $(A,A',I) \sim \mathcal{D}_2(B)$ passes the direct product test with probability at least $\frac{\delta^2}{8}$. Let this distribution be $\mathcal{D}_2(D)$. As before, letting the distribution $\mathcal{D}_1(A)$ be 
the distribution over $(A,A',I) \sim D_1$ conditioned on sampling $D$, the statistical distance between $\mathcal{D}_1(D)$ and $\mathcal{D}_2(D)$ is $o(1)$. Therefore we get that,
\[\Pr_{D \sim \mu_d}\big[
\text{the }(k,s)-\text{direct product test} \text{ passes w.p. } \geq \delta^2/8-o(1) \text{ inside }D\big] \geq 1-o(1),\]
which completes the proof.
\end{proof}
We refer to a $d$-face
$D\in X(d)$ for which
the event in Lemma~\ref{lem:agr-test-sampling} holds as good,
and thus conclude that 
$1-o(1)$ fraction of 
the $d$-faces are good.
Note that the above argument would also 
work for $d/2$-faces, 
and thus we similarly 
define the notion of
good $d/2$-faces.

\subsubsection{Getting a list on each good Johnson and generating a gap}\label{sec:getting_a_list}
Fix a good $d$-face $D$, 
and consider the assignment $F$ when restricted to $k$-sets inside $k$. For 
notational convenience, 
we denote this restricted 
assignment by $F_D$. 
Thus, the event in 
Lemma~\ref{lem:agr-test-sampling} translates to
saying that the direct
product tester over 
the Johnson scheme
passes inside $D$ with
noticeable probability.
Thus, using direct product
testing results over the Johnson scheme, we may
``explain'' this consistency via correlations of $F_D$
with true direct product 
functions. Towards this end, we use a result due to~\cite{DinurG08} (see also~\cite{ImpagliazzoKW09}, who state a version 
that is more convenient for our purposes).

\begin{theorem}\label{thm:DG}
Suppose that $F_D$ passes the $(k,\sqrt{k})$ direct product test in $D$ with probability  $\eps$. Then there is a function $g: [d] \rightarrow \{0,1\}$ such that  
\[
\Pr_{A\subseteq_k D}
\big[
\Delta(g|_{A},F_D[A]) \leq 1/k^{\Omega(1)}
\big]\geq \Omega(\eps^6).
\]
\end{theorem}

Theorem~\ref{thm:DG} by itself is not enough for us, and we need an 
idea that is often useful 
in conjunction with such results: list decoding. We wish to consider all 
direct product functions 
that are correlated with 
$F_D$ and have these as 
the lists. Alas, there is a technical issue: the number of direct product functions that are correlated with $F_D$ need not be bounded in terms of $\eps$, the probability that the test passes. 
To remedy this issue we require
the notion of $\eta$-covers, defined below.
\begin{definition}
Let $\mathcal{F}\subseteq \mathcal{G}$ be two
families of functions 
from $[d]$ to $\{0,1\}$.
We say that $\mathcal{F}$
is an $\eta$-cover for 
$\mathcal{G}$ if for
any $g\in \mathcal{G}$
there exists $f\in\mathcal{F}$
such that $\Delta(f,g)\leq \eta$.
\end{definition}

We are now ready to present a procedure that, 
given a good $d$-face $D$, generates a short list of functions that
``explain'' most of the
probability that $F_D$
passes the direct product test inside $D$, and which is also short. 
The procedure takes 
as input a restriction
of the assignment $F$ 
to a face $D$, which below
we denote by $G$, 
and finds one by one 
direct product functions 
that are correlated with $G$, following by randomizing $G$ at appropriate places.
\begin{mdframed}
\begin{algorithm}
The short list algorithm.
\label{algo-gap} \mbox{}
\begin{description}
\item[Input:] $G: \binom{[d]}{k} \rightarrow \{0,1\}^k$, $\delta>0$, $r \in \mathbb{N}$, $\eta \in (0,1)$.
\item [Output:] List of functions $\{f_1,\ldots,f_m\}$ from $[d] \rightarrow \{0,1\}$.
\item [Operation:]\mbox{}
\begin{enumerate}
    \item Set $t = k^{-c}$ for $0 < c < 1$, $\delta_0= \Theta(\delta^6)$, $\widetilde{G}_0 = G$, and initialize $L_1, I_1 =\emptyset$.
    \item For $i \in \{0,\ldots, \lfloor1/\delta^{80}\rfloor\}$:
    \begin{itemize}
        \item  If there exists $f$ with $\agr_t(f, \widetilde{G}_i) > \delta_i$, add $i$ to $I_1$ and $f_i$ to $L_1$.
        \item Obtain $\widetilde{G}_{i+1}$ by randomizing $\widetilde{G}_i$ on $k$-sets $A\in \Agr_t(f,\widetilde{F}_i)$.
        \item $\delta_{i+1} = \delta_i - \delta^{100}$.
    \end{itemize}
    \item Create lists $I_2, L_2$ as follows: for all $i \in I_1$, add $i$ to $I_2$ and $f_i$ to $L_2$ iff $i \geq r$.
    \item Construct a graph $G$ whose vertices are $L_2$, and $f,g\in L_2$ are adjacent if $\Delta(f,g) < \eta$. Take a maximal independent set in $G$ and add the corresponding functions to $L_3$.
    \item Output $L_3$.
\end{enumerate}
\end{description}
\end{algorithm}
\end{mdframed}

The following lemma summarizes some 
of the basic properties 
of the short list algorithm. We will use the parameters and notation specified in the algorithm throughout this section.

\begin{lemma}\label{lem:algo1}
 When ran on $G = F_D$ 
 for a good $d$-face $D$
 with parameter $\Theta(\delta^{2})$ in place of $\delta$, setting $\delta' = \Theta(\delta^{12})$, with probability $1-o(1)$
 Algorithm~\ref{algo-gap} outputs a list $L = \{(i,f_i)\}_{i \in I}$ with $I \subset \{0,\ldots,1/\delta'^{80}\}$ such that,
\begin{enumerate}
    \item $0 \neq |I_1| \leq \frac{2}{\delta'}$.
    \item For all $i \in I_1$, $\agr_t(f_i, G) > \delta'-i\delta'^{100}-o(1)$.
    \item If $i \notin I_1$ then for all $g$, $\agr_{t}(g, \widetilde{G}_i) < \delta_i$.
    \item For all $i \in \lfloor 1/\delta'^{80}\rfloor$ and $B\subseteq_{d/2} A$, if $g\colon B\to\{0,1\}$ is a function such that 
    $\min_{j\in I_1, j\geq i}\Delta(g, f_j|_B) > \Omega(\log(1/\delta')t)$  and $\agr_t(g,\widetilde{G}_{i+1}|_B) < \theta$, then $\agr_t(g, G|_B) < \theta+\exp(-\Omega(tk\log(1/\delta')))$.
\end{enumerate}
\end{lemma}

\begin{proof}
First note that by Theorem~\ref{thm:DG} we get that there is at least one function with $\agr_t(f) \geq \delta'$, therefore the list is non-empty. Let us start by proving the upper bound on the size.

\paragraph{Proof of (1):} At the $i^{th}$ iteration we add a function to the list only if $\agr_t(f_i,\tilde{G}_i) > \delta_i$ which is always at least $\delta'-\delta'^{20}$. Let $\cR \subseteq \binom{D}{k}$ be the $k$-sets that have been randomized in the algorithm so far, so $|\cR| \geq (\delta' - \delta'^{20}) {d \choose k}$. Using the Chernoff bound we get that every function $g\colon D\to \{0,1\}$ satisfies:
\[\Pr\left[\frac{|\Agr_t(g) \cap \cR|}{|\cR|} > \frac{2{k \choose tk}}{2^k}\right] \leq \exp\left(-\frac{{k \choose tk}}{2^k}\delta {d \choose k}\right) \leq \exp(-(d/4)^k).\]
Therefore by a union bound we get that with probability $1-o(1)$, for all functions on $D$ the above holds, and we condition on this event. Hence, the contribution of $\mathcal{R}$ to the 
agreement of function found in later steps in the procedure is always at most $o(1)$. Thus, each newly found function in the process 
increases the measure of $\mathcal{R}$ by at least $\delta' -\delta'^{20}-o(1)\geq \delta'/2$.
Therefore, with probability $1-o(1)$ the process terminates after at most $2/\delta'$ steps, which is thus also 
an upper bound on the list size $I_1$.

\paragraph{Proof of (2):} 
If we inserted $f$ into
the list at step $i$, 
then $\agr_t(f, \tilde{G}_i)\geq \delta'-i\delta'^{100}$. As we have already argued, with probability $1-o(1)$ at most $o(1)$ of this agreement comes 
from $k$-sets in which $\tilde{G}_i$ was randomized, and it follows that $\agr_t(f,G)\geq \delta'-i\delta'^{100} - o(1)$.

\paragraph{Proof of (3):} If $i\not\in I_1$ then 
the process terminated before step $i$, meaning 
that the assignment at
that time no longer was
$\delta_i$-correlated with
any direct product function.

\paragraph{Proof of (4):} Denote by $\cR_i$ the collection of all $k$-sets in which the assignment has been randomized in steps prior to the $i+1$th iteration, 
and consider $\widetilde{G}_{i+1}$. By Claim~\ref{claim:dist-vs-agr} for all $j \geq i, j \in I_1$ we get,
\begin{equation}
\Pr_{A\subseteq_k B}\big[A\in \Agr_t(g, G|_B) \cap \Agr_t(f_j|_B, G|_B)\big] \leq \exp(-\Omega(tk\log(1/\delta'))),    
\end{equation}

and so 
\[
\Pr_{A\subseteq_k B}\big[A\in \Agr_t(g, G|_B) \cap \cR_i\big] \leq 1/\delta'\cdot \exp(-\Omega(tk\log(1/\delta')))+o(1)\leq \exp(-\Omega(tk\log(1/\delta'))).
\]
It follows from the above that
\begin{align*}
\Pr_{A\subseteq_k B}
\big[A\in \Agr_t(g, G|_{B})]
&=
\Pr_{A\subseteq_k D}\big[A\in \Agr_t(g, G|_B) \cap \cR_i\big]
+
\Pr_{A\subseteq_k D}\big[A\in \Agr_t(g, G|_B) \cap \overline{\cR_i}\big]\\
&\leq 
\exp(-\Omega(tk\log(1/\delta')))
+
\Pr_{A\subseteq_k D}\big[A\in \Agr_t(g, \tilde{G}_{i+1}|_B) \cap \overline{\cR_i}\big]\\
&\leq
\exp(-\Omega(tk\log(1/\delta')))
+
\Pr_{A\subseteq_k D}\big[A\in \Agr_t(g, \tilde{G}_{i+1}|_B)\big],
\end{align*}
which is at most
$\theta+\exp(-\Omega(tk\log(1/\delta')))$.
\end{proof}

We will now consider 
the run of the short list
algorithm on a $d$-face with
various options for parameters, and its relationship with 
direct product functions on 
$d/2$ sub-faces. We will especially 
care about the relationship between the functions in the list
of the $d$-face $D\in X(d)$, 
and direct product functions on its $d/2$-faces that have large
correlation with the assignment $F$. In a sense, 
we will want to show that these
are ``the same functions''; ultimately, this is where 
the local consistency of the
lists comes from. 

Towards this end, we will run
the algorithm above for
$D$ faces, and denote the outputted list by $L[D]$, 
For $d/2$ sub-faces of $D$, 
we will let $L[B]$ be an 
$\eta$-cover for functions
that have sufficient agreement with $F|_{B}$.
The following lemma summarizes the properties
of such runs of the short list algorithm:
\begin{lemma}\label{lem:agr-johnson}
Let $\eps,\delta >0$, $\eta = 2^{-1/\delta^{1200}}$,
let $k$ be sufficiently large and let $d \geq \poly(k)\exp(\poly(1/\delta))$.
Suppose that $F_D$ passes the  $(k,\sqrt{k})$ direct product tester inside $D$ with probability at least $\delta$. 
Then choosing $r, i \sim \lfloor1/\delta^{80}\rfloor$ uniformly and
running Algorithm~\ref{algo-gap} with parameters $r$ and $\eta' = \delta^{-100i}\eta$ on $D$ and on all $d/2$ sub-faces, with 
probability $1-O(\delta^{68})$
the algorithm outputs a list $L[D]$ such that: \begin{enumerate}
    \item 
    \textbf{Non-empty, short list:} $0 \neq |L[D]| \leq 1/\delta'$, where $\delta'=\Theta(\delta^6)$.
    \item 
    \textbf{Significant correlation:} For all $f \in L$, $\agr_t(f, F_D) \geq \delta_r:= \delta'-r\delta'^{100}$, where $t=k^{-\Omega(1)}$.
    \item 
    \textbf{Large distance in the list:} $\Delta(L[D]) > \delta^{-100}\eta'$. 
    \item 
    \textbf{Downwards consistent:} $\Pr_{B\subseteq_{d/2} D}[\forall f \in L[D], \exists g \in L[B] \text{ with } \Delta(f|_B,g) \leq \eta'] \geq 1-o(1)$. In words, for each function in the list of $D$, projecting it onto a random
    $B\subseteq_{d/2} D$ yields a
    function which is very close
    to a function in the list of $B$.
    \item 
    \textbf{Upwards consistent:} $\Pr_{B \subseteq_{d/2} D}[\forall g \in L[B], \exists f \in L[D] \text{ with } \Delta(g,f|_B) \leq 2\eta'] \geq 1-o(1)$. In words, choosing 
    a random $B\subseteq D$, every function in the list 
    $L[B]$ is close to a projection of some function
    from the list $L[D]$.
\end{enumerate}
For each $B \subseteq_{d/2} D$, $L[B]$ is an $\eta'$-cover for functions on $B$ with $\agr_t(g,F|_B) > \delta_r-\delta^{200}$.
\end{lemma}
The first four items in Lemma~\ref{lem:agr-johnson} 
are not too hard to establish; the fifth item however requires 
more care, and this is where we
are going to utilize results from
random sub-instances of max-$k$-CSPs. In particular, we require
the following lemma which follows
from results in~\cite{AlonVKK02} (and more precisely, from Theorem~\ref{thm:max-csp}).

\begin{lemma}\label{lem:no-jump}
For all $\zeta \in (0,1)$, $d \geq \poly(k)\exp(1/\zeta^2)$, and all functions $G: \binom{[d]}{k} \rightarrow \{0,1\}^k$ that satisfy $\agr_{t}(g,G) \leq \alpha$ for all $g\colon [d]\to\{0,1\}$, the following holds:
\[\Pr_{B\subseteq_{d/2} [d]}[\max_g \agr_t(g|_B,G|_B) < \alpha+\zeta] \geq 1-\poly(1/d).\]
\end{lemma}

\begin{proof}
Consider the following Max-$k$-CSP $
\Psi = ([d], \cF)$. The constraints in $\cF$ are as follows: for every $k$-subset $I$ we have the constraint $f_I:\{0,1\}^k \rightarrow \{0,1\}$ defined as,
\[f_I(x) = \begin{cases}
1, ~~~\text{ if }\Delta(x,G[A]) \leq t,\\
0, ~~~\text{ otherwise.}
\end{cases}\]
Thus, the value of $\Psi$ is $\val(\Psi) = \max_g \agr_t(g,G)$. Applying Theorem~\ref{thm:max-csp} with $\tau = 1/d^c$ for small enough $c > 0$, we get that with probability $1-\tau$ over the choice of $B\subseteq_{d/2} [d]$, 
$\val(\Psi|_{B})\leq \val(\Psi)+\zeta$, which 
is at most $\alpha+\zeta$. 
Noting that 
$\val(\Psi|_{B}) = \max_g \agr_t(g|_B,G|_B)$ finishes 
the proof.
\end{proof}

We are now ready to prove Lemma~\ref{lem:agr-johnson}.

\begin{proof}[Proof of Lemma~\ref{lem:agr-johnson}]

The proofs of (1) and (2) are immediate from point (1) and (2) of Lemma~\ref{lem:algo1}. 

\paragraph{Proof of (3):} Consider the lists produced
by the algorithm and consider the pairwise distances $\Delta(f_i,f_j)$ for $f_i,f_j \in L[D]$. Since $|L_2| \leq 1/\delta'$ there are at most $1/\delta'^2$ different pairwise distances, therefore with probability $1-O(\delta^{68})$ over $i \in \{0,\ldots,1/\delta^{80}\}$ we have that for all $i \neq j$ either $\Delta(f_i,f_j) < \eta'$ or $>\delta^{-100}\eta'$. In that case, a maximal independent set $L_3$ obtained in the foruth step of the short list algorithm satisfies that for all $f_i,f_j \in L_3$, $\Delta(f_i,f_j) > \delta^{-100}\eta'$. 

\paragraph{Proof of (4):} By Lemma~\ref{lem:sampling}, we get that for each $f\in L[D]$, with probability $1-o(1)$ over the choice of $B\subseteq_{d/2} D$ 
we have that $\agr_t(f|_B, F|_B)\geq \delta_r - o(1)$.
Thus, by the upper bound on the size of $L[D]$ and the union bound we get
\[\Pr_{B\subseteq_{d/2} D}\big[\forall f \in L[D], \agr_t(f|_a) \geq \delta_r-o(1)\big] \geq 1-o(1).\]
By the property of $\eta'$-covers we conclude that
\[\Pr_{B\subseteq_{d/2} D}[\forall f \in L[D], \exists g \in L[B] \text{ with } \Delta(f|_B,g) \leq \eta'] \geq 1-o(1).\]

\paragraph{Proof of (5):} Note that the list $L_2$ has size at most $1/\delta'$, hence with probability at least $1-O(\delta^{74})$ over the choice of $r$, we get that $r+1 \notin I_1$. This means that we have a gap: $\forall h, \agr_t(h, \widetilde{G}_{r+1}) < \delta_{r+1}$. Condition on $r$ being chosen so that this holds; by Lemma~\ref{lem:no-jump} we get that
\[\Pr_{B \subseteq_{d/2} D}\big[\max_h \agr_t(h|_B,\widetilde{G}_{r+1}|_B) < \delta_{r+1}+\delta^{200}\big] \geq 1-o(1).\]
Fix a $B$ where the above holds, let $L[B]$ be an $\eta'$-cover as in the statement of the lemma, and take $g \in L[B]$. Assume for contradiction that $\Delta(f|_B,g) > \eta'+\Omega(\log(1/\delta')t)$ for all $f \in L_3$. 
By the maximality of the independent set $L_3$, we get that for all $f \in L_2 \setminus L_3$, there exists $f' \in L_3$ such that $\Delta(f,f') < \eta'$. Therefore if $g$ is $\Omega(t\log(1/\delta'))+\eta$-far from all $f \in L_3$, then it is $\Omega(\log(1/\delta')t)$-far from all $f' \in L_2$ and in particular from all $f_j \in L_1$ for $j \geq r, j \in I_1$. Since $\agr_t(g, \widetilde{G}_{r+1}|_B) < \delta_{r+1}+\delta^{200}$, we may apply the fourth item in Lemma~\ref{lem:algo1} to get that $\agr_t(g,G|_B) <\delta_{r+1}+\delta^{200}+\exp(-\Omega(tk\log(1/\delta'))) < \delta_r-\delta^{200}$, for $G = F_D$, which is a contradiction to $g$ being in $L[B]$.    
\end{proof}

\subsubsection{Consistency of the local lists}\label{sec:consistency-of-lists}
In this section, we finish 
the proof of Lemma~\ref{lem:agr-to-list-agr}. Fix parameters
as therein, let $\mathcal{D}$
be the set of good faces 
(namely, faces in which 
the $(k,\sqrt{k})$ agreement 
test passes with probability 
at least $\delta' = \delta^2/16$), 
and recall that by Lemma~\ref{lem:agr-test-sampling}
we have that $\mu_d(\mathcal{D})\geq 1-o(1)$.

Let $\eta = 2^{-1/\delta^{1200}}$. We sample $r$ and $i$ integers
between $1$ and $\lceil 1/\delta^{80}\rceil$ uniformly, set $\eta' = \delta^{-100i}\eta$
and run the short list algorithm
on each $D\in \mathcal{D}$ 
with the parameters $r$ and $i$. 
For each $D$, with probability 
$1-O(\delta^{68})$ (over the choice of $r,i$) we get a list
$L[D]$ as in Lemma~\ref{lem:agr-johnson}. It follows by linearity
of expectation and an averaging
argument that we may choose $r$
and $i$ such that we get lists
$L[D]$ for at least $1-O(\delta^{68})$ of $D\in\mathcal{D}$ such 
that $L[D]$ satisfies the 
conditions of Lemma~\ref{lem:agr-johnson}, and we fix such $r$ and $i$ henceforth. Below,
we refer to a good $D$ that
additionally has a list
$L[D]$ satisfying the 
conditions of Lemma~\ref{lem:agr-johnson} as very good, and 
we note that the probability
that $D$ is very good 
is at least 
$1-O(\delta^{68})-o(1) = 1-O(\delta^{68})$. 
For each $B\in X(d/2)$, we fix 
$L[B]$ to be an $\eta'$ cover 
of the collection of functions 
$g\colon B\to\{0,1\}$ such that 
$\agr_t(g, F|_{B})\geq \delta_r = \delta' - r\delta'^{100}$.

The first three items
in the statement of Lemma~\ref{lem:agr-to-list-agr} clearly hold by Lemma~\ref{lem:agr-johnson}, and
in the rest of the argument we argue that the list agreement test passes.
Towards this end, consider 
a generation of queries
for the list agreement test. 
Namely, sample $B \sim \mu_{d/2}$ and independently sample $D,D' \supset_d B$. We say a triple $(D,B,D')$ is good if:
\begin{enumerate}
    \item The $d$-faces $D$ and $D'$ are very good.
    \item 
    It holds that $\Delta(L[D]|_B), \Delta(L[D']|_B) > \frac{1}{2}\delta^{-100}\eta'$.
    \item For all $f \in L[D]$, there exists $g \in L[B]$ with $\Delta(f|_B,g) < \eta'$, and for all $g \in L[B]$ there exists $f \in L[D]$ with $\Delta(g,f|_B) < 2\eta'$. The same holds when $D$ is replaced by $D'$.
\end{enumerate}
Note that since marginally, 
each one of $D$ and $D'$ is 
distributed according to $\mu_d$,
we get that the first item holds
with probability $1-O(\delta^{68})$. Note that the marginal distribution of $(B,D)$ is
the same as sampling $D\sim \mu_d$, and then $B\subseteq_{d/2} D$. Thus,
if the first item holds, then
$\Delta(L[D])\geq \delta^{-100}\eta'$, hence 
by Claim~\ref{claim:dist-chernoff} we get that 
the second item holds with
probability $1-o(1)$. 
Lastly, if the first item holds, 
then by Lemma~\ref{lem:agr-johnson} we get that the 
third item holds with probability $1-o(1)$. 
Overall by the union bound, 
we get that all of the events above holds together with 
probability at least $1-O(\delta^{68})$.

To finish the proof, we argue that if $(D,B,D')$ is good, then the list agreement test passes on it. For that, we show that for each $f \in L[D]$ there exists a unique $f' \in L[D']$ s.t. $\Delta_B(f,f')\leq 3\eta'$ and vice versa. We show 
the argument only in one of 
the directions, and the other direction is identical.
Take $f \in L[D]$ and consider $f|_{B}$; 
by the $\eta'$-cover property 
we can find a $g \in L[B]$ with $\Delta(f|_B,g) \leq \eta'$. 
By the third property above, for $g$ we may find $f' \in L[D']$ with $\Delta(g,f'|_B) \leq 2\eta'$, so by the triangle inequality  $\Delta_B(f,f') \leq 3\eta$. Next, we show the uniqueness of $f'$. For any $f''\in L[D']\setminus\{f'\}$, by the second property above 
$\Delta(f''|_{B}, f'|_{B})\geq \frac{1}{2}\delta^{-100}\eta'$, 
so 
\[
\Delta(f''|_{B}, f|_{B})
\geq 
\Delta(f''|_{B}, f'|_{B})
-\Delta(f'_{B}, f|_{B})
\geq 
\frac{1}{2}\delta^{-100}\eta'
-3\eta'
\geq 
100\eta'.
\]

\subsection{List Agreement Testing Using UG Coboundary Expansion: Proof of Lemma~\ref{lem:list-agr-test}}\label{sec:list-agr-testing}
The goal of this section is to prove Lemma~\ref{lem:list-agr-test}. 
Throughout this section, we fix lists $\{L[D]\}_{D \in X(d)}$ satisfying the premise of Lemma~\ref{lem:agr-to-list-agr}. We refer to a $d$-face $D\in X(d)$ for which the properties in Lemma~\ref{lem:agr-to-list-agr} are satisfied 
as good, and note that the measure of the set of good $d$-faces under $\mu_d$ is at least $1-O(\tau)$. Our first 
goal is to define a locally consistent instance of 
Unique-Games on which we can apply coboundary expansion.
At the moment though we have assignments only to the 
$d$-faces, and our UG coboundary expansion only holds
for much lower levels. Thus, we will first show
how to project our list assignments to lower levels.

\subsubsection{Global consistency of the list sizes}
We begin with establishing several basic claims that will be useful in the projection process. 
The following claim asserts 
that almost all of the lists $L[D]$ have the same size. 
More precisely,
\begin{claim}\label{claim:list-size}
There exists $\ell \leq {\sf poly}(1/\delta)$ such that 
$\Pr_{D \sim \mu_d}\big[|L(D)| \neq \ell\big] \leq 10\tau$.
\end{claim}
\begin{proof}
Suppose this is not the case. Then the set of $d$-faces
$X(d)$ can be broken into two disjoint parts, $P_1\cup P_2$
such that $\mu_d(P_1), \mu_d(P_2)\geq 10\tau$ and for
every $D\in P_1$, $D'\in P_2$ we have that $|L[D]|\neq |L[D']|$. Note that in that case, there can never be 
a $1$-to-$1$ correspondence between $L[D]$ and $L[D']$, 
and hence we conclude that the list agreement test fails
whenever it picks $D\in P_1$ and $D'\in P_2$. Suppose
without loss of generality that $\mu_d(P_1)\leq 1/2$

On the other hand, considering the graph $G$ on $X(d)$ generated
by the list agreement test, which equivalently can be 
stated as pick $D\sim \mu_d$, $B\subseteq_{d/2} D$ and
then $D'\supseteq_{d} B$ according to $D'\sim \mu_d$. 
By Lemma~\ref{lem:spectral_gap_of_graphs_from_HDX} second
eigenvalue of $G$ is at most $\frac{1}{2}+O(d^2\gamma)$. 
It follows from Cheeger's inequality that the 
edge expansion of $P_1$ is at least $\frac{1}{4}-O(d^2\gamma)\geq 1/8$. It follows that a randomly 
sampled edge goes from $P_1$ to $P_2$ 
with probability at least $\mu_d(P_1)/8\geq 10\tau/8$.
This contradicts the fact that the probability 
that the list agreement test fails is at most $\tau$.
\end{proof}

We pick $\ell$ to be the list size parameter 
from Claim~\ref{claim:list-size}. In the next 
claim we prove that the fact that the list $L[D]$
typically has a large distance implies that 
its projection onto a sub-face has the same size.
\begin{claim}\label{claim:projected-dist}
For $t \geq 102\frac{\log(\ell/\tau)}{\delta^{-100}\eta'}$ 
it holds that
$\Pr_{\substack{D \sim \mu_d \\B \subset_t D}}\big[|L(D)|_B| \neq \ell\big] \leq O(\tau)$.
\end{claim}
\begin{proof}
We prove that for each good $D$, conditioning on $D$, 
the above probability is at most $O(\tau)$, and the 
claim trivially follows. Fix a good $D$ and consider 
distinct $f,g\in L[D]$. Then $\Delta(f,g)\geq \delta^{-100}\eta'$, 
and the probability over the choice of $B$ that 
$f|_{B} = g|_{B}$ is at most 
$(1-\delta^{-100}\eta')^t\leq e^{-\delta^{-100}\eta' t}
\leq \frac{\tau^{50}}{\ell^{50}}$.
By a union bound we have that it follows that
the probability there are distinct $f,g\in L[D]$ 
such that $f|_{B} = g|_{B}$ is at most $\ell^2\frac{\tau^{50}}{\ell^{50}}\leq \tau$, completing the proof.
\end{proof}

\subsubsection{Majority decoding}
Next, we show that for $t$ that is not too large, 
for a typical $t$-face $B$, almost all of the $d$-faces $D$
have the same projection of $L[D]$ onto $B$. More precisely:
\begin{claim}\label{claim:list-equality}
For $t \leq \frac{\tau}{\eta'}$, 
with probability at least $1-O(\sqrt{\tau})$ 
over the choice of $B\sim \mu_t$
it holds that 
\[
\Pr_{\substack{D,D' \supseteq_d~ B}}[L[D]|_B = L[D']|_B] \geq 1-O(\sqrt{\tau}).
\]
\end{claim}
\begin{proof}
Consider the following sampling procedure:
sample $B\sim \mu_t$, then $C\supseteq_{d/2} B$ 
and then sampling $D,D'\supseteq_{d} C$ independently.
We first claim that with probability $1-\tau$ it holds that $L[D]|_{C} \neq_{<\eta'} L[D']|_{B}$. Indeed, first note that 
the marginal distribution of $D$ and $D'$ is $\mu_d$, 
and the marginal distribution of $(D, C, D')$ is 
according to the list agreement test. Hence,
with probability $1-O(\tau)$ it holds that $D$, $D'$ 
are good and the list agreement test passes on $(D,C,D')$.
In that case, we may find a matching $\pi\colon L[D]\to L[D']$ such that $\Delta(\pi(f)|_{C}, f|_C)\leq \eta'$ for all $f\in L[D]$. Noting that $B$ is a random subset of $C$ of 
size $t$, we get that $\Delta(\pi(f)|_{B}, f_B) = 0$
with probability at least $1-t\eta'\geq 1- \tau$.

For each $B\in X(t)$, let 
\[
p_B = \Pr_{\substack{C\supseteq_{d/2} B\\ D,D'\supseteq C}}\big[L[D]|_{B} = L[D']|_{B}].
\]
Rephrasing the conclusion of the previous discussion,
we have that $\E_{B\sim \mu_t}\big[p_B\big]\geq 1-\tau$. By an averaging argument, defining 
$\mathcal{B} = \{B\in X(t)~|~p_B\geq 1-O(\sqrt{\tau})\}$
we have that $\mu_t(\mathcal{B})\geq 1-O(\sqrt{\tau})$. 
Fix $B\in\mathcal{B}$; we argue that there exists 
a list of functions $L[B]$ on $B$ such that 
\[
\Pr_{D\supseteq_d B}\big[L[D]|_{B} = L[B]\big]
\geq 1-O(\sqrt{\tau}).
\]
Indeed, otherwise we may partition the set of 
$D$'s containing $B$ into $P_1$ and $P_2$ 
of relative measure at least $c'\sqrt{\tau}$ 
so that $L[D]|_{B}\neq L[D']|_{B}$ for all $D\in P_1$,
$D'\in P_2$, where $c'$ is an absolute constant to be determined.
Consider the $G$ on $d$-faces containing $B$, whose edges are sampled 
by first picking $C\supseteq_{d/2} B$ and then 
$D, D'\supseteq_{d} C$ independently; by 
Lemma~\ref{lem:spectral_gap_of_graphs_from_HDX} this 
graph has second eigenvalue at most $1/2 + \poly(d)\gamma$, and thus the fraction 
of edges inside the graph that go from $P_1$ to $P_2$ 
is at least $(1/4-O(d^2\gamma))c'\sqrt{\tau}\geq c'\sqrt{\tau}/8$. On 
any such edge we have that $L[D]|_{B}\neq L[D']|_{B}$, 
and it follows that $p_B\leq 1-c'\sqrt{\tau}/8$, and contradiction provided that $c'$ is sufficiently large.
\end{proof}

With Claim~\ref{claim:list-equality} in hand, one may 
naturally project the lists that we have on $d$ faces
to $t$-faces in a way that ``preserves their essence''.
More precisely, take a parameter $t$ in
the range
\begin{equation}\label{eq:choose_t}
102\frac{\delta^{100}\log(\ell/\tau)}{\eta'}\leq t\leq \frac{\tau}{\eta'}.
\end{equation}
For each $B \in X(t)$ define a list
for $B$ using weighted majority 
\[
L[B] := \maj_{D \supset_d B}\left(L[D]|_B\right),
\]
where the weight of $D$ is $\Pr_{D'\supseteq_{d} B}\big[D' = D\big]$

\begin{claim}\label{claim:maj-list}
For $t$ in the range as in~\eqref{eq:choose_t}, we have that:
\begin{enumerate}
    \item $\Pr_{B\sim \mu_t}\big[|L[B]| = \ell\big]\geq 1-O(\sqrt{\tau})$.
    \item 
    Choosing $B\sim \mu_t$, with probability at least $1-O(\sqrt{\tau})$ it holds that $\Pr_{\substack{D\supseteq_d B}}\big[L[D]|_{B} = L[B]\big]\geq 1-O(\sqrt{\tau})$.
\end{enumerate}
\end{claim}

\begin{proof}
The second item holds for every $B$ satisfying 
the conclusion of Claim~\ref{claim:list-equality}, 
and hence it follows. The first item follows 
from the second item when it is combined with
 Claim~\ref{claim:projected-dist} using the union bound.

\end{proof}

\subsubsection{Designing the Unique Games instance and proving triangle consistency}
Fix a $t$ as in~\eqref{eq:choose_t}.
Our next goal is to define a Unique-Games instance 
on the weighted 
graph $G$ whose vertices are $X(t)$ and whose 
edge correspond to $2t$-faces: the edges are 
$(u,v)$ where $u\cup v\in X(2t)$, and its weight
is proportional to $\mu_{2t}(u\cup v)$. We 
remark that strictly speaking, we only 
define a partial Unique-Games instance on 
the subset of $t$-faces $B$ where $|L[B]| = \ell$. 
By Claim~\ref{claim:maj-list} these $t$-faces 
constitute 
almost all of $X(t)$, and we encourage the reader
to ignore this point.\footnote{Alternatively, one 
may think of picking an arbitrary list of 
size $\ell$ for every $B\in X(t)$ where 
$|L[B]| \neq \ell$.} 

\paragraph{List ordering, permutations and concatenation.} 
Towards this end,
we fix an ordering 
for each one of the lists constructed thus far 
(both for $d$-faces as well as for $t$-faces). 
Thus, we will think of the list of $B\in X(t)$ 
as $L[B] = (L_1[B],\ldots,L_{\ell}[B])$. For a 
permutation $\pi\in S_{\ell}$, we define 
$\pi(L[B]) = (L_{\pi(1)}[B],\ldots,L_{\pi(\ell)}[B])$.
For $u,v\in X(t)$ such that $u\cup v\in X(2t)$ and $\pi\in S_{\ell}$, we denote
\[
L[u]\circ \pi(L[v]) = 
\left(L_1[u]\circ L_{\pi(1)}[v],\ldots, L_{\ell}[u]\circ L_{\pi(\ell)}[v]\right),
\]
and think of it as a list of assignments to $u\cup v$.

\paragraph{Defining the constraints of the Unique Games instance}
Consider the set of $2t$-faces $W\in X(2t)$, and note that
one has the analog of Claim~\ref{claim:maj-list} for these
as well, and thus we fix lists $L[W]$ satisfying 
Claim~\ref{claim:maj-list} for $2t$-faces as well. 
Let $\mathcal{W}\subseteq X(2t)$ be the collection 
of all $2t$-faces for the items in Claim~\ref{claim:maj-list}
hold.

We now define a Unique-Games instance $\Psi$ over 
$G$ by describing the constraints on the graph $G$. 
For each edge $(u,v)$, we put a constraint as follows.
If $u\cup v\not\in\mathcal{W}$, we put an arbitrary constraint. Else, we put a constraint between $u$ and 
$v$ if $L[u\cup v]|_{u} = L[u]$ and $L[u\cup v]|_{v} = L[v]$.
Note that in that case, there is a natural
$1$-to-$1$ correspondence between $L[u]$, $L[u\cup v]$ 
and $L[v]$, and we fix it to be the constraint between $L[u]$
and $L[v]$. Stated otherwise, the constraint on $(u,v)$
is the unique permutation $\pi = \pi(u,v)\in S_{\ell}$ such that
$L[u\cup v] = L[u]\circ \pi(L[v])$ (when both sides are 
thought of as assignments to $u\cup v$). We think of
edges as being directed, and note that then 
$\pi(u,v) = \pi(v,u)^{-1}$.

The following claim asserts that $\Psi$ is 
$(1-O(\sqrt{\tau}))$ strongly triangle consistent.
\begin{claim}\label{claim:ug_instance_triangle_consistent}
$\Pr_{\substack{Z \sim \mu_{3t} \\ Z= u\cup v\cup w}}[(u,v,w) \text{ is strongly consistent in $\Psi$}] \geq 1- O(\sqrt{\tau})$.
\end{claim}
\begin{proof}
We use Claim~\ref{claim:maj-list} for $3t$-faces, 
and denote the set of $3t$-faces for which 
the items there hold by $\mathcal{Z}$. Thus, 
$\mu_{3t}(\mathcal{Z})\geq 1-O(\sqrt{\tau})$.
Note that sampling $D\sim \mu_d$, then $Z\subseteq_{3t} D$
and then writing $Z = u\cup v\cup w$, with probability 
$1-O(\sqrt{\tau})$ we have that there is a $1$-to-$1$
correspondence between the list of each one of 
$u,v,w$, the lists of $u\cup v, v\cup w, u\cup w$, 
the list of $Z$ and the list of $D$. We 
get a $1$-to-$1$ correspondence between the 
list of $u$ and the list of $v\cup w$, and we 
denote it by $\pi(u,v\cup w)$, and all of these 
correspondences are consistent. 
In particular, we get that $\pi(u,w) = \pi(u,v)\circ \pi(v,w)$ (as both can be thought of as re-alignments 
of the list of $w$ to concatenate with the list of $u$
so that they agree with $L[u\cup w]$), and 
hence $\pi(w,u)\pi(u,v)\pi(v,w) = {\sf id}$. This
proves triangle consistency, and strong triangle
consistency readily follows.
\end{proof}

\subsubsection{Applying UG coboundary expansion}
By Claim~\ref{claim:ug_instance_triangle_consistent} 
we get that $\Psi$ is $(1-O(\sqrt{\tau}))$ strongly
triangle consistent, and applying the Unique-Games 
coboundary expansion we get that there is $g\colon X(t)\to S_m$ such that 
\begin{equation}\label{eq:after_ug_expand}
\Pr_{u\cup v\in X(2t)}\big[\pi(u,v) = g(u)g(v)^{-1}\big]
\geq 1-c.
\end{equation}
We now pick an element from the list of each $u$. 
More precisely, define $h\colon X(t) \to [\ell]$ 
defined by $h(v) = g(v)(1)$. Note that if $(u,v)$
is an edge such that the event in~\eqref{eq:after_ug_expand} holds, then 
\[
\pi(u,v)(h(v)) 
= \pi(u,v)g(v)(1)
= g(u)(1)
= h(u).
\]


In other words, for each vertex $u$ we picked an assignment 
from the list of $u$ in a locally consistent way. We may 
thus define the assignment $R(u) = L[u]_{h(u)}$; 
our goal is to show that there is a global function on
$X(1)$ that agrees with many of these selections. 
Towards this end, we first show that $R$ passes the 
standard direct product test with probability close to 
$1$.
\begin{lemma}\label{lem:cobdry-to-global}
We have that
\[
\Pr_{\substack{D\sim \mu_{d}\\ Q\subseteq_{t/2} D\\ 
Q\subseteq B, B'\subseteq_{t} D}}
\big[R(B)|_{Q} = R(B')_{Q}]
\geq 1-O(\tau^{1/4}+c^{1/2}).
\]
\end{lemma}
\begin{proof}
Sample $Z\sim \mu_{3t}$, 
and write $Z = u\cup v\cup w$, $Z = u\cup v'\cup w'$
independently. Note that by the strong triangle consistency,
get that with probability at least 
$1-O(\sqrt{\tau}+c)$ we have that 
\[
L[Z] 
= L[u]\circ \pi(u,v) L[v] \circ \pi(u,w) L[w]
= L[u]\circ \pi(u,v') L[v'] \circ \pi(u,w') L[w']
\]
and the edges $(u,v),(u,w), (u,v'), (u,w')$ are satisfied.
In that case, we conclude that $R(u)\circ R(v)\circ R(w)$
and $R(u)\circ R(v')\circ R(w')$ correspond to the same 
function in the list of $Z$, and so we get that
\begin{equation}\label{eq:multi_slice}
\Pr_{\substack{Z\sim \mu_{3t}\\ Z =u\cup v\cup w=u\cup v'\cup w'}}\big[R(u)\circ R(v)\circ R(w) = R(u)\circ R(v')\circ R(w')\big]
\geq 1-O(\sqrt{\tau}+c).
\end{equation}
For $Z\in X(3t)$, we associate splittings as $Z = u\cup v\cup w$ points in 
the multi-slice 
\[
\binom{[3t]}{t,t,t} = \big\{x\in \{0,1,2\}^{3t}~|~\forall j\in\{0,1,2\}, \big |\{i~|~x_i = j\}| = t\big\}
\]
by identifying $u$ with the set of coordinates equal to 
$0$, $v$ with the set of coordinates equal to $1$ and 
$w$ with the set of coordinates equal to $2$. 
We define $\tilde{R}_Z(x) = R(u)\circ R(v)\circ R(w)$. 
For each $j\in\{0,1,2\}$, consider the Markov chain $\mathrm{T}_j$ on $\binom{[3t]}{t,t,t}$ that from $x$ moves to $y$ where the set of 
coordinates that are $0$ are kept, and the rest are randomized. Then~\eqref{eq:multi_slice} implies that
\[
\Pr_{\substack{Z\sim \mu_{3t}\\ x\in \binom{[3t]}{t,t,t}, y\sim \mathrm{T}_0 x}}\big[\tilde{R}_Z(x) = \tilde{R}_Z(y)\big]\geq 1-O(\sqrt{\tau}+c)
\]
and analogously we have the same statement for $\mathrm{T}_1$
and $\mathrm{T}_2$, hence by the union bound
\[
\Pr_{\substack{Z\sim \mu_{3t}\\ x\in \binom{[3t]}{t,t,t}, y\sim \mathrm{T}_2\mathrm{T}_1\mathrm{T}_0 x}}\big[\tilde{R}_Z(x) = \tilde{R}_Z(y)\big]\geq 1-O(\sqrt{\tau}+c).
\]
Therefore, for at least 
$1-O(\tau^{1/4} + \sqrt{c})$ of $Z$, we 
have that 
\[
\Pr_{\substack{x\in \binom{[3t]}{t,t,t}, y\sim \mathrm{T}_2\mathrm{T}_1\mathrm{T}_0 x}}\big[\tilde{R}_Z(x) = \tilde{R}_Z(y)\big]\geq 1-O(\tau^{1/4} + \sqrt{c}),
\]
and we call such $Z$ decisive. Fix a decisive $Z$; 
the Markov chain $\mathrm{T}_2\mathrm{T}_1\mathrm{T}_0$ 
has second eigenvalue at most $1-\Omega(1)$, and 
thus from the above it follows that 
\[
\Pr_{x,y\in \binom{[3t]}{t,t,t}}\big[\tilde{R}_Z(x) = \tilde{R}_Z(y)\big]\geq 1-O(\tau^{1/4} + \sqrt{c}),
\]
and we define $R(Z)$ to be the most
popular value of $\tilde{R}_Z(x)$. Concluding, for 
decisive $Z$ we get  
\[
\Pr_{Z = u\cup v\cup w}
\big[R(z) = R(u)\circ R(v)\circ R(w)\big]
\geq 
1-O(\tau^{1/4}+c^{1/2}).
\]
Fix a decisive $Z$, and consider the
following direct product tester over $Z$: 
choose $Q\subseteq_{t/2} Z$, and then 
$Q\subseteq B, B'\subseteq_{t} Z$ such that $B\cap B' = Q$.
With probability at least $1-O(\tau^{1/4}+c^{1/2})$
we get that $R(B)_{Q} = R(Z)|_Q = R(B')|_{Q}$. 
Noting that sampling $Z\sim \mu_{3t}$ and then generating
$Q, B, B'$ yields a distribution of $(B,Q, B')$ that 
is $O(t^2/d) = o(1)$ close to the distribution of $Q, B, B'$
in the direct product tester in the lemma, so the conclusion
follows.
\end{proof}
\subsubsection{Concluding the global
structure}
With Lemma~\ref{lem:cobdry-to-global} in hand, we apply Theorem~\ref{thm:dinur-kaufman} 
to get that there exists a global function $G: X(1) \rightarrow \{0,1\}$ such that
\[
\Pr_{B \sim \mu_t}[G|_{B} = R(B)] 
\geq 1-O(\tau^{1/4}+c^{1/2}+\gamma).
\]
In the next lemma we analyze the agreement of $G$ with our lists 
$L[D]$ for $D\in X(d)$, thereby completing the proof
of Lemma~\ref{lem:list-agr-test}.
\begin{claim}
    $\Pr_{D\sim \mu_d}\big[\Delta(G|_{D}, L[D])\leq \frac{100\log(2\ell)}{t}\big]\geq 1-O(\tau^{1/4}+c^{1/2}+\gamma)$
\end{claim}
\begin{proof}
Sample $D\sim \mu_d$ and then $B\subseteq_{t} D$. 
Then $L[B] = L[D]|_B$ with probability at least 
$1-O(\tau)$, and $G|_{B}\in L[B]$ with probability
$1-O(\tau^{1/4}+c^{1/2}+\gamma)$, hence 
\[
\Pr_{B\subseteq_{t} D\sim \mu_d}\big[\Delta(G|_{B}, L[D]|_{B}) = 0\big]
\geq 1-O(\tau^{1/4}+c^{1/2}+\gamma).
\]
We get that with probability at least $1-O(\tau^{1/4}+c^{1/2}+\gamma)$ over the 
choice of $D$, it holds that 
$\Delta(G|_{B}, L[D]|_{B}) = 0$ 
with probability at least $1/2$ 
over the choice of $B$, and we argue that event 
in question holds for each such $D$.
To see that, first note that
fixing $f\colon D\to \{0,1\}$ 
such that $\Delta(G|_{D},f)\geq 100\log(\ell)/t$, 
it holds that $G|_{B} = f|_{B}$ with 
probability at most 
\[
\left(1-\frac{100\log(2\ell)}{t}\right)^t
\leq (2\ell)^{-100}.
\]
Thus, if $\Delta(G|_{D}, L[D])\geq 100\log(\ell)/t$,
then by the union bound 
$\Delta(G|_{B}, L[D]|_{B}) = 0$ 
with probability at most $(2\ell)^{-99} < 1/2$.
\end{proof}

\section{Proof of Theorem~\ref{thm:HDX_dp}}
In this section we outline 
the proof of 
Theorem~\ref{thm:HDX_dp}, restated formally below:

\begin{theorem}\label{thm:stronger_hdx_testing}
There is $c>0$ such 
that for all $\delta>0$,
there are $\xi,\eta,\gamma>0$, 
$m\in\mathbb{N}$, $C>1$ and $H\in\mathbb{N}$ such that
the following holds.
Suppose $k\in\mathbb{N}$ 
is such that $X$ is an 
$(m,k/C,\xi,c)$ weak UG coboundary expander, and 
$F\colon X(k)\to\{0,1\}^k$ 
passes the $(k,s)$ direct product tester over $X$ for 
$s = \eta k$ with probability at least $\delta$. 
Then there exists 
$f\colon X(1)\to\{0,1\}$
such that
\[
\Pr_{A\sim \mu_k}
[\Delta(F[A], f|_{A})\leq H/k]\geq \eta.
\]
\end{theorem}

The proof of 
Theorem~\ref{thm:stronger_hdx_testing} 
proceeds in exactly the same 
way as the proof of Theorem~\ref{thm:HDX_dp_weak}, 
except that we use a different 
agreement theorem over the 
Johnson scheme: we seek a global function $f$
that has stronger agreement with the assignment $F$. 
Namely, we want $\Delta(f|_{A}, F[A]) = O(1)$
for a large fraction of the 
$k$-faces $A$. Below is a 
formal statement of the 
Johnson agreement theorem
we need. 

\begin{theorem}\label{thm:johnson-dp}
For every $\eps>0$ there is $\alpha \in (0,1)$, $t \in \N$ such that for the following holds 
for sufficiently large $k$, 
and for $n$ sufficiently large
compared to $k$. If $F\colon \binom{[n]}{k}\to\{0,1\}^k$ passes the $(k, \alpha k)$ 
direct product test with probability $\eps$, then there exists $g$ such that,
\[\Pr_{A \subseteq_{k} [n]}[
F[A] \neq_{<\frac{t}{k}} g(A)] \geq \Omega(\eps^{12}).\]
\end{theorem}
\begin{proof}
The proof is deferred to Section~\ref{sec:johnson_thm}.
\end{proof}

\subsection{Deriving 
Theorem~\ref{thm:stronger_hdx_testing} from Theorem~\ref{thm:johnson-dp}}
Following the argument in Section~\ref{sec:basic_pf} 
with slight modifications and using 
Theorem~\ref{thm:johnson-dp} instead of Theorem~\ref{thm:DG}, 
one gets the``sufficient'' part of 
Theorem~\ref{thm:stronger_hdx_testing}. 
Below, we elaborate on the slight modifications that
are necessary.
\begin{enumerate}
    \item Following the argument in Section~\ref{sec:localizing}, we localize 
    the tester to $d$ faces again, and the test 
    we consider is the $(k,\alpha k)$ test (just like
    in the overall complex). 
    
    \item We run a procedure which is the same as the 
    short list algorithm in Section~\ref{sec:getting_a_list},
    however the parameters are a bit different (as the
    soundness of Theorem~\ref{thm:johnson-dp} is a $12$th
    power of $\eps$, as opposed to a $3$rd power of $\eps$
    as in Theorem~\ref{thm:DG}. The effect of that 
    is that all the powers in the description of the short 
    list algorithm grow by a factor of $4$.
    
    \item When we run the analog of Lemma~\ref{lem:agr-johnson} in our context, the distance parameter changes.
    This amounts to changing the $\eta$ in Lemma~\ref{lem:agr-johnson} (which is inherited from 
    the distance in Theorem~\ref{thm:DG}) to be $t/k$
    (which is the distance in Theorem~\ref{thm:johnson-dp}). 
    This change propogates throughout the argument.

    \item In Section~\ref{sec:list-agr-testing}, when we 
    project the lists onto level $r$, it is important 
    for us that the list sizes do not collapse. As seen
    in the proof of Claim~\ref{claim:projected-dist}, 
    this ultimately boils down to the fact that for 
    almost all of the $d$-faces $D$, the functions in 
    the list of $D$ are pairwise far from each other. 
    In the context of Section~\ref{sec:basic_pf} this
    distance is a constant fractional distance (hence
    we can reduce to a constant level $r$). In the context
    of this section though, our distance promise is milder
    and stands at $C/k$ for a large constant $C$, and 
    therefore we are able to only reduce to level $k/C'$ 
    and retain the non-collapse property of the list. 
    This is ultimately the reason that in Theorem~\ref{thm:stronger_hdx_testing} we 
    require UG coboundary expansion for a level which 
    is comparable to $k$.
\end{enumerate}
This summarizes the modifications required 
in Section~\ref{sec:basic_pf} to make the proof of
Theorem~\ref{thm:stronger_hdx_testing} go through.
\section{UG Coboundary Expansion for Known Complexes}\label{sec:known_complex}

\subsection{LSV Complexes}
In this section, we apply Theorem~\ref{thm:necessary-formal} to get that there exist LSV complexes~\cite{LSV1,LSV2} that do not support direct product tests with soundness $1/2$. We will make use
of the following result due to~\cite{KaufmanKazLub,EvraKaufman}. In topological 
language, the result asserts that there are LSV complexes with non-vanishing cohomology, and in fact elements in the cohomology must have large weight. We refrain from defining these notions and instead state their result in our language. We will need a quantitatively stronger version though, which was recently established by Dikstein and Dinur~\cite{DiksteinD23}.

\begin{theorem}\label{thm:evra-kaufman}
There exists $\mu > 0$, such that for all large enough $d \in \N$ and $\gamma > 0$, there is an infinite sequence of $d$-dimensional LSV complexes $\{X_n\}$ such that the following holds. For all $n$, the complex $X_n$ is a $\gamma$ two-sided local spectral expander and there exists a UG instance $\Psi_n = (G_1[X_n],\Pi_n)$ over $\F_2$ that is $1$-triangle consistent but $\val(\Psi_n) \leq 1-\mu$.
\end{theorem}

Using Theorem~\ref{thm:evra-kaufman} in conjunction with
Theorem~\ref{thm:necessary-formal} we get the following
corollary.

\begin{corollary}\label{cor:lsv}
There exists $\eps > 0$, such that for all $\eta > 0$, there is large enough $d$ and $k \leq d$ and an infinite sequence of $d$-dimensional LSV complexes $\{X_n\}_{n \in \N}$, such that  the following holds for each $X_n$. There is an assignment $F: X_n(k) \rightarrow \{0,1\}^k$ that passes the $(k,\sqrt{k})$ direct product tester with probability at least $1/2$, but for all $f: X(1) \rightarrow \{0,1\}$ we have that,
\[\Pr_{A\sim \mu_k}
\big[\Delta(F[A], f|_A) \geq \eps \big]\leq \eta.
\]
\end{corollary}

\begin{proof}
Using Theorem~\ref{thm:evra-kaufman}, for every $d, \gamma$ we can get a sequence of $\gamma$-spectral $d$-dimensional LSV complexes, and for each of them a locally consistent UG instance $\Psi_n$ on $G_1[X_n]$ with value at most $1-\mu$. We can now check that in fact for each $n$, we can generate a collection of lists on $X(1), X(3)$ such that $\Pi_n$ is strongly consistent with respect to the lists. To see this, associate each $u \in X(1)$ by the list $L(u) = \{0,1\}$ and each face $(u,v,w) \in X(3)$ by the list $\{0\cdot \pi(u,v)(0)\cdot\pi(u,w)(0), 1\cdot\pi(u,v)(1)\cdot\pi(u,w)(1)\}$. One can check that the latter list is well-defined ($L(u,v,w) = L(u,w,v)$ for example) because $\pi(u,w)=\pi(u,v)\pi(v,w)$. This shows that $\pi$ is strongly consistent with respect to all triangles.
Therefore $X_n$ is \emph{not} a $(m=2,r=1,\xi = 0, c=\mu)$ weak UG coboundary expander. Applying Theorem~\ref{thm:necessary-formal} we immediately get that there is a large enough $k$ such that there exists a function $F_n: X(k) \rightarrow \{0,1\}^k$ that passes the $(k,\sqrt{k})$ agreement test with probability $1/2$ (since $\xi = 0$) but for all global functions, $\Delta(F[A],f|_A) \geq \eps$ holds with probability at most $\eta$ over $A \sim \mu_k$.
\end{proof}

\subsection{Johnson and Grassmann Complexes}
In this section we verify that the Johnson and Grassmann complexes are UG coboundary expanders as defined in Definition~\ref{def:ug_coboundary_expander_simplified}. They are known to be standard coboundary expanders, but we give the proof that they satisfy UG coboundary expansion for completeness.

Let $\cC_n$ denote the complete complex with vertices being $[n]$. Let $\cC(i)$ denote the faces at the $i^{th}$ level.

\begin{lemma}\label{lem:johnson-coboundary}
For all $n \in N$, the complete complex $\cC_n$ is a $(m,r,\xi,\xi+o(1))$ UG coboundary expander for all $m \in \N$, $r = o(\sqrt{n})$ and $\xi \in (0,1)$.  
\end{lemma}

\begin{proof}
We will drop the subscript $n$ in $\cC_n$ henceforth.
Consider any UG instance $\Psi = (G_r[\cC],\Pi)$ with $1-\xi$-fraction of locally consistent triangles. It suffices to show that $\Psi$ has an assignment that satisfies a $1-\xi-o(1)$ fraction of the edges.

Let $\cT$ denote the uniform distribution over all triangles in $G_r[\cC]$ and for a face $U \in \cC(r)$, let $\cT_U$ denote the distribution over triangles that contain $U$. By averaging we know there exists $U \in \cC(r)$ such that,
\[\Pr_{(U,V,W) \sim \cT_U}[(U,V,W) \text{ is consistent}] \geq 1-\xi.\]
Let $N_U$ denote the neighborhood of $U$ in $G_r[\cC]$, i.e. the set of faces $V \in \cC(r)$ that are disjoint to $U$. One can check that $N_U$ has fractional size at least $1-r^2/n = 1-o(1)$ as $r = o(\sqrt{n})$. Now define an assignment of permutations $P: \cC(r) \rightarrow S_m$ as follows:
\begin{equation}
P(V) = \begin{cases}
\text{id} ~~~~~~~~~~~ \text{ if } V = U,\\
\pi(u,v) ~~~ \text{ if } V \in N_U\\
\text{id} ~~~~~~~~~~~\text{ otherwise.}
\end{cases}    
\end{equation}

Note that by definition $P$ satisfies all the edges incident on $U$ and in fact it is easy to check that $P$ also satisfies all the edges $(V,W)$ where $(U,V,W)$ forms a consistent triangle. Therefore it suffices to lower bound the measure of such edges:
\begin{align*}
\val(P) &\geq \Pr_{(V,W) \sim G_r[\cC]}[V,W \in N_U, (U,V,W) \text{ is consistent }] \\
&\geq \left(1-2\frac{r^2}{n}\right)\Pr_{(V,W) \sim G_r[\cC]}[(U,V,W) \text{ is consistent } \mid V,W \in N_U] \\
&\geq (1-o(1))(1-\xi),
\end{align*}
where in the first inequality we used that the fractional size of $N_U$ is at least $1-r^2/n$, and in the last inequality we used the fact that the distribution over triangles obtained by sampling a random edge in the neighborhood of $U$ is equal to the distribution $\cT_U$. This shows the existence of an assignment with value $1-\xi-o(1)$ as required.
\end{proof}

Let $\text{Gr}_{n,q}$ denote Grassmann complex over the ambient space $\F_q^n$ where the $i$-level faces in the complex for $i \leq n$ are the $i$-dimensional linear subspaces of $\F_q^n$. 
Note that this is not a simplicial complex but we can still consider the graph $G_r[\text{Gr}_{n,q}]$ -- the vertices are the $r$-level faces of the complex and we have edges between two subspaces that are disjoint.

\begin{lemma}
For all finite fields $\F_q$ and $n \in N$, the Grassmann complex $\text{Gr}_{n,q}$ is a $(m,r,\xi,\xi+o(1))$ UG coboundary expander for all $m \in \N$, $r = o(\log_q n)$ and $\xi \in (0,1)$.  
\end{lemma}

\begin{proof}
The proof of this lemma is exactly the same as that of Lemma~\ref{lem:johnson-coboundary}, except that we note that the neighborhood of every vertex in $G_r[\text{Gr}_{q,n}]$ has size at least $1-q^r/q^n = 1-o(1)$ if $r = o(\log_q n)$. Therefore by following the rest of the proof we get that for any UG instance that is $1-\xi$ triangle consistent, there exists an assignment with value at least $1-\xi-o(1)$, as required.   
\end{proof}

\section{Acknowledgements}
We thank Mehtaab Sawhney for suggesting the use of the weak regularity lemma and Nikhil Vyas for bringing~\cite{AlonVKK02} to our attention, which led to a considerable simplification on our earlier attempts at proving Lemma~\ref{lem:no-jump}. We thank Tali Kaufman for helpful comments on an earlier version of this paper. We also thank the Simons institute for hosting us during the ``Analysis and TCS: New Frontiers'' program, where part of the work was carried out.

\bibliographystyle{alpha}
\bibliography{references}

\appendix
\section{Proofs of Omitted Claims}
\begin{claim}\label{claim:strong-to-weak-consistency}
If a set of lists of size $m$ and permutations in $S_m$, $(\cL,\Pi)$ on $X(t),X(2t),X(3t)$ are $1-\xi$ consistent according to Definition~\ref{def:triangle_consistent_strong}, then $\Pi$ is $1-3\xi$ consistent according to Definition~\ref{def:triangle_consistent_weak}.    
\end{claim}

\begin{proof}
Let $\cT$ be the distribution over triangles generated by sampling $T \sim X(3t)$ and then randomly splitting it into $(u,v,w)$. We can use a union bound to get that with probability $1-3\xi$, $(u,v,w) \sim \cT$ satisfies,
\begin{align*}
L'(T) &= L(u)\circ \pi(u,v) L(v) \circ \pi(u,w) L(w) \\
&=L(w)\circ \pi(w,u) L(u) \circ \pi(w,v) L(v)\\
&=L(v)\circ \pi(v,w) L(w) \circ \pi(v,u) L(u).
\end{align*}
Note that for each $i\in [m]$, the first line asserts
that the assignment 
$\pi(u,w)(i)$ in the list
of $L(u)$ is consistent
with the assignment $i$ 
of $w$ with respect to 
$L'(T)$. By the third 
line, this assignment
is consistent with the assignment
$\pi(v,u)\circ \pi(u,w)(i)$ 
for $L(v)$, and by the 
second line this is consistent
with the assignment $\pi(w,v)\circ \pi(v,u)\circ \pi(u,w)(i)$ for $w$. As 
$L'(T)$ contains $m$ distinct assignments, we conclude that it must be the case that $\pi(w,v)\circ \pi(v,u)\circ \pi(u,w)(i) = i$ for all $i\in[m]$, and so
$\pi(w,u)=\pi(u,w)^{-1} = \pi(w,v)\pi(v,u)$ for such a triangle $(u,v,w)$. 
\end{proof}

\section{Theorem~\ref{thm:HDX_dp_weak} with Improved UG Coboundary Parameters}\label{sec:appx-improved-cbdry}
In this section we outline the proof of the sufficient theorem that has weaker requirements on the coboundary expansion parameters. We show that any local spectral expander that is an $(m,r,\exp(-o(r)),c)$ UG coboundary expander supports a direct product test with low soundness.
\begin{theorem}\label{thm:stronger-coboundary_formal}
There is $c>0$ such that for all $\eps,\delta>0$ there is $\eta>0$ and $m,r\in\mathbb{N}$ such that for sufficiently large $k$, sufficiently large $d$ and $\gamma$ small enough function of $d$\footnote{For concreteness, we show that being an $(m,r,\exp(-\sqrt{r}),c)$ UG coboundary expander suffices. This argument  requires $k \geq \text{tower}_{\poly(1/\delta)}(1/\delta)$ and $d \geq \exp(\exp(k))$. It is possible to get improved dependence of $\delta$ on $k$, while getting worse parameters in the coboundary expansion, but we refrain from stating these versions.}, the 
following holds. If a $d$-dimensional simplicial complex $X$ is a $\gamma$-spectral expander and $(m,r,\exp(-o(r)),c)$ weak UG coboundary expander, then the direct product test over $X$ with respect
to sufficiently large $k$ has soundness $\delta$. Namely, if 
$F\colon X(k)\to\{0,1\}^k$ passes the
$(k,\sqrt{k})$ direct product tester with respect to $X$ with probability at least $\delta$, 
then there is $f\colon X(1)\to\{0,1\}$
such that 
\[
\Pr_{A\sim \mu_k}
[\Delta(F[A], f|_A)\leq \eps ]\geq \eta.
\]
\end{theorem}

The proof of Theorem~\ref{thm:stronger-coboundary_formal} proceeds in the same way as the proof of Theorem~\ref{thm:HDX_dp_weak} except that we show how to modify the argument to get that the list agreement test in Lemma~\ref{lem:agr-to-list-agr} passes with probability $1-f(\delta)$ instead of $1-\poly(\delta)$, where $f(\delta) \ll \poly(\delta)$. Using this stronger lemma, we can improve the parameters in the list-agreement testing result to show that a weaker assumption on the coboundary expansion of the complex suffices. 

We now outline these details, starting with the modified agreement to list agreement reduction. For concreteness we will show that it suffices to have UG coboundary expansion with parameters $(m,r,\exp(-\sqrt{r}),c)$. We emphasize that the same argument works to show that it suffices to have UG coboundary expansion with parameters $(m,r,\exp(-o(r)),c)$, example for, $o(r) = r/\log\log r$, by paying in the dependence of $k$ on $\delta$.

For all $i \in [1/\delta^{80}]$ let us set:
\begin{equation}\label{eq:eta}
\eta_0 = 1/\text{tower}_{O(1/\delta^{80})}(1/\delta), \eta_1 \geq \Omega(1/\log(1/\eta_0)), \eta_2 \geq \Omega(\log(1/\eta_1)), \ldots    
\end{equation}
We can check that for all $i \leq 1/\delta^{80}$, $\eta_i \leq \delta$.

\begin{lemma}
\label{lem:agr-to-list-agr-improved-cbdry}
For all $\delta>0$, 
for sufficiently large $k,d \in \N$,
sufficiently small $\gamma$ compared to 
$d$, some $i \in [1/\delta^{80}]$ and $\tau' = O(\eta_i)$, the following holds. Suppose that $X$ is a $d$-dimensional simplicial complex which is a $\gamma$-spectral expander, and $F: X(k) \rightarrow \{0,1\}^k$ passes the $(k,\sqrt{k})$-agreement-test~\ref{agr-test-hdx} with probability $\delta$. Then, there exists lists $(L[D])_{D \in X(d)}$ satisfying: 
\begin{enumerate}
\item 
\textbf{Short, non-empty lists:} With probability $1-O(\tau')$ over the choice of $D\sim X(d)$, the list $L[D]$ is non-empty and has size at most $O(1/\delta^{12})$.
\item  
\textbf{Good agreement:} For all $D\in X(d)$ and every $f \in L[D]$, we have that $\agr_\nu(f, F|_D) \geq \Omega(\delta^{12})$ for $\nu = 1/k^{\Omega(1)}$.
\item 
\textbf{Distance in the lists:} With probability at least $1-O(\tau')$ over the choice of $D\sim X(d)$, the list $L[D]$ has distance at least $\Omega(\eta_{i+1})$.
\end{enumerate}
Furthermore the lists above pass the List-Agreement-Test~\ref{list-agr-test-hdx} with parameter $\eta_i$, with probability $1-\tau'$. 
\end{lemma}

\begin{proof}
The proof of this lemma proceeds in the same way, albeit with some changes in parameters and an improved argument for (3).

\paragraph{Changes to Lemma~\ref{lem:algo1}:}
Firstly in the short list algorithm~\ref{algo-gap}, recall that we had set $t = k^{-c}$. We now choose $r$ from a larger range, $\{1,\ldots, \delta^{20}2^{k^{c'}}\}$, for $c'$ chosen appropriately small as a function of $c$. At each step we decrement $\delta_i$ only by $2^{-k^{c'}}$, i.e. $\delta_{i+1} = \delta_i - 2^{-k^{c'}}$, therefore the gaps between consecutive $\delta_i$'s is now smaller. The list-size bound in Lemma~\ref{lem:algo1} remains as is, and really the only change is to point (2) which now says that  for all $i \in I_1$, $\agr_t(g, \widetilde{G}_i) > \delta' - i2^{-k^{c'}}$. 

\paragraph{Changes to Lemma~\ref{lem:agr-johnson}:}
The changes in $r$ and the decrement value now propagate to Lemma~\ref{lem:agr-johnson}. We choose $i \sim [1/\delta^{80}]$ as before, and run the short list algorithm with the parameters $r$ and $\eta_i$ (from \eqref{eq:eta}). We will show that in fact points (1), (2), (4), (5) of Lemma~\ref{lem:agr-johnson} hold with probability $1-o(1)$ over the choice of $r, i$ and point (3) now says that we have distance in the list is at least $\eta_{i+1}$ (as set in \eqref{eq:eta}) with probability $1-\poly(\delta)$ over the choice of $i$. The only non-trivial change here is while proving (5) hence we discuss that below.

For (5), with probability at least $1-2^{-\Omega(k^{c'})}$ over the choice of $r$ we get that $r+1 \not\in I_1$, i.e. $\forall h$, $\agr_t(h,\widetilde{G}_{r+1}) < \delta_{r+1}$. Taking $d$ to be much larger ($\exp(\exp(k))$), and applying Lemma~\ref{lem:no-jump} we can  ensure that the maximum agreement on $B \subset_{d/2} D$ is at most $\delta_{r+1}+2^{-k^{c'}}$ with probability $1-o(1)$ over $B$. Since $c'$ has been chosen to be small enough the rest of the proof for Lemma~\ref{lem:agr-johnson} (5) suffices to prove that (5) holds for $L[D]$ with probability $1-\exp(-k^{c'})$ over $r$.\\

We are now ready to prove the improved lemma statement using the above changes.

\paragraph{Proofs of (1), (2):} First note that points (1) and (2) follow immediately from the modifications to Lemma~\ref{lem:agr-johnson} that we discussed above.

\paragraph{Proof of (3):} We will now argue that for some choice of $r,i$, $1-O(\eta_i)$ of the lists $L[D]$ have distance $> \eta_{i+1}$. First by linearity of expectation, fix $r,i$ for which points 1,2,4,5 hold for $1-o(1)$ fraction of the lists $L[D]$ and $i$ for which $1-O(\delta^{68})$ fraction of the lists have distance $\geq \eta_{i+1}$. Most of the ingredients of this proof come from Section~\ref{sec:consistency-of-lists} hence we skip the obvious details.

As in Section~\ref{sec:consistency-of-lists}, let $\cD$ be the set of good $d$-faces (where the $(k,\sqrt{k})$ agreement test passes with probability at least $\delta'$), but define the set of very good $d$-faces to be those that are good and where points (1),(2), (4), (5) of Lemma~\ref{lem:agr-johnson} hold. Note that a very good $d$-face may not have large distance on its list. Call a triple $(D, B, D')$ good if:
\begin{enumerate}
\item $D, D'$ are very good.
\item For all $f \in L[D]$, there exists $g \in L[B]$ with $\Delta(f|_B,g) < \eta_i$, and for all $g \in L[B]$ there exists $f \in L[D]$ with $\Delta(g,f|_B) < 2\eta_i$. The same holds when $D$ is replaced by $D'$.
\end{enumerate}

First note that $1-o(1)$-fraction of the triples are good. On a good triple we get that for each $f \in L[D]$, there exists an $f' \in L[D']$ (possibly many such functions) s.t. $\Delta_B(f,f') \leq 3\eta_i$ and vice versa. Because of this correspondence between the lists of $D, D'$, it is easy to see that $|\Delta(L[D]) - \Delta(L[D'])| < O(\eta_i)$. Now consider the graph $G$ on $X(d)$ given by this random walk (as done in Claim~\ref{claim:list-size} for instance). By Lemma~\ref{lem:spectral_gap_of_graphs_from_HDX} the second eigenvalue of $G$ is at most $1/2+O(d^2\gamma)$ and hence the associated Laplacian has its second smallest eigenvalue lower bounded by $1/2-O(d^2\gamma)$. Consider the function $V$ on $X(d)$, where $V[D]$ is defined as $\Delta(L[D])$. Taking the quadratic form of the Laplacian of $G$ with the function $V$, we therefore get that,
\[\E_{(D,D')\sim E(G)}[|\Delta([L(D)]) - \Delta(L[D'])|^2] \geq \Omega(\text{Var}(V)).\]
Now note that the LHS of the above equation is at most $O(\eta_i^2)$, which implies that the variance of $V$ is small. By an averaging argument this gives that $|V[D] - \E_D[V[D]]| \leq O(\sqrt{\eta_i})$ with probability at least $1-O(\eta_i)$.

Since we know that $\Delta(L[D]) > \eta_{i+1}$ for at least $1-O(\delta^{68})$ fraction of $D$'s, $\E_D[V[D]] \geq \Omega(\eta_{i+1})$. We can now conclude that with probability $1-O(\eta_i)$ it must be the case that $\Delta(L[D]) > \Omega(\eta_{i+1}) - O(\sqrt{\eta_i}) \geq \Omega(\eta_{i+1})$ since $\eta_{i+1} \gg \eta_i$.

\paragraph{Proof that List-Agreement Test passes:} Given that $1-O(\eta_i)$-fraction of the lists also have distance $\geq \Omega(\eta_{i+1})$, the proof that the test passes with probability $1-O(\eta_i)$ now exactly follows along the lines of Section~\ref{sec:consistency-of-lists}.
\end{proof}

We can now show that the list agreement test is sound given weaker assumptions on the coboundary expansion.

\begin{lemma}\label{lem:list-agr-test-improved-cbry}
Assume there exists a collection of lists $\{L[D]\}_{D \in X(d)}$ that satisfy the premise of Lemma~\ref{lem:agr-to-list-agr}, and assume that $X$ is a $\gamma$-spectral expander for $\gamma < 1/\poly(d)$ and a weak $(O(1/\delta^{12}),t, \exp(-\sqrt{t}), c)$\footnote{As noted before, we can actually handle any function $\exp(-o(t))$, but we look at $\exp(-\sqrt{t})$ for concreteness.} UG coboundary expander for $t=\Theta(1/\eta_{i+1}^2)$.
Then there exists $G: X(1) \rightarrow \{0,1\}$ such that
\[
\Pr_{D \sim X(d)}\left[\Delta(G(D),L[D]) \leq \delta\right] \geq 1-O(c^{1/2} + \exp(-\sqrt{t}) + \gamma).\]
\end{lemma}

\begin{proof}
First note that Claim~\ref{claim:list-equality} holds with $\tau$ replaced with $\tau' = o(1)$. Next, we specify the changes to Claims~\ref{claim:list-size} and~\ref{claim:projected-dist}. 

\paragraph{Changes to Claim~\ref{claim:list-size}:}
We get that, for $\sqrt{t} \geq \Omega(1/\eta_{i+1})$ and $\eta_i \leq \exp(-\sqrt{t})$ it holds that: 
\[\Pr_{\substack{D \sim \mu_d \\B \subset_t D}}\big[|L(D)|_B| \neq \ell\big] \leq \exp(-\sqrt{t}).\] The proof of this is the same as Claim~\ref{claim:list-size}.

\paragraph{Changes to Claim~\ref{claim:projected-dist}:}
We have that, for $t \leq \frac{\exp(-\sqrt{t})}{\eta_{i}}$, 
with probability at least $1-\exp(-\sqrt{t})$ 
over the choice of $B\sim \mu_t$
it holds that 
\[
\Pr_{\substack{D,D' \supseteq_d~ B}}[L[D]|_B = L[D']|_B] \geq 1-\exp(-\sqrt{t}).
\]
The change in parameters propagates to the proof, which  changes nothing but the calculations.

The requirements on $t$ from the above claims can be summarized as:
\[\Omega(1/\eta_{i+1}^2) \leq t \leq \frac{\exp(-\sqrt{t})}{\eta_{i}},\]
which is possible to satisfy if $\eta_{i+1} \geq \Omega(1/\log(1/\eta_i))$.

For the rest of the proof set $\tau = \exp(-\sqrt{t})$. Above we've shown that Claims~\ref{claim:list-size} and~\ref{claim:projected-dist} hold with the parameter $\tau$. The rest of the proof of the lemma proceeds exactly in the same way with $\tau$ in Lemma~\ref{lem:list-agr-test} replaced with $\exp(-\sqrt{t})$ everywhere.
\end{proof}

The proof of Theorem~\ref{thm:stronger-coboundary_formal} now follows immediately by combining Lemmas~\ref{lem:agr-to-list-agr-improved-cbdry} and~\ref{lem:list-agr-test-improved-cbry}.

\section{The Johnson Agreement Theorem: Proof of Theorem~\ref{thm:johnson-dp}}\label{sec:johnson_thm}
In this section we give the proof of Theorem~\ref{thm:johnson-dp}, and throughout 
we use the parameters:
\[
  0\ll R^{-1}\ll h^{-1}\ll \nu\ll \alpha\ll \eps<1.
\]
The proof follows the lines of the argument 
in~\cite{ImpagliazzoKW09}, except that in the 
end we use small-set expansion type arguments
to glue together the local decoded functions
similarly to the way it is done in~\cite{BKM3}.
\subsection{Auxiliary Claims}
We need a few standard auxiliary claim. The first
claim asserts that if two functions are far, then
they disagree on many $k$-sets.
\begin{claim}\label{claim:trivial_distance_to_diff}
  If $g,h\colon [n]\to \{0,1\}$ are functions such that $\Delta(g,h)\geq R/k$, then
  \[
  \Pr_{A\subseteq_k[n]}\big[g|_{A} = h|_{A}\big]\leq e^{-R}.
  \]
\end{claim}
\begin{proof}
  The probability is at most
  \[
  \left(\frac{n-n\Delta(g,h)}{n}\right)\cdot\left(\frac{n-1-n\Delta(g,h)}{n-1}\right)
  \cdots \left(\frac{n-(k-1)-n\Delta(g,h)}{n-(k-1)}\right)
  \leq (1-\Delta(g,h))^{k}
  \leq e^{-R}.
  \]
\end{proof}

The second claim asserts that if we have a graph in 
which the second singular value is small, then any
set of vertices which is not-too-small contains a 
sizable number of edges.
\begin{claim}\label{claim:repalce_CS}
  Suppose that $M$ is the normalized adjacency matrix of a graph $G$, and the second singular value of $M$ is at most 
  $\delta$. Then for any set of vertices $S$ of relative size $\eps$ we have
  \[
  \Pr_{x, y\sim M x}\big[ x\in S, y\in S\big]\geq \eps^2 - \delta\eps.\qedhere
  \]
\end{claim}
\begin{proof}
  The probability can be written as 
  \[
    \eps\langle 1_S, M1_S\rangle
    =\eps^2 + \langle 1_S -\eps, M(1_S-\eps)\rangle
    \geq \eps^2 - \|1_S\|\|M(1_S-\eps)\|_2
    \geq \eps^2 - \sqrt{\eps}\delta\sqrt{\eps}
    =\eps^2 - \eps\delta.
  \]
\end{proof}
\subsection{Local Structure}
We will think of a $k$-set as partition as $A_0\cup B_0$ where $A_0$ has size $\alpha k$ and $B_0$ has size $(1-\alpha)k$.
For a partition of $k$-set $(A_0,B_0)$, we define
\[
{\sf Cons}(A_0,B_0) = \big\{B\subseteq [n]\setminus A_0~\big|~F[A_0\cup B_0]|_{A_0} = F[A_0\cup B]|_{A_0}\big\}.
\]
We use a few definitions and results from~\cite{ImpagliazzoKW09}.
\subsubsection{Good and Excellent}
\begin{definition}
  We say $(A_0,B_0)$ is good if $\Pr_{B\subseteq_{(1-\alpha)k}[n]\setminus A_0}\big[B\in {\sf Cons}(A_0,B_0)\big]\geq \frac{\eps}{2}$.
\end{definition}
We will think of $B$ as being split into $D\cup E$ where $|E| = \alpha k$ and $|D| = (1-2\alpha) k$.
\begin{definition}
  We say $(A_0,B_0)$ is excellent if it is good, and furthermore
  \[
    \Pr_{\substack{B_1 = (D_1\cup E)\subseteq [n]\setminus A_0\\ B_2 = (D_2\cup E)\subseteq [n]\setminus A_0}}
    \big[(D_1,E), (D_2,E)\in {\sf Cons}(A_0,B_0)\land F[A_0\cup D_1\cup E]|_{E}\neq_{>h} F[A_0\cup D_2\cup E]|_{E}\big]
    \leq \nu.
  \]
\end{definition}
\begin{lemma}\label{lem:good_excellent}
    The following holds:
  \begin{enumerate}
    \item A randomly chosen $(A_0,B_0)$ is good with probability at least $\eps/2$.
    \item A randomly chosen good $(A_0,B_0)$ is excellent with probability at least $1-\frac{2^{-\Omega(h)}}{\nu}$.
  \end{enumerate}
\end{lemma}
\begin{proof}
  These are~\cite[Lemma 3.5]{ImpagliazzoKW09}, \cite[Lemma 3.6]{ImpagliazzoKW09}.
\end{proof}

\subsubsection{Decoding Local Structure}
\begin{lemma}\label{lemma:excellent_agree}
  Suppose $(A_0,B_0)$ is excellent. Then there exists $g_{A_0,B_0}\colon [n]\to\{0,1\}$ such that
  \[
  \Pr_{B\in {\sf Cons}(A_0,B_0)}\big[g_{A_0,B_0}(B) \neq_{>2R/k} F[A_0\cup B]|_{B}\big]\leq \sqrt{\nu}.
  \]
\end{lemma}
\begin{proof}
  This is~\cite[Lemma 3.8]{ImpagliazzoKW09}.
\end{proof}

\subsection{Gluing Local Structure}
\subsubsection{Getting correlation over the global functions}
Consider the following joint distribution over pairs of $(A_0,B_0)$ and $(A_0', B_0')$,
which we call $\mathcal{D}$:
\begin{enumerate}
  \item Sample $\tilde{A}\subseteq_{\sqrt{\alpha} k}[n]$. 
  \item Sample $\ell$ from a binomial distribution ${\sf Bin}(\alpha k, \sqrt{\alpha})$ and $A_0, A_0'\subseteq_{\alpha k} \tilde{A}$ conditioned on $|A_0\cap A_0'| = \ell$. 
  \item Sample $B_0, B_0'\subseteq_{(1-\alpha) k}[n]\setminus \tilde{A}$.
\end{enumerate}
The following lemma shows that the functions $g_{A_0,B_0}$ and $g_{A_0', B_0'}$ for pairs generated in this way are very close to each other with
noticeable probability.
\begin{claim}\label{claim:similar_restrictions}
  It holds that
  \[
    \Pr_{((A_0,B_0), (A_0', B_0'))\sim \mathcal{D}}\left[\Delta(g_{A_0, B_0}, g_{A_0', B_0'})\lll \frac{R\log(1/\eps)}{k} \right]\ggg \eps^6.
  \]
\end{claim}
\begin{proof}
Choose $\tilde{A}$ of size $\sqrt{\alpha} k$, choose $B_0\subseteq_{(1-\alpha)k} [n]\setminus \tilde{A}$ and $B'\subseteq_{(1-\sqrt{\alpha})k}[n]\setminus \tilde{A}$ independently,
and partition $\tilde{A}$ randomly as $A_0\cup A_1$ where the size of $A_0$ is $\alpha k$. Then
\[
\E_{\substack{\tilde{A}\\ B'\subseteq_{(1-\sqrt{\alpha})k} [n]\setminus \tilde{A}}}\Big[\E_{\tilde{A} = A_0\cup A_1, B_0}\big[1_{(A_0, B_0)\text{ is excellent}}
\cdot 1_{g_{A_0, B_0}(A_1\cup B')\neq_{\leq 2R/k} F[A_0\cup A_1\cup B']|_{A_1\cup B'}}\big]\Big]\ggg \eps^2,
\]
where we used the fact that the distribution of $(A_0,B_0)$ is uniform, hence by Lemma~\ref{lem:good_excellent}
it is excellent with probability at least $\eps/4$. Also, conditioned on $A_0,B_0$, the distribution of $A_1\cup B'$
is $O(k^2/n)=o(1)$ close to uniform; thus, it is in ${\sf Cons}(A_0,B')$ with probability at least
$\eps/2 - o(1)$. Conditioned on that, the distribution of 
$A_1\cup B'$ is $o(1)$ close to uniform in ${\sf Cons}(A_0,\tilde{B})$, so
by Lemma~\ref{lemma:excellent_agree} we have that $g_{A_0, B_0}(A_1\cup B')\neq_{\leq R/k} F[A_0\cup A_1\cup B']|_{A_1\cup B'}$ with probability $1-o(1)$. 

It follows by an averaging argument that with probability at least $\Omega(\eps^2)$ over the choice of 
$\tilde{A}, B'$, we have that
\begin{equation}\label{eq:replace_CS_argument}
\Pr_{A_0\subseteq \tilde{A}, B_0}\Big[g_{A_0, B_0}(B')\neq_{\leq 2R/k} F[\tilde{A}\cup B']|_{B'}\Big]
\ggg\eps^2,
\end{equation}
and we fix such $\tilde{A}, B'$. Consider the graph induced by the sampling procedure $\mathcal{D}$, i.e. its vertices are 
$(A_0, B_0)$ and the weight of an edge $(A_0, B_0)$ and $(A_0', B_0')$ is proportional to the probability this is sampled from $\mathcal{D}$
conditioned on $\tilde{A}$. Note that this is a product graph, where one graph is the $A_0$ vs $A_0'$ graph, and the other graph 
is the $B_0$ vs $B_0'$ graph. We argue that the absolute value of the largest non-trivial eigenvalue in each one of these graphs is at most $\sqrt{\alpha}$,
and hence the same holds for the product graph. As both arguments are essentially the same, we explain the argument for the $A_0$ vs $A_0'$ graph.
The adjacency matrix of the $(A_0,A_0')$ graph is $\mathrm{M} = \sum\limits_{\ell} p(\ell) M_{\ell}$ where $p(\ell) = \Pr\big[{\sf Bin}(\alpha k, \sqrt{\alpha}) = \ell\big]$
and $M_{\ell}$ is the adjacency matrix of the Johnson graph with intersection parameter $\ell$. The matrices $M_{\ell}$ have the same eigenvectors, and the largest non-trivial
eigenvalue of $M_{\ell}$ in absolute value is $\frac{\ell}{\alpha k}$, hence the largest non-trivial eigenvalue of $\mathrm{M}$ is 
$\frac{1}{\alpha k}\sum\limits_{\ell} p(\ell)\ell = \sqrt{\alpha}$. Combining this observation with~\eqref{eq:replace_CS_argument} via Claim~\ref{claim:repalce_CS}
yields that
\[
\Pr_{(A_0,A_0'), (B_0, B_0')}\Big[g_{A_0, B_0}(B')\neq_{\leq 2R/k} F[\tilde{A}\cup B']|_{B'}, g_{A_0', B_0'}(B')\neq_{\leq 2R/k} F[\tilde{A}\cup B']|_{B'}\Big]
\ggg\eps^4,
\]
so by the triangle inequality
\[
\Pr_{(A_0,A_0'), (B_0, B_0')}\Big[g_{A_0, B_0}(B')\neq_{\leq 4R/k} g_{A_0', B_0'}(B')\Big]
\ggg\eps^4
\]
Taking expectation over $\tilde{A}$ and $B'$ gives (noting that the distribution over $B'$ is a
uniform subset of $[n] \setminus \tilde{A}$ of size $(1-\sqrt{\alpha})k$, hence it is $O(k^2/n) = o(1)$ close to uniform)
\[
\E_{((A_0,B_0), (A_0', B_0'))\sim \mathcal{D}}\left[\E_{B'}\big[
1_{
g_{A_0, B_0}(B')
\neq_{\leq 4R/k} 
g_{A_0', B_0'}(B')}\big]\right]\ggg \eps^6.
\]
By an averaging argument, with probability at least $\Omega(\eps^6)$ over the choice of $((A_0,B_0), (A_0', B_0'))\sim \mathcal{D}$,
we get that $\E_{B}\big[1_{g_{A_0, B_0}(B)\neq_{\leq 4R/k} g_{A_0', B_0'}(B)}\big]\ggg \eps^6$. In that case, sampling $C\subseteq [n]$
of size $k/10 R$ we get that $g_{A_0, B_0}|_{C} = g_{A_0', B_0'}|_{C}$ with probability $\Omega(\eps^6)$, and by Claim~\ref{claim:trivial_distance_to_diff}
it follows that $\Delta(g_{A_0, B_0}, g_{A_0', B_0'})\lll R\log(1/\eps)/k$.
\end{proof}

\subsubsection{Applying Expansion}
We identify pairs $(A_0,B_0)$ with points in the multi-slice $\binom{[n]}{\alpha k, (1-\alpha)k, n-k}$ in the obvious way:
a pair $(A_0, B_0)$ is identified with a point $x\in \binom{[n]}{\alpha k, (1-\alpha)k, n-k}$ whose set of coordinates
$i$ such that $x_i = 0$ is $A_0$, whose set of coordinates $i$ such that $x_i = 1$ is $B_0$, and whose set of coordinates
$i$ such that $x_i = 2$ is $[n]\setminus (A_0\cup B_0)$. Thus, the distribution $\mathcal{D}$ defines a Markov chain
$\mathrm{T}$ over $\binom{[n]}{\alpha k, (1-\alpha)k, n-k}$. To simplify notation, we denote the function $g_{A_0,B_0}$ by $g_x$, 
where $x$ is the point in the multi-slice associated with $(A_0,B_0)$. Using these notations,  Claim~\ref{claim:similar_restrictions} implies
that sampling $x$ uniformly and then $y\sim \mathrm{T} x$ we have that $\Delta(g_x, g_y)\lll R\log(1/\eps)/k$ with probability
$\Omega(\eps^6)$. Looking at the definition of $\mathcal{D}$, it is clear that $\mathrm{T}$ is symmetric.
\begin{claim}\label{claim:eigenvalue_bound}
  $\lambda_2(\mathrm{T})\lll \alpha^{\Omega(1)}$.
\end{claim}
\begin{proof}
  The proof is deferred to Section~\ref{sec:proof_of_ev_bound_dp_johnson}.
\end{proof}
Claim~\ref{claim:eigenvalue_bound} gives us the following small set expansion conclusion:
\begin{claim}\label{claim:sse_conclusion}
Suppose $f\colon \binom{[n]}{\alpha k, (1-\alpha)k, n-k}\to \{0,1\}$ has expectation $p$. Then
\[
\left\langle f, \mathrm{T} f\right\rangle
\leq p^{1.5} + \alpha^{\Omega(1)} p.
\]
\end{claim}
\begin{proof}
  By Cauchy-Schwarz and Parseval
  \[
  \left\langle f, \mathrm{T} f\right\rangle^2
  \leq p\|\mathrm{T} f\|_2^2
  =p\left(\|\mathrm{T}(f - p)\|_2^2 + p^2\right)
  \leq
  p\left(\alpha^{\Omega(1)}\|f - p\|_2^2 + p^2\right)
  \leq p^2\alpha^{\Omega(1)} + p^3,
  \]
  and the result follows.
\end{proof}

The following claim asserts that the functions $g_x$ are close to each other with noticeable probability.
\begin{claim}\label{claim:globally_close_to_each_other}
It holds that $\Pr_{x,y}\big[\Delta(g_x,g_y)\lll R^2\log(1/\eps)^2/k\big]\ggg \eps^{12}$.
\end{claim}
\begin{proof}
  From Claim~\ref{claim:similar_restrictions} we get that
  \[
  \E_{x,y\sim \mathrm{T} x}
  \left[\E_{S\subseteq_{k/ R^2\log(1/\eps)} [n]}\big[1_{g_x|_S \equiv g_y|_{S}}\big]\right]\ggg \eps^6.
  \]
  For each $S\subseteq [n]$ of size $k/ R^2\log(1/\eps)$ and $v\in \{0,1\}^S$, let $f_v\colon \binom{[n]}{\alpha k, (1-\alpha)k, n-k}\to\{0,1\}$
  be defined by $f_{S,v}(x) = 1$ if and only if $g_x|_{S} = v$. Rewriting the above by flipping the order of expectations we get that
  \[
  \E_{S\subseteq_{k/ R^2\log(1/\eps)} [n]}
  \left[\sum\limits_{v\in \{0,1\}^S} \left\langle f_{S,v}, \mathrm{T} f_{S,v}\right\rangle\right]
  \ggg \eps^6.
  \]
  By Claim~\ref{claim:sse_conclusion} we get that
  \[
  \sum\limits_{v\in \{0,1\}^S} \left\langle f_{S,v}, \mathrm{T} f_{S,v}\right\rangle
  \leq \sum\limits_{v\in \{0,1\}^S} \|f_{S,v}\|_{2}^{3} + \alpha^{\Omega(1)}\|f_{S,v}\|_{2}^{2}
  \leq \alpha^{\Omega(1)} + \sqrt{\sum_{v}\|f_{S,v}\|_{2}^4}
  \leq \eps^{9}+ \sqrt{\sum_{v}\|f_{S,v}\|_{2}^4},
  \]
  where we used Cauchy-Schwarz and $\sum_{v}\|f_{S,v}\|_{2}^2 \leq 1$.  We get that $\E_{S\subseteq_{k/ R^2\log(1/\eps)} [n]}
  \left[\sqrt{\sum_{v}\|f_{S,v}\|_{2}^4}\right]
  \ggg \eps^6$, so by Cauchy-Schwarz
  \[
  \E_{S\subseteq_{k/ R^2\log(1/\eps)} [n]}
  \left[\sum_{v}\|f_{S,v}\|_{2}^4\right]
  \ggg \eps^{12}.
  \]
  The last inequality implies that
  \[
  \E_{S\subseteq_{k/ R^2\log(1/\eps)} [n]}
  \left[\E_{x,y}\big[1_{g_x|_S \equiv g_y|_{S}}\big]\right]\ggg \eps^{12},
  \]
  where now the distribution over $x$ and $y$ is independent. Flipping the order of expectations again, we get that
  \[
  \E_{x,y}\left[
  \E_{S\subseteq_{k/ R^2\log(1/\eps)} [n]}
  \big[1_{g_x|_S \equiv g_y|_{S}}\big]\right]\ggg \eps^{12}.
  \]
  By an averaging argument with probability at least $\Omega(\eps^{12})$ over $x,y$ we have
  $\E_{S\subseteq_{k/ R^2\log(1/\eps)} [n]}
  \big[1_{g_x|_S \equiv g_y|_{S}}\big]\ggg \eps^{12}$. The proof is concluded by Claim~\ref{claim:trivial_distance_to_diff}.
\end{proof}

\subsubsection{Concluding Theorem~\ref{thm:johnson-dp}}
We now finish off the proof of Theorem~\ref{thm:johnson-dp}.

By Claim~\ref{claim:globally_close_to_each_other}
we may find $x$ such that $\Delta(g_x,g_y)\lll R^2\log(1/\eps)^2/k$ for at least $\Omega(\eps^{12})$ fraction of the $y$'s.
We fix such $x$ henceforth, and sample $y$ independently of $x$. Identify $y$ back with $A_0, B_0$ and then
sample $B\subseteq [n]\setminus A_0$ of size $(1-\alpha)k$. Then $\Delta(g_x,g_y)\lll R^2\log(1/\eps)^2/k$ with
probability $\Omega(\eps^{12})$, and we condition on that. By Chernoff's inequality we have that
$\Delta_B(g_x, g_y)\lll  R^2\log(1/\eps)^3/k$ except with probability $1-\eps^{100}$.
By Lemma~\ref{lemma:excellent_agree}, with probability at least $\Omega(\eps^2)$ we get that
$\Delta(g_y(B), F[A_0\cup B])\lll R/k$. By a union bound and triangle inequality it follows that with probability
$\Omega(\eps^2)$ we have $\Delta_B(g_x, F[A_0\cup B])\lll  R^2\log(1/\eps)^3/k$. In conclusion,
\[
\Pr_{x, A_0, B}\big[\Delta_B(g_x, F[A_0\cup B])\lll  R^2\log(1/\eps)^3/k\big]\ggg \eps^{12}.
\]
As $A_0, B$ are independent of $x$, we get that
\[
\Pr_{x, A}\big[\Delta_A(g_x, F[A])\lll  R^2\log(1/\eps)^3/k\big]\ggg \eps^{12},
\]
concluding the proof.
\qed
\subsection{Proof of Claim~\ref{claim:eigenvalue_bound}}\label{sec:proof_of_ev_bound_dp_johnson}
Throughout, we denote $\mathcal{U} =  \binom{[n]}{\alpha k, (1-\alpha)k, n-k}$.
  Consider the symmetric group $S_n$, and let $I,J\subseteq [n]$ be disjoint subsets of sizes $\alpha k$ and $(1-\alpha)k$ respectively.
  Let $\rho >0$ to be determined, and let $\mathrm{T}''_{\rho} = e^{-\rho L}$ where $L\colon L_2(S_n)\to L_2(S_n)$ is the Laplacian
  $Lf(\pi) = f(\pi) - \E_{i\neq j\in [n]}\big[f(\pi_{i,j} \pi)\big]$ where $\pi_{i,j}$ is the transposition permutation between
  $i$ and $j$. Let $\mathrm{S}\colon L_2(S_n)\to L_2(S_n)$ be the operator defined as $\mathrm{S}f(\pi) = \E_{\pi':\pi'(I) = \pi(I)}\big[f(\pi')\big]$.
  Define $\mathrm{T}_{\rho}' = \mathrm{S}\circ \mathrm{T}''_{\rho}$. 

  We will now explain how $\mathrm{T}_{\rho}'$ can be thought of as an operator on $L_2(\mathcal{U})$. To do that, we associate with each permutation
  $\pi\in S_n$ a point $x = x(\pi)\in\mathcal{U}$, where $x_{\pi(I)} = 0$, $x_{\pi(J)} = 1$ and $x_{[n]\setminus \pi(I)\cup\pi(J)} = 2$, where $I = \{1,\ldots,\alpha k\}$, $J$ is the next $(1-\alpha)k$ elements. Thus, we may
  define an operator $\mathrm{H} \colon L_2(\mathcal{U}) \to L_2(S_n)$ by $\mathrm{H} f(\pi) = f(x(\pi))$. Thus, we define the operator
  $\mathrm{T}_{\rho} = \mathrm{H}^{*}\mathrm{T}_{\rho}'\mathrm{H}\colon L_2(\mathcal{U})\to L_2(\mathcal{U})$. Our plan is to show the following two claims:
  \begin{claim}\label{claim:move_between operators}
    For $\rho = n\ln(1/\alpha)/2$ we have that $\|\mathrm{T}_{\rho} - \mathrm{T}\|_2 \leq \alpha^{\Omega(1)}$.
  \end{claim}

  \begin{claim}\label{claim:bound_ev_of_other_op}
    For $\rho = n\ln(1/\alpha)/2$, $\lambda_2(\mathrm{T}_{\rho})\leq \alpha^{\Omega(1)}$.
  \end{claim}
  Together, the two claims finish the proof of Claim~\ref{claim:eigenvalue_bound} immediately.

  \subsubsection{Proof of Claim~\ref{claim:move_between operators}}
  Consider a sampling of $x\sim \mathcal{U}$, $y\sim \mathrm{T} x$ and $y'\sim \mathrm{T}_{\rho} x$. For each $\ell\in\mathbb{N}$, let
  $E^{\ell}$ be the event that $|x^{-1}(0) \cap y^{-1}(0)| = \ell$, let $E_{\rho}^{\ell}$ be the event that $|x^{-1}(0) \cap {y'}^{-1}(0)| = \ell$,
  and note that the distributions $y~|~E^{\ell}$ and $y'~|~E_{\rho}^{\ell}$ are identical. Thus, the statistical distance between $y$ and $y'$ is
  at most
  \[
  \Delta := \sum\limits_{\ell}\left|\Pr[E^{\ell}] - \Pr[E_{\rho}^{\ell}]\right|.
  \]
  It follows that there is a coupling between $y$ and $y'$ such that $\Pr[y\neq y']$.
  Fix $f$ with $2$-norm equal to $1$ for which $\|\mathrm{T} - \mathrm{T}_{\rho}\|_2 =\|(\mathrm{T} - \mathrm{T}_{\rho})f\|_2$.
  Then
  \[
  \|\mathrm{T} - \mathrm{T}_{\rho}\|_2^2
  =\E_{x\sim \mathcal{U}}\Big[\left|\E_{y,y'}\big[f(y) - f(y')\big]\right|^2\Big]
  =
  \E_{x\sim \mathcal{U}}\Big[\left|\E_{y,y'}\big[(f(y) - f(y'))1_{y\neq y'}\big]\right|^2\Big],
  \]
  which by Cauchy-Schwarz is at most
  \[
  \E_{x\sim \mathcal{U}}\Big[\Pr_{y,y'}\big[y\neq y'\big]\E_{y,y'}\big[\left|f(y) - f(y')\right|^2\big]\Big]
  \leq \Delta 4\|f\|_2^4
  =4\Delta.
  \]
  The following claim finishes the proof of Claim~\ref{claim:move_between operators}.
  \begin{claim}
    $\Delta\leq \alpha^{\Omega(1)}$.
  \end{claim}
  \begin{proof}
    Consider the distribution of $|x^{-1}(0) \cap {y'}^{-1}(0)|$. An equivalent way to think about the Poisson sampling,
    is that for each $i\neq j$ we have an independent poisson random variable $Z_{i,j}\sim {\sf Poisson}(\rho/n(n-1))$, and
    then we apply the transpositions corresponding to $Z_{i,j}$. Define $Z_i = 1_{\sum\limits_{j} Z_{i,j}>0}$, 
    $Z = \sum\limits_{i\in I} Z_i$, and note that with probability $1-\alpha^{\Omega(1)}$ it holds that 
    $|x^{-1}(0)\cap {y'}^{-1}(0)| = \alpha k - Z$; the reason is that $Z$ counts the number of $i\in I$ such that 
    we picked a transposition of the form $\pi_{i,j_i}$ in the process for some $j_i$, and with probability $1-\alpha^{\Omega(1)}$
    all of these $j_i$'s are outside $I$ and are all distinct. Thus, the distribution $q_{\ell} = \Pr[E^{\ell}]$ is $\alpha^{\Omega(1)}$
    close in statistical distance to the distribution $p_{\ell} = \Pr\big[Z = \alpha k-\ell\big]=\Pr\big[\sum\limits_{i\in I}1-Z_i = \ell\big]$. 
    Note $1-Z_i$ are independent Bernouli random variables with 
    \[
    \E[1-Z_i] 
    = \prod\limits_{j\neq i}\Pr\big[Z_{i,j} = 0\big]
    = e^{-\rho/n} 
    = \alpha^{1/2},
    \]
    so $p_{\ell}$ is exactly the law of $|x^{-1}(0)\cap y^{-1}(0)|$.
  \end{proof}

  \subsubsection{Proof of Claim~\ref{claim:bound_ev_of_other_op}}
  Fix $f\colon \mathcal{U}\to\mathbb{R}$ with $2$-norm equal to $1$ and expectation $0$ to be an eigenvector
  of $\mathrm{T}_{\rho}$ with eigenvalue $\lambda_2(\mathrm{T}_{\rho})$, and define $g = \mathrm{H} f$. Note that $\E[g] = 0$ and
  \[
  \lambda_2(\mathrm{T}_{\rho})
  =\|\mathrm{H}^{*}\mathrm{T}_{\rho}'\mathrm{H} f\|_2
  \leq
  \|\mathrm{T}_{\rho}' g\|_2.
  \]
  Thus, it suffices to upper bound the second eigenvalue of $\mathrm{T}_{\rho}'$.
  To study these eigenvalues, we use formulas (31), (32) and (33) from~\cite[Section 4.2]{filmus2022log}. These 
  formulas assert that for each non-trivial partition $\lambda\vdash n$, the corresponding eigenvalue is
  $e^{-\rho(1-c_{\lambda})}$ where
  \[
  c_{\lambda} = \frac{1}{n(n-1)}\sum\limits_{i}^{} \lambda_i^2 - (2i-1)\lambda_i
  \leq \frac{1}{n(n-1)} (n-1)^2
  \leq \frac{n-1}{n}.
  \]
  Thus, for the eigenvalue we have that
  $e^{-\rho(1-c_{\lambda})}\leq e^{-\rho\frac{1}{n}}\leq e^{-\ln(1/\alpha)/2} = \sqrt{\alpha}$.

\end{document}